\documentclass[11pt,a4paper,USenglish]{article}
\usepackage[in]{fullpage}

\usepackage{microtype}
\usepackage{boxedminipage}

\usepackage{aurical}
\usepackage[T1]{fontenc}
\usepackage{xspace}
\usepackage{esvect}
\usepackage{amsmath,amsfonts,amssymb,mathtools}
\usepackage[c2]{optidef}
\usepackage{datetime}
\usepackage{stmaryrd}
\usepackage{tcolorbox}
\usepackage{wasysym, fancybox}
\usepackage{dsfont}
\usepackage{eqparbox}
\usepackage{comment}
\usepackage{enumitem,amsthm,hyperref}

\usepackage[normalem]{ulem} 

\usepackage{tikz-cd}

\graphicspath{{./Figures/}}

\usepackage{graphicx}%
\usepackage{multirow}%
\usepackage{mathrsfs}%
\usepackage{transparent}

\usepackage[title]{appendix}%
\usepackage{xcolor} 

\usepackage[normalem]{ulem}

\usepackage{textcomp}%
\usepackage{manyfoot}%
\usepackage{booktabs}%
\usepackage{algorithm}%
\usepackage{algorithmicx}%
\usepackage{algpseudocode}%
\usepackage{listings}%

\usepackage{caption}
\usepackage{subcaption}

\graphicspath{{Figures/}}

\definecolor{lightblue}{rgb}{.80,.85,1}
\definecolor{darkgreen}{RGB}{0,100,0}
\definecolor{firebrick}{RGB}{178,34,34}
\definecolor{salmon}{RGB}{250,128,114}
\definecolor{turquoise}{RGB}{0,128,114}
\definecolor{turquoise2}{RGB}{0,180,140}
\definecolor{turquoise3}{RGB}{0,160,220}
\definecolor{bordeaux}{RGB}{144, 12, 63}
\definecolor{darkorchid}{rgb}{0.60,0.20,0.80}
\definecolor{lightorange}{RGB}{234,219,173}
\definecolor{bianca}{RGB}{0,92,153}
\definecolor{lila}{RGB}{193,128,255}
\definecolor{mistyrose}{RGB}{255,228,225}

\newtheorem{theorem}{Theorem}
\newtheorem{lemma}[theorem]{Lemma}

\newtheorem{remark}{Remark}

\newtheorem{definition}{Definition}

\newcommand{\ignore}[1]{}

\newcommand{\Remember}[1]{}

\newcommand{\Modif}[1]{#1}

\newcommand{\defunder}[1]{\underset{\text{def.}}{#1} \:}
\newcommand{\UnitMatrix}{\mathds{1}}
\newcommand{\Skel}{\operatorname{Sk}}

\newcommand{\Step}[1]{\medskip\noindent{\underline{\bf Step #1:}}\hspace{1mm}}

\newcommand{\Vertexset}[1]  {\operatorname{Vert}{#1}}

\newcommand{\Shadow}[1]{\boldsymbol{\lvert} #1 \boldsymbol{\rvert}}
\newcommand{\relint}[1]{\operatorname{relint}(#1)}

\newcommand{\wdel}{\ensuremath{w}\xspace}

\newcommand{\Rips}[2]{\mathcal{R}(#1,#2)}
\newcommand{\Cech}[2]{\mathcal{C}(#1,#2)}
\newcommand{\Del}[1]{\operatorname{Del}(#1)}
\newcommand{\Delloc}[3]{\operatorname{Delloc}_{#3}(#1,#2)}

\newcommand{\Star}[2]{\operatorname{St}(#1,#2)}

\newcommand{\Rspace}{\ensuremath{\mathbb{R}}\xspace}
\newcommand{\paraboloid}{\ensuremath{\mathscr{P}}\xspace}
\newcommand{\Dim}{N}
\newcommand{\Zspace}{\ensuremath{\mathbb{Z}}\xspace}
\newcommand{\Sspace}{\ensuremath{\mathbb{S}}\xspace}
\newcommand{\M}{\ensuremath{\mathcal{M}}\xspace}
\newcommand{\C}{\ensuremath{\mathcal{C}}\xspace}

\newcommand{\D}{\ensuremath{\mathcal{D}}\xspace}
\newcommand{\K}{\ensuremath{\mathcal{K}}\xspace}

\newcommand{\B}{\mathbb{B}}

\newcommand{\AngularDeviation}[1]{\operatorname{angularDeviation}_{\M}(#1)}
\newcommand{\code}[1]{\operatorname{code}_{#1}}
\newcommand{\BaryCoord}[3]{\operatorname{BarycentricCoord}_{#3}^{#2}(#1)}

\newcommand{\s}[1]{\sffamily \scriptsize $#1$}
\newcommand{\styleitem}[1]{\textcolor{darkgray}{\sffamily\bfseries\upshape\mathversion{bold}#1}\xspace}

\newcommand{\card}[1]{\operatorname{card}#1}

\DeclarePairedDelimiter\abs{\lvert}{\rvert}

\newcommand \Above[2]{\genfrac{}{}{0pt}{0}{#1}{#2}}

\newcommand{\sign}[2]{\operatorname{sign}_{#2}(#1)}

\newcommand{\Support}[1]{\operatorname{Supp} #1}
\newcommand{\Indicator}[1]{\mathbf{1}_{#1}}

\newcommand{\DotProd}[2]{#1 \cdot #2}
\newcommand{\Power}[2]{\operatorname{Power}_{#1}(#2)}
\newcommand{\Id}{\operatorname{Id}}

\newcommand{\Edel}{\operatorname{E_{\operatorname{del}}}}
\newcommand{\load}{\operatorname{load}}
\newcommand{\flux}{\operatorname{flux}}
\newcommand{\Load}[4]{\load_{#1,#2}(#4)}
\newcommand{\LoadLoc}[4]{\load_{#1,#2,#3}(#4)}

\newcommand{\Comb}[3]{\lambda_{#2}^{#1}(#3)}

\newcommand{\Differential}{\mathrm{D}}
\newcommand{\NorSpace}{\boldsymbol{N}}
\newcommand{\TanSpace}{\boldsymbol{T}}

\newcommand{\Conv}[1]{\operatorname{conv}#1}
\newcommand{\Aff}[1]{\operatorname{aff}#1}
\newcommand{\Volume}[1]{\operatorname{vol}(#1)}

\newcommand{\Reach}[1]{\operatorname{reach}#1}
\newcommand{\reach}{\ensuremath{\mathcal{R}}\xspace}
\newcommand{\Offset}[2]{#1^{\oplus #2}}
\newcommand{\Diam}[1]{\operatorname{Diam}(#1)}

\newcommand{\Sep}[1]{\operatorname{separation}(#1)}

\newcommand{\height}[1]{\operatorname{height}(#1)}
\newcommand{\Height}[2]{\operatorname{height}(#1,#2)}
\newcommand{\ProtectGlob}[2]{\operatorname{protection}(#1,#2)}
\newcommand{\protection}[2]{\operatorname{protection}(#1,#2)}
\newcommand{\pro}{\mathsf{p}}
\newcommand{\sepvar}{\mathsf{s}}
\newcommand{\hei}{\Height{P}{\rho}}
\newcommand{\MA}[1]{\operatorname{axis}(#1)}

\newcommand{\Tangent}[2]{\mathbf{T}_{#1}#2}

\def\restrict#1{\raise-.5ex\hbox{\ensuremath|}_{#1}}

\newcommand\restr[2]{{
  \left.\kern-\nulldelimiterspace 
  #1 
  \vphantom{\big|} 
  \right|_{#2} 
}}

\newcommand{\ts}{\ensuremath{\tilde{\sigma}}}

\newcommand{\RestConv}[2]{\Conv_{\mid #2}{#1}}

\newcommand{\length}{\operatorname{length}}

\title{Delaunay-like Triangulation of Smooth Orientable Submanifolds by $\ell_1$-Norm
  Minimization}

\author{Dominique Attali\footnote{Univ. Grenoble Alpes, CNRS,
  GIPSA-lab, Grenoble, France. \texttt{Dominique.Attali@grenoble-inp.fr}}
\and
Andr\'e Lieutier\footnote{Aix-en-Provence, France. \texttt{andre.lieutier@gmail.com}}
}

\begin{document}

\maketitle

\begin{abstract}
  In this paper, we study the shape reconstruction problem, when the
  shape we wish to reconstruct is an orientable smooth $d$-dimensional
  submanifold of the Euclidean space. Assuming we have as input a
  simplicial complex $K$ that approximates the submanifold (such as
  the \v Cech complex or the Rips complex), we recast the problem of
  reconstucting the submanifold from $K$ as an $\ell_1$-norm
  minimization problem in which the optimization variable is a
  $d$-chain of $K$ over the field $\Rspace$. Providing that $K$
  satisfies certain reasonable conditions, we prove that the
  considered minimization problem has a unique solution which
  triangulates the submanifold and coincides with the flat Delaunay
  complex introduced and studied in a companion paper
  \cite{AttaliLieutierFlatDelaunay2022}. Since the objective is a
  weighted $\ell_1$-norm and the constraints are linear, the
  triangulation process can thus be implemented by linear programming.    
\end{abstract}

%
%
%
%
%

\section{Introduction}

In many practical situations, the shape of interest is only known
through a finite set of sample points. Given as input these points, a
natural question is how to construct a \emph{triangulation} of the
shape, that is, a set of simplices whose union is homeomorphic to the
shape. This problem, known as {\em shape reconstruction}, has been
widely studied \cite{hoppe1992surface,
  amenta1999surface,Amenta_Choi_Kolluri:2001:powerCrust,edels:2003:wrap,giesen2002surface,cheng2005manifold,Dey:2006:reconstruction_book,bauer2024wrapping}. In
the paper, the shape is assumed to be a smooth orientable
$d$-dimensional submanifold of the Euclidean space.  We show that,
under appropriate conditions, a triangulation of that submanifold can
be expressed as the solution of a weighted $\ell_1$-norm minimization
problem under linear constraints.

\paragraph{Overview of the method in the particular case of planar curves.}

We first give an informal description of
our variational formulation in the easy case of the triangulation of a
closed, connected smooth curve in the plane.
Assume that we are given a set of points $P$ that sample an unknown
curve $\C$ (example in Figure \ref{fig:PlanarCurve_1}). Consider the
graph $K$ whose vertices are the points of $P$ and whose edges connect
pair of points that are within a certain given distance; see Figure
\ref{fig:PlanarCurve_2}. Our goal is to compute a triangulation of the
curve $\C$ ({\em i.e.} a closed polygonal line) whose vertices and
edges are in $K$ and which follows ``nicely'' $\C$.

\begin{figure}[htb]
	\centering
	\begin{subfigure}[b]{0.49\textwidth}
		\centering
\begingroup%
  \makeatletter%
  \providecommand\color[2][]{%
    \errmessage{(Inkscape) Color is used for the text in Inkscape, but the package 'color.sty' is not loaded}%
    \renewcommand\color[2][]{}%
  }%
  \providecommand\transparent[1]{%
    \errmessage{(Inkscape) Transparency is used (non-zero) for the text in Inkscape, but the package 'transparent.sty' is not loaded}%
    \renewcommand\transparent[1]{}%
  }%
  \providecommand\rotatebox[2]{#2}%
  \newcommand*\fsize{\dimexpr\f@size pt\relax}%
  \newcommand*\lineheight[1]{\fontsize{\fsize}{#1\fsize}\selectfont}%
  \ifx\svgwidth\undefined%
    \setlength{\unitlength}{180.20492313bp}%
    \ifx\svgscale\undefined%
      \relax%
    \else%
      \setlength{\unitlength}{\unitlength * \real{\svgscale}}%
    \fi%
  \else%
    \setlength{\unitlength}{\svgwidth}%
  \fi%
  \global\let\svgwidth\undefined%
  \global\let\svgscale\undefined%
  \makeatother%
  \begin{picture}(1,0.72588113)%
    \lineheight{1}%
    \setlength\tabcolsep{0pt}%
    \put(0.50848009,0.66717425){\color[rgb]{0.39215686,0.58431373,0.92941176}\makebox(0,0)[lt]{\lineheight{1.25}\smash{\begin{tabular}[t]{l}$\C$\end{tabular}}}}%
    \put(0,0){\includegraphics[width=\unitlength,page=1]{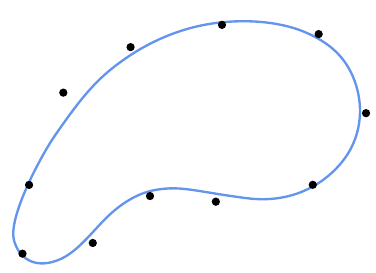}}%
  \end{picture}%
\endgroup%

		\subcaption{}
		\label{fig:PlanarCurve_1}
	\end{subfigure}
	\hfill
	\begin{subfigure}[b]{0.49\textwidth}
		\centering 
\begingroup%
  \makeatletter%
  \providecommand\color[2][]{%
    \errmessage{(Inkscape) Color is used for the text in Inkscape, but the package 'color.sty' is not loaded}%
    \renewcommand\color[2][]{}%
  }%
  \providecommand\transparent[1]{%
    \errmessage{(Inkscape) Transparency is used (non-zero) for the text in Inkscape, but the package 'transparent.sty' is not loaded}%
    \renewcommand\transparent[1]{}%
  }%
  \providecommand\rotatebox[2]{#2}%
  \newcommand*\fsize{\dimexpr\f@size pt\relax}%
  \newcommand*\lineheight[1]{\fontsize{\fsize}{#1\fsize}\selectfont}%
  \ifx\svgwidth\undefined%
    \setlength{\unitlength}{180.20492313bp}%
    \ifx\svgscale\undefined%
      \relax%
    \else%
      \setlength{\unitlength}{\unitlength * \real{\svgscale}}%
    \fi%
  \else%
    \setlength{\unitlength}{\svgwidth}%
  \fi%
  \global\let\svgwidth\undefined%
  \global\let\svgscale\undefined%
  \makeatother%
  \begin{picture}(1,0.72588113)%
    \lineheight{1}%
    \setlength\tabcolsep{0pt}%
    \put(0,0){\includegraphics[width=\unitlength,page=1]{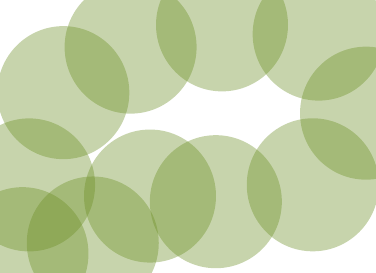}}%
    \put(0.88205378,0.15665707){\color[rgb]{0.00392157,0.43529412,0.24313725}\makebox(0,0)[lt]{\lineheight{1.25}\smash{\begin{tabular}[t]{l}$r$\end{tabular}}}}%
    \put(0,0){\includegraphics[width=\unitlength,page=2]{PlanarCurve_2.pdf}}%
    \put(0.0450506,0.30476805){\makebox(0,0)[lt]{\lineheight{1.25}\smash{\begin{tabular}[t]{l}$K$\end{tabular}}}}%
    \put(0,0){\includegraphics[width=\unitlength,page=3]{PlanarCurve_2.pdf}}%
    \put(0.11341141,0.48940131){\color[rgb]{0,0,0}\makebox(0,0)[lt]{\lineheight{1.25}\smash{\begin{tabular}[t]{l}$P$\end{tabular}}}}%
  \end{picture}%
\endgroup%

		\subcaption{}
		\label{fig:PlanarCurve_2}
	\end{subfigure}
	\vskip\baselineskip
	\begin{subfigure}[b]{0.49\textwidth} 
	  \centering
      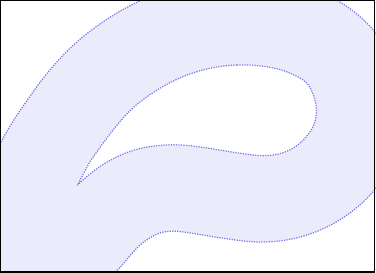
	  \subcaption{}
	  \label{fig:PlanarCurve_3}
	\end{subfigure}
	\hfill
	\begin{subfigure}[b]{0.49\textwidth}   
	  \centering
      {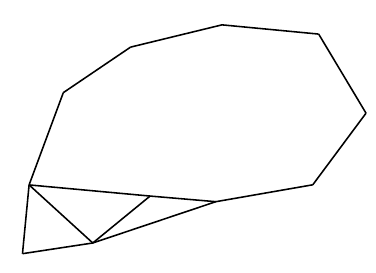}
	  \subcaption{}
      \label{fig:PlanarCurve_4}
	\end{subfigure}
	\caption{
      (a) A finite set of (black) points $P$ that sample a (blue)
      curve $\C$. (b) The graph $K$ is a proximity graph constructed
      from $P$ by connecting every pair of points that are less than
      $2r$ apart. Each edge in $K$ is arbitrarily oriented. (c) The
      segment $[a,b]$ intersects both $\C$ and $K$ transversally with
      $b$ in the inside region and $a$ in the outside region. (d) A
      1-cycle of $K$ whose flux through $[a,b]$ equals 2.
      \label{fig:PlanarCurve_first}}
\end{figure}

Let us orient (arbitrarily) edges in $K$; see Figure
\ref{fig:PlanarCurve_2}.  A $1$-chain $\gamma$ on $K$ with real
coefficients is the assignation of a real number $\gamma(e)$ to each
oriented edge $e$ in $K$. A $1$-cycle is a $1$-chain $\gamma$ such
that $\partial \gamma=0$, which means that, at each vertex $v$ of $K$,
the sum of the coefficients of the edges entering $v$ equals the sum
of coefficients of edges leaving $v$; see Figure
\ref{fig:PlanarCurve_4}.

We now define informally what we mean by a normalized $1$-cycle.
Consider a tubular neighborhood of $\C$ sufficiently small, so that
its complement consists only of two connected components, one bounded
component called the inside region and one unbounded component called
the outside region. Suppose furthermore that $K$ is contained in this
small tubular neighborhood; see Figure \ref{fig:PlanarCurve_3}.  The segment $[a,b]$ shown on Figure
\ref{fig:PlanarCurve_3} has one endpoint $a$ in the outside region,
one endpoint $b$ in the inside region and it cuts both the curve
$\C$ and the graph $K$ transversally. Let us say that an oriented edge
$[p,q]$ of $K$ intersects $[a,b]$ in a {\em forward direction} ({\em
  resp.}  {\em backward direction}) if $a$ lies on the left ({\em
  resp.} {\em on the right}) of the directed line through $p$ and
$q$. Given a $1$-chain $\gamma$, we then define the {\em flux} through
$[a,b]$ of $\gamma$ as
\[
\flux_{[a,b]}(\gamma) = \sum_{e^+} \gamma(e^+) - \sum_{e^-} \gamma(e^-),
\]
where the first sum is over all edges $e^+$ of $K$ that cross $[a,b]$
in a forward direction and the second sum is over all edges $e^-$ of
$K$ that cross $[a,b]$ in a backward direction.  For example, the flux
of $\gamma$ through $[a,b]$ is $2$ on Figure
\ref{fig:PlanarCurve_4}. Note that the flux through $[a,b]$ is a
linear form on the vector space of $1$-cycles. Moreover, the flux of a
$1$-cycle $\gamma$ through $[a,b]$ does not depend upon the location
of $[a,b]$, as long as $a$ remains in the outside region and $b$
remains in the inside region. Indeed, the expression of the flux
changes only when the edge $[a,b]$ passes through a vertex of $K$, at
which time one can check that the value of the flux remains constant
if $\gamma$ is a 1-cycle.  We say that a 1-cycle is {\em normalized}
if its flux through $[a,b]$ is equal to
$1$. Figure~\ref{fig:PlanarCurve_5} depicts such a normalized cycle.

\begin{figure}[htb]
	\centering
	\begin{subfigure}[b]{0.49\textwidth} 
	  \centering
      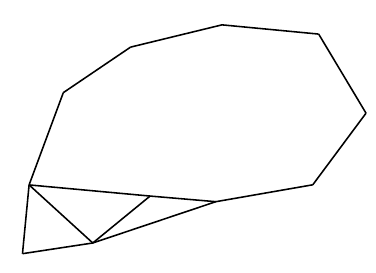
		\subcaption{}
		\label{fig:PlanarCurve_5}
	\end{subfigure}
	\hfill
	\begin{subfigure}[b]{0.49\textwidth}   
	  \centering
      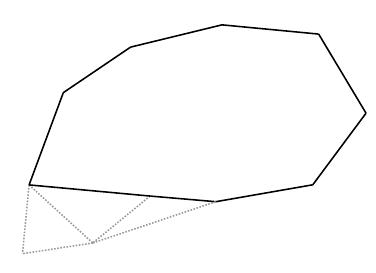
	  \subcaption{}
	  \label{fig:PlanarCurve_6}
	\end{subfigure}
	\vskip\baselineskip
	\begin{subfigure}[b]{0.49\textwidth} 
	  \centering
      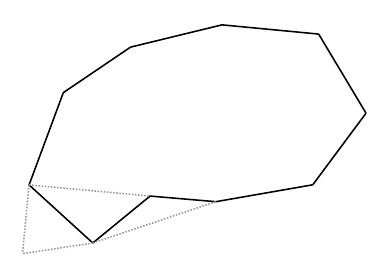
	  \subcaption{}
	  \label{fig:PlanarCurve_7}
	\end{subfigure}
	\hfill
	\begin{subfigure}[b]{0.49\textwidth}   
	  \centering
      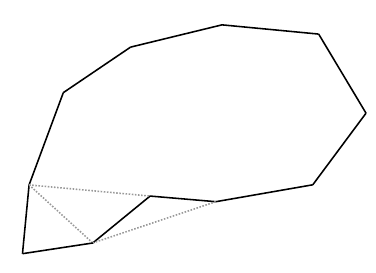
	  \subcaption{}
	  \label{fig:PlanarCurve_8}
	\end{subfigure}
	\caption{ (a) A 1-cycle is said to be normalized if its flux
      through $[a,b]$ is equal to 1. (b)~A normalized cycle that
      minimizes the length. (c)~A normalized cycle that minimizes the
      sum of the edge lengths squared. (d)~A normalized cycle that
      minimizes the sum of the edge lengths cubed, also called the
      Delaunay energy (up to a multiplicative constant). Our
      reconstruction method returns the support of that cycle.}
	\label{fig:PlanarCurve_second}
\end{figure}

Our reconstruction method consists in computing the minimal cycle
among all normalized cycles.  By minimal cycle, we mean here a cycle that
minimizes a weighted $\ell^1$-norm. Perhaps the most natural weighted
$\ell^1$-norm for geometric $1$-cycles could be the length:
\[
\length(\gamma) = \sum_e \length(e) \, | \gamma(e) |,
\]
where the sum is over all edges $e$ of $K$.  The normalized cycle
minimizing the length is depicted in Figure~\ref{fig:PlanarCurve_6},
where one can observe that it is indeed a triangulation of the curve
$\C$. However, in order to minimize the length, the resulting minimal
cycle favors long edges and skips intermediate sample points whenever
possible, so that small features of the curve are ignored.

In contrast, thanks to Pythagorean Theorem, weighting edges by the
square of the length would make the minimum cycle follow intermediate
points, as long as there exist at those intermediate points an
incoming edge and an outcoming edge forming an angle larger than
$\frac{\pi}{2}$; see Figure \ref{fig:PlanarCurve_7}.

Instead, we consider minimizing the Delaunay energy,  which, in the particular case of one dimensional simplices ({\em i.e.} edges),  consists 
in weighting edges (up to a constant factor)  by the cube of the length:
 \[
\Edel(\gamma) = \frac{1}{6} \sum_e  \length(e)^3 \, | \gamma(e) |
\]
The cycle minimizing this Delaunay energy is depicted on Figure
\ref{fig:PlanarCurve_8}. One can again check (and we shall actually
prove in the paper) that this is indeed a triangulation of $\C$. Our
method for reconstructing one dimensional curves can then be expressed
as solving the following optimization problem over $1$-chains of $K$:
\begin{tcolorbox}
  \vspace{-4mm}
  \begin{mini*}
    {\gamma}{\Edel(\gamma)}{}{} \label{problem:reconstruction1D}
    \addConstraint{\partial \gamma}{=0}
    \addConstraint{\flux_{[a,b]} (\gamma)}{=1}
  \end{mini*}
\end{tcolorbox}


We aim at generalizing the above problem beyond dimension one
manifolds and at identifying sufficient conditions for the support of
its solution to be a triangulation. But for that we need first to make
a detour to Delaunay complexes.

\paragraph{Variational formulation for Delaunay complexes.}
The starting point of the work presented in this paper is the
observation that when we consider a point cloud $P$ in $\Rspace^d$,
its Delaunay complex can be expressed as the solution of a particular
$\ell_p$-norm minimization problem. This fact is best explained by
lifting the point set $P$ vertically onto the paraboloid $\paraboloid
\subseteq \Rspace^{d+1}$ whose equation is $x_{d+1} = \sum_{i=1}^d
x_i^2$. Denoting by $\hat{P}$ the lifted points, it is well-known that
the Delaunay complex of $P$ is isomorphic to the boundary complex of
the lower convex hull of $\hat{P}$.

Starting from this equivalence, Chen has observed in
\cite{chen2004optimal} that the Delaunay complex of $P$ minimizes over
all triangulations $T$ of $P$ the $\ell_p$-norm of the difference
between the following two functions: the first function maps each
point $x$
to its vertical projection onto the lifted triangulation $\hat{T}$ and
the second function maps each point $x$
to the lifted point $\hat{x} \in
\paraboloid$. This variational formulation has been successfully
exploited in
\cite{alliez2005variational,chen2011efficient,chen2014revisiting} for
the generation of {\it Optimal Delaunay Triangulations}. When $p=1$,
the $\ell_p$-norm associated to $T$ is what we call in this paper the
Delaunay energy of $T$ and, can be interpreted as the $(d+1)$-volume
enclosed between the lifted triangulation $\hat{T}$ and the
paraboloid~\paraboloid. Given a $d$-simplex $\sigma$, we call the
$(d+1)$-volume enclosed between the convex hull of $\hat{\sigma}$ and
$\paraboloid$ the {\em Delaunay weight} of $\sigma$ and denote it as
$\omega(\sigma)$. The Delaunay energy of $T$ can then be simply
expressed as the sum of the Delaunay weights of its $d$-simplices.

\paragraph{Contributions.}

We present a variational formulation to submanifold reconstruction in
Euclidean space, that both extends our variational approach for curve
reconstruction in the plane and the variational approach for
generating Delaunay complexes in $\Rspace^d$. Consider a finite set of
points $P$ that sample an unknown $d$-dimensional submanifold $\M
\subseteq \Rspace^\Dim$ and suppose that we have at hand a simplicial
complex $K$ whose vertex set is $P$ (such as for instance the \v Cech
complex of $P$ or the Vietoris-Rips complex of $P$). Given as input
$K$, our goal is to find a triangulation of $\M$ contained in $K$.

A crucial ingredient in our variational formulation is to embed the
triangulations contained in $K$ inside the vector space formed by
simplicial $d$-cycles\footnote{Or relative $d$-cycles when the
considered domain has a boundary.} of $K$ over the field $\Rspace$. In
spirit, this is similar to what is done in the theory of minimal
surfaces, when oriented surfaces are identified to particular elements
of a much larger set, namely the space of currents
\cite{morgan2016geometric}, which enjoys the nice property of being a
vector space. Our method for reconstructing submanifolds consists in
solving the following optimization problem over $d$-chains of $K$:
\begin{tcolorbox}
  \vspace{-4mm}
  \begin{mini*}
    {\gamma}{\Edel(\gamma)}{}{} 
    \addConstraint{\partial \gamma}{=0} 
    \addConstraint{\load_A (\gamma)}{=1}
  \end{mini*}
\end{tcolorbox}

The objective function is the Delaunay energy, whose definition
we adapt to the $d$-chains $\gamma$ of $K$ by setting
\[
\Edel(\gamma) = \sum_{\sigma} \omega(\sigma) \, | \gamma(\sigma) |,
\]
where the sum is over all $d$-simplices $\sigma$ of $K$ and
$\omega(\sigma)$ is the Delaunay weight of $\sigma$.  The first
constraint expresses the fact that we are searching for
$d$-cycles. The second constraint is a linear equation which may be
interpreted as a kind of normalization of $\gamma$ that extends the
condition $\flux_{[a,b]}(\gamma)=1$ beyond dimension one. The letter
$A$ designates a set of parameters that specifies where the load is
computed.
Thus, our optimization problem is an $\ell_1$-norm
minimization problem under linear constraints. As such, it can be
turned into a linear optimization problem in the standard form through
slack variables as explained in
Appendix~\ref{appendix:linear-programming}, and can be addressed by
standard linear programming techniques such as the simplex algorithm.

One important point is that the objective function ({\em i.e.} the
Delaunay energy) is a weighted $\ell_1$-norm. Put it simply, we are
searching for weighted $\ell_1$-minima.  The celebrated sparsity of
$\ell_1$-minima manifests itself in our context by the fact that the
support of such minima is sparse, in other words it is non-zero only
on a small subset of simplices of $K$. Our main result is a set of
conditions under which our optimization problem has a unique solution
whose support is a triangulation of $\M$ (either with theoretical or
practical normalization).

\paragraph{Proof technique.}

The proof requires us to introduce an elaborate construction, the {\em
  Delloc complex} of $P$, as a tool to describe the solution. The
$d$-simplices of that complex possess exactly the property that we
need for our analysis. In a companion paper
\cite{AttaliLieutierFlatDelaunay2022} we show that the Delloc complex
indeed provides a triangulation of the manifold, assuming the set of
sample points $P$ to be sufficiently dense, safe, and not too
noisy. Incidently, the Delloc complex coincides with the {\em flat
  Delaunay complex} introduced in our companion paper
\cite{AttaliLieutierFlatDelaunay2022} and is akin to the {\em
  tangential Delaunay complex} introduced and studied in
\cite{boissonnat2014manifold,boissonnat2018geometric}.  When the
manifold is sufficiently densely sampled by the data points, all three
constructions are locally isomorphic to a (weighted) Delaunay
triangulation computed in a local tangent space to the
manifold. Intuitively, this indicates that the Delaunay energy should
locally reach a minimum for all three constructions and, therefore
ought to be also a global minimum.  Actually, turning this intuitive
reasoning into a correct proof turns out to be more tricky than it
appears and is the main purpose of the present paper.  In particular,
we need to globally compare the Delaunay energy of the cycle carrying
the Delloc complex with that of alternate $d$-cycles, and this
requires us to carefully distribute the Delaunay energy along
barycentric coordinates. 

For the purpose of the proof, it is convenient to first consider a
rather artificial problem (Problem (\ref{problem:reconstruction}))
where, besides the sample $P$, the manifold $\M$ is known. At the end
of the paper, we show how to turn this problem into a more realistic
one (Problem (\ref{problem:practical})) that takes as input only the
sample of the unknown manifold, and is correct assuming that
reasonable sampling conditions are met.

An overview of the proof is provided in Section
\ref{section:proof-overview}.

\paragraph{Related work.}
A closely related problem is the computation of $\ell_1$-minimum
homology representative cycles. Several authors, with
computational topology or topological data analysis motivations, have
considered the computation of such cycles, generally for integers or
integers modulo $p$ coefficients
\cite{chen_HardnessResultsHomology_2011,
  borradaile_MinimumBoundedChains_2020,dey_OptimalHomologousCycles_2011,chambers_MinimumCutsShortest_2009,
  dey_ComputingMinimalPersistent_2020}.

A combinatorial counterpart of Delaunay energy, called {\em
  lexicographic order}, has been considered, using the field
$2\Zspace/\Zspace$ of integer modulo $2$ instead of $\Rspace$ as chain
coefficients.  Given a point cloud $P$ in Euclidean space, the
lexicographic minimal chain, among chains with vertices in $P$ and
whose boundary support is the boundary of the convex hull of $P$ has
the Delaunay triangulation as its support \cite{cohen2023delaunay}.
In \cite{cohen2022lexicographic}, the authors consider lexicographic
minimal chains for practical applications to surface reconstruction
in $\Rspace^3$.

\paragraph{Outline.}
Section~\ref{section:preliminaries} introduces the necessary
terminology.  Section~\ref{section:Delaunay-complex} reviews Delaunay
complexes and characterizes them as the triangulations with smallest
Delaunay energy. Section~\ref{section:Delaunay-weight} defines
Delaunay weights and expresses the Delaunay energy as a sum of
Delaunay weights. Section~\ref{section:variational-formulation}
tackles the problem of reconstructing a submanifold $\M$ from a
simplicial complex $K$ whose vertices sample $\M$. The section
presents a convex optimization problem on the $d$-chains of $K$ whose
objective function is the Delaunay energy and whose constraints are
linear. We then state our main result, which are conditions under
which a solution to that optimization problem provides a triangulation
of $\M$. Section \ref{section:proof-main-result} is dedicated to
proving our main result.  Section~\ref{section:PracticalAlgorithm}
discusses practical aspects.

%
%
%
%
%

\section{Preliminaries}
\label{section:preliminaries}

In this section, we review the necessary background and explain some of our terms.

\subsection{Subsets and submanifolds}

\Modif{Given a set $A \subseteq \Rspace^N$ and a point $x \in \Rspace^\Dim$,
we say that $x$ is an \emph{affine combination} of points in $A$ if we
can find  $a_1, \ldots, a_p$ in $A$ and real numbers $\lambda_1, \ldots, \lambda_p$ summing up
to 1 such that $x = \sum_{i=1}^p \lambda_i a_i$. The set of all affine
combinations of points in $A$ is called the \emph{affine space} spanned
by $A$ and is denoted as $\Aff A$. Similarly, we say that $x$ is a
\emph{convex combination} of points in $A$ if we can find $a_1, \ldots, a_p$ in $A$ and 
$\lambda_1, \ldots, \lambda_p$ such that $x = \sum_{i=1}^p \lambda_i a_i$,
where $\sum_{i=1}^p \lambda_i = 1$ and $\lambda_i \geq 0$ for all $1 \leq i \leq p$. The set
of all convex combinations of points in $A$ form the \emph{convex hull}
of $A$ and is denoted as $\Conv A$.}
The \emph{relative interior} of $A$, denoted as
  $\relint{A}$, represents the interior of $A$ within $\Aff{A}$.
For any $x \in \Rspace^\Dim$ and any $r
\in \Rspace$, we denote the closed ball with center $x$ and radius
$r$ by $B(x,r)$. We shall say that $A
\subseteq \Rspace^\Dim$ is \emph{$r$-small} if it can be enclosed
in a ball of radius $r$.
The {\em $r$-tubular neighborhood} of $A$ is the set of
points \Modif{$\Offset A r = \bigcup_{a \in A} B(a,r)$}.
The {\em medial axis} of $A$, denoted as $\MA{A}$, is the set of
points in $\Rspace^\Dim$ that have at least two closest points in
$A$. The {\em reach} of $A$ is the infimum of distances between $A$
and its medial axis, and is denoted as $\Reach A$.  Furthermore, we
define the \emph{projection map} $\pi_A : \Rspace^\Dim \setminus \MA{A}
\to A$, which associates to each point $x$ its unique closest point in
$A$. This projection map is well-defined on every subset of
$\Rspace^\Dim$ that does not intersect the medial axis of $A$. In
particular, it is well-defined on every $r$-tubular neighborhood of
$A$ with $r < \Reach A$.  Recall that the {\em angle} between two
vector spaces $V_0$ and $V_1$ is defined as $\angle(V_0,V_1) =
\max_{v_0 \in V_0} \min_{v_1 \in V_1} \angle v_0,v_1$.  The definition
is not symmetric in $V_0$ and $V_1$, unless the two vector spaces
$V_0$ and $V_1$ share the same dimension. The angle between two affine
spaces $A_0$ and $A_1$ whose corresponding vector spaces are $V_0$ and
$V_1$ is $\angle(A_0,A_1) = \angle(V_0,V_1)$
\cite{BSMF_1875__3__103_2}.

\subsection{Simplicial complexes}

In this section, we review some background notation on algebraic
topology and refer the reader to \cite{munkres1993elements} for a
detailed introduction to the topic.

All simplices and simplicial complexes that we consider in the paper
are abstract.  We recall that an {\em abstract simplicial complex} is
a collection $K$ of finite non-empty sets with the property that if
$\sigma$ belongs to $K$, so does every non-empty subset of
$\sigma$. Each element $\sigma$ of $K$ is called an {\em abstract
  simplex} and its {\em dimension} is one less than its cardinality,
$\dim \sigma = \card \sigma - 1$.  A simplex of dimension
  $i$ is called an $i$-simplex. If $\tau$ and $\sigma$ are two
  simplices such that $\tau \subseteq \sigma$, then $\tau$ is called a
  {\em face} of $\sigma$, and $\sigma$ is called a {\em coface} of
  $\tau$. The elements of $\sigma$ are also referred to as the
  \emph{vertices} of $\sigma$ and the \emph{vertex set} of $K$ is the
  set of vertices of all simplices in $K$, $\Vertexset K =
  \bigcup_{\sigma \in K} \sigma$. When an abstract simplex $\sigma
  \subseteq \Rspace^\Dim$ has its vertices in $\Rspace^\Dim$, it is
naturally associated to the geometric simplex defined as
$\Conv{\sigma}$.
The dimension of $\Conv{\sigma}$, which is the dimension of the affine space
$\Aff{\sigma}$, cannot be larger than the dimension of the
abstract simplex $\sigma$. When the dimension of the geometric simplex
$\Conv{\sigma}$ coincides with that of the abstract simplex $\sigma$,
we say that $\sigma$ is \emph{non-degenerate}.
For a simplicial complex $K$ with vertices in
  $\Rspace^\Dim$, we say that $K$ is \emph{geometrically realized}
(or \emph{embedded}) if (1) $\dim(\sigma)=\dim(\Aff\sigma)$ for all
$\sigma \in K$, and (2) $\Conv(\alpha \cap\beta) = \Conv\alpha \cap
\Conv\beta$ for all $\alpha, \beta \in K$.

Given a set of abstract simplices $\Sigma$ with vertices in $\Rspace^\Dim$ (not
necessarily forming a simplicial complex), we let $\Sigma^{[i]}$
designate the set of $i$-simplices of $\Sigma$. We define the
\emph{shadow} of $\Sigma$ as the subset of $\Rspace^\Dim$
covered by the relative interior of the geometric simplices associated
to the abstract simplices in $\Sigma$, $\Shadow{\Sigma} =
\bigcup_{\sigma \in \Sigma} \relint{\Conv{\sigma}}$. The closure of
$\Sigma$ is the smallest simplicial complex that contains $\Sigma$.

\subsection{Barycentric coordinates}

\Modif{Consider an abstract simplex $\alpha \subseteq \Rspace^\Dim$
  and note that $\alpha$ is non-degenerate if and only if its vertices
  are affinely independent. Suppose that $\alpha$ is non-degenerate
  and consider the affine combination $x = \sum_{a \in \alpha}
  \lambda_a a$ with $\sum_{a \in \alpha} \lambda_a = 1$. Then, the
  $\lambda_a$ are uniquely determined by $x$ and are called the
  \emph{(normalized) barycentric coordinates} of $x$ with respect to
  $\alpha$. In the paper, we shall denote each
  $\lambda_a$ as $\BaryCoord x \alpha a$.  }

\subsection{Chains and weighted norms}

Chains play an important role in this work as they provide a tool to
embed the discrete set of candidate solutions (triangulations of $\M$
in some simplicial complex $K$) into a larger continuous space (the
$d$-chains of $K$). In this section, we recall some standard
definitions concerning chains from \cite{munkres1993elements}.
Given an abstract simplex $\sigma$, two orderings of the
  vertices of $\sigma$ are said to be \emph{equivalent} if they differ
  from one another by an even permutation. The orderings of the
  vertices of $\sigma$ fall into equivalent classes: two classes if
  $\dim\sigma>0$ and one class if $\dim\sigma=0$. Each of these
  classes is called an \emph{orientation} of $\sigma$. An
  \emph{oriented simplex} is a simplex $\sigma$ together with an
  orientation of $\sigma$. We denote as $[v_0,\ldots,v_d]$ the
  oriented $d$-simplex consisting of the $d$-simplex $\{v_0, \ldots,
  v_d\}$ together with the equivalent class of the particular ordering
  $(v_0,\ldots,v_d)$.  Consider an abstract simplicial complex $K$
and assume that each simplex $\sigma$ in $K$ is given an arbitrary
orientation. A {\em $d$-chain} of $K$ with coefficients in \Rspace is
a formal sum $\gamma = \sum_{\sigma} \gamma(\sigma) \sigma$, where
$\sigma$ ranges over all $d$-simplices of $K$ and $\gamma(\sigma)
\in \Rspace$ is the value (or the coordinate) assigned to the 
$d$-simplex $\sigma$ with the rule that if $\sigma$ and $\sigma'$
are the same simplex but have two different orientations,
then $\sigma = - \sigma'$. The set of such $d$-chains is a vector
space denoted by $C_d(K,\Rspace)$. Recall that the $\ell_1$-norm of
$\gamma$ is defined by $\| \gamma \|_1 = \sum_{\sigma}
|\gamma(\sigma)|$. Let $W$ be a weight function which assigns a
non-negative weight $W(\sigma)$ to each $d$-simplex $\sigma$ of $K$.
The $W$-weighted $\ell_1$-norm of $\gamma$ is expressed as
$\|\gamma\|_{1,W} = \sum_{\sigma} W(\sigma) \abs{\gamma(\sigma)}$. We
shall say that a chain $\gamma$ is {\em carried by} a subcomplex $D$
of $K$ if $\gamma$ has value 0 on every simplex that is not in
$D$. The {\em support} of $\gamma$ is the set of simplices on which
$\gamma$ has a non-zero value. It is denoted by
$\Support{\gamma}$. The {\em boundary operator} is a homomorphism
$\partial: C_d(K,\Rspace) \to C_{d-1}(K,\Rspace)$ that associates to
each oriented $d$-simplex $\sigma = [v_0, \ldots, v_d]$ the
$(d-1)$-chain:
\[
\partial \sigma = \sum_{i=0}^d (-1)^i \, [v_0, \ldots, \hat v_i, \ldots, v_d],
\]
where the symbol $\hat v_i$ means that the vertex $v_i$ has been
deleted from the sequence of vertices forming $\sigma$. \Modif{A
  $d$-chain $\gamma \in C_d(K,\Rspace)$ whose boundary vanishes,
  $\partial \gamma = 0$, is called a \emph{$d$-cycle}.}

%
%
%
%
%

\section{Background on Delaunay complexes}
\label{section:Delaunay-complex}

In this section, we recall basic facts about Delaunay complexes
(Section~\ref{section:Delaunay-geometric-characterization}). We then
give a variational characterization of Delaunay complexes
(Section~\ref{section:Delaunay-variational-characterization}).
Throughout the section, $P$ designates a finite point set of
$\Rspace^\Dim$.

\subsection{Definitions and basic property}
\label{section:Delaunay-geometric-characterization} 

\begin{definition}[Delaunay simplex]
A \emph{Delaunay simplex} of $P$ is an abstract simplex $\sigma
\subseteq P$ for which there exists a ball $B$ whose boundary
circumscribes $\sigma$ and whose interior does not contain any point
of $P$.
\end{definition}

\begin{definition}[Delaunay complex]
  The set of Delaunay simplices form an abstract simplicial complex
  called the \emph{Delaunay complex} of $P$ and is denoted as $\Del{P}$.
\end{definition}

We now state a classical result on Delaunay complexes, for which we
need two extra definitions.

\begin{definition}[General position]
  Letting $d = \dim(\Aff P)$, we say that $P \subseteq \Rspace^\Dim$
  is in \emph{general position} if no $d+2$ points of $P$ lie on a
  common $(d-1)$-dimensional sphere and no $k+2$ points of
    $P$ lie on the same $k$-dimensional flat for $k<d$.
\end{definition}

\begin{definition}[Triangulation]
A \emph{triangulation} of $P$ is an abstract simplicial complex whose
vertex set is $P$, whose shadow is $\Conv P$, and which is
geometrically realized.
\end{definition}

\begin{theorem}
  \label{theorem:Delaunay-complexes-are-triangulations}
When $P$ is in general position, $\Del{P}$ is a triangulation of
$P$.
\end{theorem}

\subsection{A variational characterization}
\label{section:Delaunay-variational-characterization} 

The Delaunay complex of $P$ optimizes many functionals over the set of
triangulations of $P$
\cite{boissonnat1998algorithmic,rippa1990minimal,musin1997properties},
one of them being the Delaunay energy that we shall now define
\cite{chen2004MeshSmoothing}.

In preparation for this, we recall a famous result which says that building a
Delaunay complex in $\Rspace^\Dim$ is topologically equivalent to
building a lower convex hull in $\Rspace^{\Dim+1}$. For simplicity, we
shall identify each point $x \in \Rspace^\Dim$ with the point $(x,0)$ in
$\Rspace^{\Dim+1}$.  Consider the paraboloid $\paraboloid
\subseteq \Rspace^{\Dim+1}$ defined as the graph of the function
$\|\cdot\|^2 : \Rspace^\Dim \to
\Rspace$, $x \mapsto \| x \|^2$, where $\|\cdot\|$ designates the
Euclidean norm; see Figure \ref{figure:lifting}, left. For each point $x \in \Rspace^\Dim$, its vertical
projection onto $\paraboloid$ is the point $\hat{x} = (x,\|x\|^2)
\in \Rspace^{\Dim+1}$, which we call the {\em lifted image} of
$x$. Similarly, the lifted image of $P \subseteq
\Rspace^\Dim$ is $\hat{P} = \{ \hat{p} \mid p \in P\}$. Recall that
the lower convex hull of $\hat{P}$ is the portion of $\Conv \hat{P}$
visible to a viewer standing at $x_{d+1} = -\infty$. A classical
result says that for all $\sigma \subseteq P$, the following
equivalence holds: $\sigma$ is a Delaunay simplex of $P$ if and only if
$\Conv \hat{\sigma}$ is contained in the lower convex hull of
$\hat{P}$ \cite{edelsbrunner1996incremental}.

\begin{figure}[htb]
  \def\svgwidth{1\linewidth}
  \centering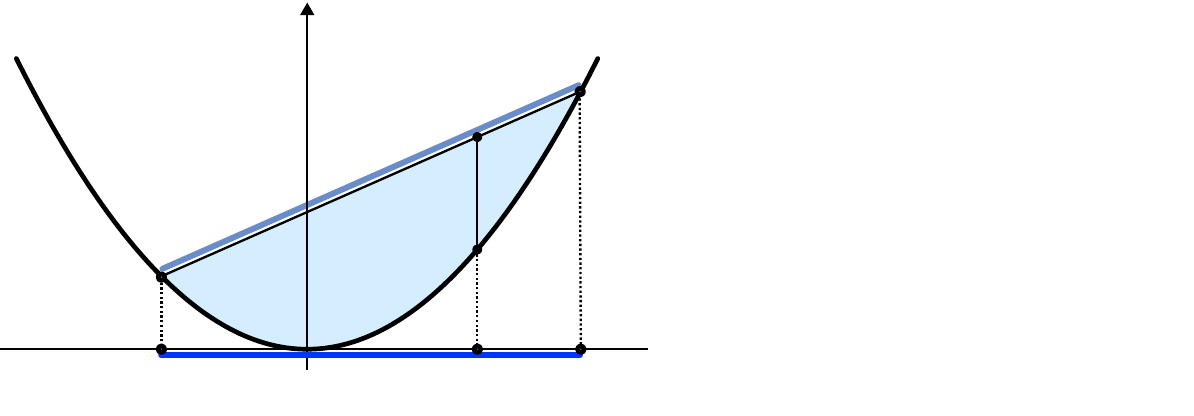
  \caption{Left: the Delaunay weight of $\sigma$ can be depicted as
    the $(d+1)$-volume of the blue region between the lifted geometric
    simplex $\Conv \hat\sigma$ and the paraboloid $\paraboloid$ (see
    Lemma \ref{lemma:interpreting-Delaunay-weights}). Right: the
    Delaunay weight of $\sigma$ is also the $(d+1)$-volume of the blue
    region lying below the graph of $-\operatorname{Power}_\sigma$ and above
    $\Conv \sigma$. \label{figure:lifting}}
\end{figure}

We are now ready to define the Delaunay energy of any triangulation
$T$ of $P$. Let $d = \dim(\Aff P)$. Given a triangulation $T$ of $P$,
the \emph{Delaunay energy} $\Edel(T)$ of $T$ is defined as the
$(d+1)$-volume between the $d$-manifold $\Shadow{\hat{T}} =
\bigcup_{\sigma \in T} \Conv{\hat{\sigma}}$ and the paraboloid
\paraboloid. Let us give a formula for this energy.  Consider a point
$x \in \Conv P$. By construction, $x$ belongs to at least one
geometric $d$-simplex $\Conv \sigma$ for some $\sigma \in T$. Erect an
infinite vertical half-line going up from $x$.  This half-line
intersects the paraboloid \paraboloid at point $\hat x$ and $\Conv
\hat{\sigma}$ at point $x_\sigma^\star$; see Figure \ref{figure:lifting}, left. We have
\[
\Edel(T) = \sum_{\sigma \in T^{[d]}} \int_{x \in \Conv \sigma} \| \hat{x} - x_\sigma^\star\|  \, dx.
\]
Let us recall a well-known result
\cite{musin1997properties,chen2004optimal}, that is a direct
consequence of the lifting construction:

\begin{theorem}
  \label{theorem:euclidean}
  When  $P$ is in general position, the triangulation of $P$ that
  minimizes the Delaunay energy is unique and is the Delaunay complex of $P$.
\end{theorem}

%
%
%
%
%

\section{Delaunay weight}
\label{section:Delaunay-weight}


\Modif{In this section, we define the Delaunay weight of a simplex with
vertices in $\Rspace^\Dim$
({Section~\ref{section:delaunay-weights-definition}}). We show that
the Delaunay energy defined in the previous section can be expressed
as a sum of Delaunay weights
({Section~\ref{section:delaunay-weights-as-heights}}). We then provide
an intrinsec expression for Delaunay weights
({Section~\ref{section:delaunay-weights-intrinsec-formulation}}) that
motivates extending the Delaunay energy to collection of $d$-simplices
in $\Rspace^\Dim$ and more generally to $d$-chains of simplicial
complexes with vertices in $\Rspace^\Dim$, as will be done in the next
section.}

\subsection{Definition}
\label{section:delaunay-weights-definition}

Let $\sigma \subseteq \Rspace^\Dim$ be an abstract simplex. If
$\sigma$ is non-degenerate, we define $S(\sigma)$ as the smallest
$(N-1)$-sphere that circumscribes $\sigma$. We also let $Z(\sigma)$
and $R(\sigma)$ denote the center and radius of $S(\sigma)$,
respectively. Finally, we introduce the map
\[
x \mapsto \Power{\sigma}{x} = \|x-Z(\sigma)\|^2 - R(\sigma)^2,
\]
which associates each point $x \in \Rspace^\Dim$ with the power
distance of $x$ from $S(\sigma)$.

\begin{definition}[Delaunay weight]
  \label{definition:delaunay-weight}
The {\em Delaunay weight} of an abstract $d$-simplex $\sigma
\subseteq \Rspace^\Dim$ is:
\begin{equation*}
  \omega(\sigma) =
  \begin{cases}
    - \int_{x \in \Conv{\sigma}} \Power{\sigma}{x}  \, dx \quad \quad &\text{if $\sigma$ is non-degenerate,}\\
    0 \quad\quad &\text{otherwise}.
  \end{cases}
\end{equation*}
\end{definition}

{Figure~\ref{figure:lifting}}, right depicts graphically the Delaunay
weight of a simplex $\sigma \subseteq \Rspace^\Dim$.  It is
  worth noting that the Delaunay weight is defined for any abstract
  $d$-simplex $\sigma$, irrespective of whether
  $\sigma$ is the Delaunay simplex of some point set $P$ or not. The
  reason for calling it a Delaunay weight will become clear in the
  next section.

\subsection{Delaunay energy as a sum of Delaunay weights}
\label{section:delaunay-weights-as-heights}

Before showing that the Delaunay energy can be expressed as a sum of
Delaunay weights in Lemma~\ref{lemma:interpreting-Delaunay-weights},
we first provide a useful expression of $\Power{\sigma}{x}$ when $x$
is an affine combination of the vertices of $\sigma$ (Lemma
\ref{lemma:power-expression-with-additional-point}) and deduce an
alternative expression of the Delaunay weight in Lemma~\ref{lemma:interpreting-Delaunay-weights}.

\begin{lemma}
  \label{lemma:power-expression-with-additional-point}
  Let $\sigma \subseteq \Rspace^\Dim$ be a non-degenerate simplex and
  let $x$ be an affine combination of the vertices of $\sigma$.  For every $z
  \in \Rspace^\Dim$
\[
\Power{\sigma}{x} = \|x-z\|^2 -  \sum_{a \in \sigma} \BaryCoord x \sigma a \, \|a-z\|^2.
\]
\end{lemma}

\begin{proof}
  Recall that $\Power{\sigma}{x} = \|x-Z(\sigma)\|^2 - R(\sigma)^2$. On one hand, we have
  \begin{eqnarray*}
    \|x - Z(\sigma)\|^2 &=& \|x-z\|^2 + 2(x-z)\cdot(z-Z(\sigma)) + \|z - Z(\sigma)\|^2.
  \end{eqnarray*}
  On the other hand, writing $\lambda_a = \BaryCoord x \sigma a$ for short, we have 
  \begin{eqnarray*}
    R(\sigma)^2 &=& \sum_{a \in \sigma} \lambda_a \|Z(\sigma) - a\|^2\\
    &=& \sum_{a \in \sigma} \lambda_a \left[ \|Z(\sigma) - z\|^2 + 2(Z(\sigma)-z)\cdot(z-a) + \|z-a\|^2 \right]\\
    &=& \|Z(\sigma) - z\|^2 + 2 (Z(\sigma)-z)\cdot(z-x) + \sum_{a \in \sigma}  \lambda_a  \|z-a\|^2.
  \end{eqnarray*}
  Substracting the above expressions of $\|x-Z(\sigma)\|^2$ and $R(\sigma)^2$ yields the result.
\end{proof}

\begin{lemma}
  \label{lemma:interpreting-Delaunay-weights}
  
  For any non-degenerate abstract $d$-simplex $\sigma \subseteq \Rspace^\Dim$, its
  Delaunay weight represents the $(d+1)$-volume between the lifted
  geometric simplex $\Conv \hat{\sigma}$ and the paraboloid
  \paraboloid, {\em i.e.}
  \[
  \omega(\sigma) = \int_{x \in \Conv \sigma} \| \hat{x} - x_\sigma^\star\|   \, dx,
  \]
  where $\hat{x} = (x, \|x\|^2)$ and $x_\sigma^\star =
    \sum_{a \in \sigma} \BaryCoord x \sigma a \hat{a}$.  If $P$
  designates a finite point set of $\Rspace^\Dim$ in general position
  and $d = \dim(\Aff P)$, then the Delaunay energy of any
  triangulation $T$ of $P$ can be expressed as
  \[
  \Edel(T) = \sum_{\sigma \in T^{[d]}} \omega(\sigma).
  \]
\end{lemma}

\begin{proof}
  Letting $x \in \Conv \sigma$, we show that $-\Power \sigma x =
  \|\hat{x} - x_\sigma^\star\|$. Applying Lemma
  \ref{lemma:power-expression-with-additional-point} with $z=0$ and
  writing $\lambda_a = \BaryCoord x \sigma a$ for short, we get that
  \begin{align*}
    - \Power{\sigma}{x} &=   \sum_{a \in \sigma} \lambda_a \, \|a\|^2 - \|x\|^2 \\
    &=  \left\|  \sum_{a \in \sigma} \lambda_a (a,\|a\|^{2}) - (x,\|x\|^2) \right\|\\
    &= \left\| \sum_{a \in \sigma} \lambda_a  \hat{a} - \hat{x} \right\|\\
    &= \| x^\star_\sigma - \hat{x} \|.
  \end{align*}
  The
  expressions of both $\omega(\sigma)$ and $\Edel(T)$ follow
  immediately.
\end{proof}

The above lemma suggests the following interpretation of the Delaunay
energy. Given a triangulation $T$ of $P$, consider the map $w_T :
\Conv P \to \Rspace$ whose restriction to any $d$-simplex $\sigma$ of
$T$ is defined by $w_T(x) = - \Power \sigma x$. The graph of $w_T$ is
a $d$-dimensional piecewise parabolic manifold $\mathcal W_T$ which
has been depicted for two different triangulations in Figure
\ref{figure:comparing-Delaunay-energy}. The Delaunay energy can then
be interpreted as the $(d+1)$-volume of the region lying below $\mathcal W_T$
and above $\Conv P$.

\begin{figure}[htb]
  \begin{center}
    \includegraphics[width=0.49\linewidth]{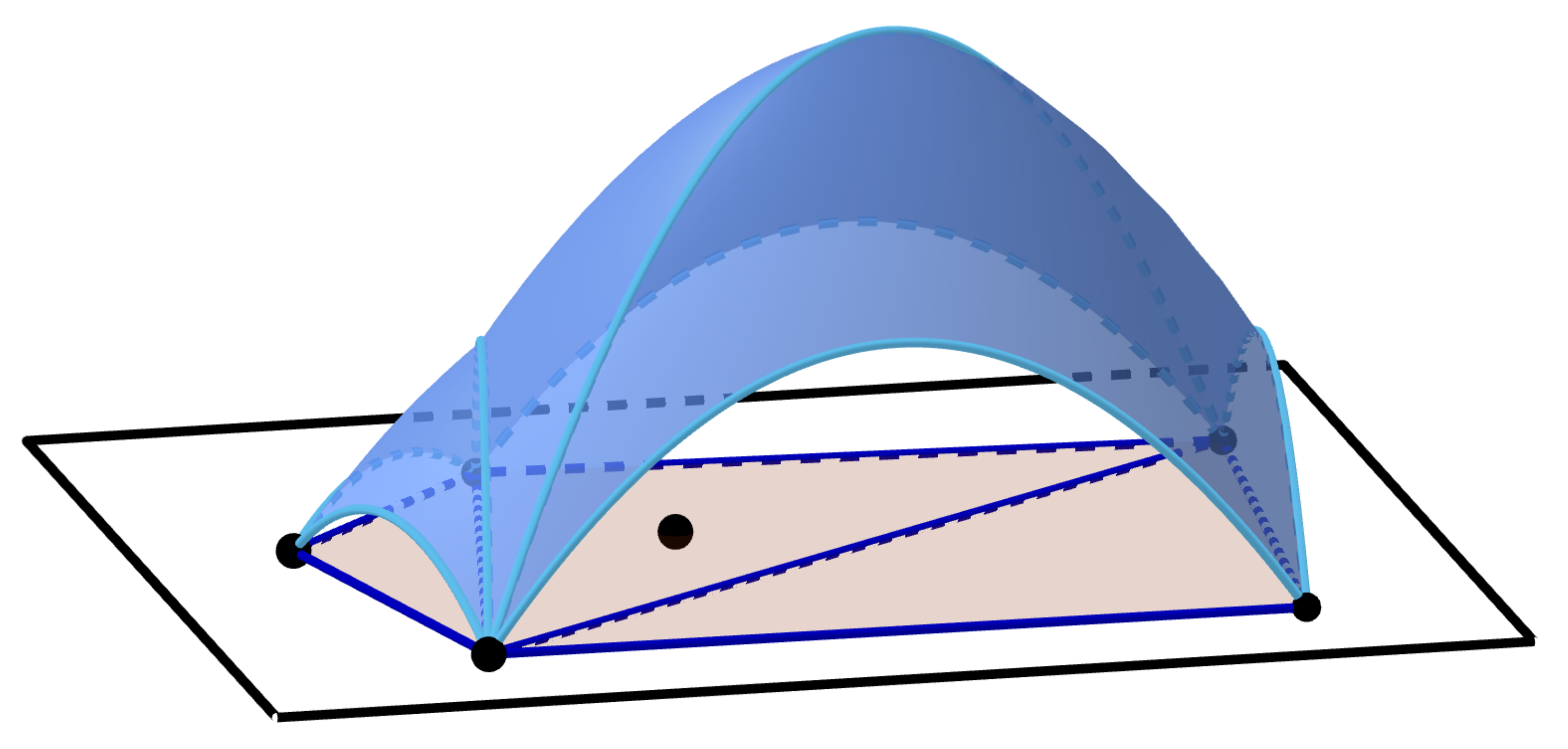}
    \includegraphics[width=0.49\linewidth]{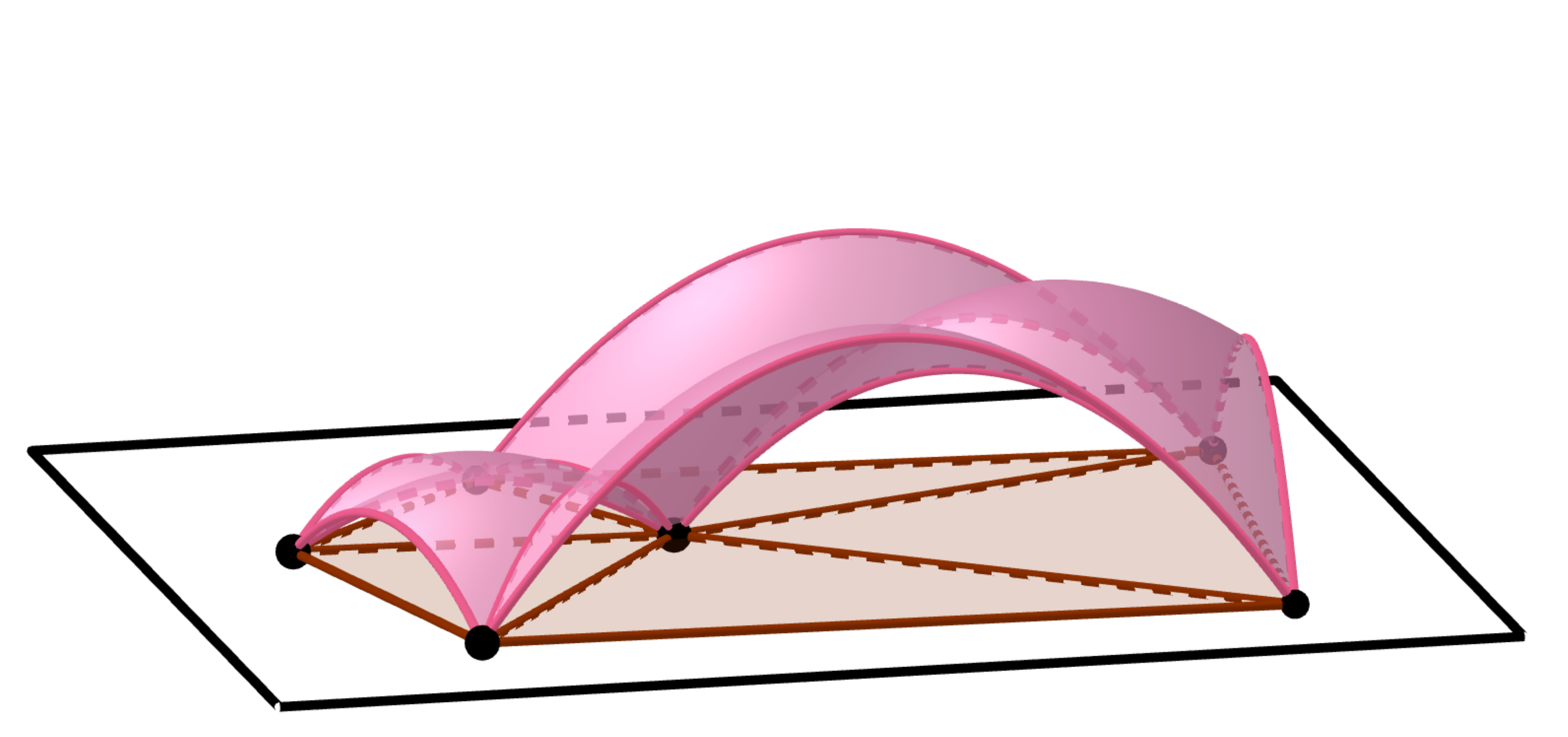}
  \end{center}
  \caption{Two triangulations of six points (black dots) in the plane. For each
    triangulation $T$, the Delaunay energy is the volume between the
    convex hull of the points and the piecewise parabolic surface
    $\mathcal W_T$. The surface $\mathcal W_T$ is lowest and therefore
    the Delaunay energy of $T$  is smallest when $T$ is the Delaunay complex
    as is the case on the
    right. \label{figure:comparing-Delaunay-energy}}
\end{figure}

\subsection{Intrinsec closed expression}
\label{section:delaunay-weights-intrinsec-formulation}

Below, we give a closed expression for the Delaunay weight due to Chen
and Holst in \cite{chen2011efficient}. For completeness, we provide a
proof. Writing $\Volume{\sigma}$ for the $d$-dimensional volume of
$\Conv{\sigma}$, we have:

\begin{lemma}[\cite{chen2011efficient}]
  \label{lemma:weight}
  The weight of the abstract $d$-simplex $\sigma = \{a_0, \ldots, a_d\}$ is
  \begin{equation*}
    \omega(\sigma) = \frac{1}{(d+1)(d+2)} \Volume{ \sigma } 
    \left[
      \sum_{0 \leq i < j \leq d} \|a_i - a_j\|^2  \right].
  \end{equation*}
\end{lemma}

\begin{proof}
  Let $\sigma = \{a_0,a_1,\ldots,a_d\} \subseteq \Rspace^\Dim$. If
  $\sigma$ is degenerate, then $\Volume\sigma=0$ and the result is
  clear. Suppose that $\sigma$ is non-degenerate and recall that
  the standard simplex is
  \[
  \Delta_d = \{ \lambda \in \Rspace^d \mid \sum_{i=1}^d \lambda_i \leq 1; \lambda_i \geq 0,
  i=1,2,\ldots, d \}.
  \]
  We introduce the map $\psi : \Rspace^d \to \Rspace^d$, defined by $\psi(\lambda)
  = a_0 + \sum_{i=1}^d \lambda_i (a_i - a_0)$, which establishes a
  one-to-one correspondence between the points $\lambda$ of the standard
  simplex $\Delta_d$ and the points $x = \psi(\lambda)$ of $\Conv{\sigma}$.  Making the change of
  variable $x = \psi(\lambda) \rightarrow \lambda$, we get that:
  \[
  \wdel(\sigma) = \int_{\lambda \in \Delta_d}  -\Power\sigma{\psi(\lambda)} \cdot | \det(\Differential\psi)(\lambda) | \, d\lambda.
  \] Noting that $\Differential\psi(\lambda)$ is the $d\times d$ matrix whose
  $i$th column is the vector $a_i - a_0$, we deduce that
  $|\det(\Differential\psi)(\lambda) | = d! \Volume{\sigma}$.  Observing that
  $\psi(\lambda)$ has (normalized) barycentric coordinates $(1 -
  \sum_{i=1}^d \lambda_i, \lambda_1, \lambda_2, \ldots, \lambda_d)$
  and applying Lemma \ref{lemma:power-expression-with-additional-point} with $z=a_0$, we 
  can write:
  \[
  \Power\sigma{\psi(\lambda)} = - \left( \sum_{i=1}^d \lambda_i \|a_i - a_0\|^2 \right) + \| \psi(\lambda) - a_0 \|^2,
  \]
  and thus obtain (after plugging in the expression of $\psi(\lambda)$)
  \[
  w(\sigma) = d! \Volume{\sigma} \int_{\lambda \in \Delta_d} 
  \left[\sum_{i=1}^d \lambda_i \|a_i - a_0\|^2 - 
  \left\| \sum_{i=1}^d \lambda_i (a_i - a_0) \right\|^2
  \right] \, d\lambda.
  \] We then use a formula for integrating a homogeneous polynomial on
  the standard simplex that may be found in \cite{zeller2001almost}:
  \[
  \int_{\lambda \in \Delta_d} \lambda_1^{\eta_1} \ldots \lambda_d^{\eta_d} \, d\lambda =
  \frac{\eta_1 ! \ldots \eta_d ! }{(d + \sum_i \eta_i)!}.
  \]
  We obtain that
  \begin{equation*}
    w(\sigma) = \frac{1}{(d+1)(d+2)} \Volume{ \sigma } 
    \left[
      d \sum_{i=1}^d \|a_i - a_0\|^2 - 2 \sum_{1 \leq i < j \leq d}
      \DotProd{(a_i-a_0)}{(a_j - a_0)}
      \right].
  \end{equation*}
  Observing that $\|a_i-a_0\|^2 + \|a_j - a_0\|^2 -
  2(a_i-a_0)\cdot(a_j-a_0) = \|a_i - a_j\|^2$, we can further rearrange
  the above formula to get the result.
\end{proof}

It follows from Definition~\ref{definition:delaunay-weight} but also
from the expression of the Delaunay weight given in
Lemma~\ref{lemma:weight} that two isometric simplices have the same
Delaunay weight. Hence, a Delaunay energy can be straightforwardly
associated to any collection $\Sigma$ of $d$-simplices
living in $\Rspace^\Dim$ by setting $E(\Sigma) = \sum_{\sigma \in
  \Sigma} \omega(\sigma)$. It is then tempting to ask what would
happen if one minimizes this energy over all collections
$\Sigma$ of $d$-simplices whose vertices sample a $d$-dimensional
submanifold \M and whose union is homeomorphic to that submanifold. As
is, the problem is non-convex. We shall transform it into a convex
problem in the next section.

%
%
%
%
%


\section{Variational formulation for submanifold reconstruction}
\label{section:variational-formulation}

Afterwards, we assume that the shape $\M$ we wish to reconstruct is a
compact orientable $C^2$ $d$-dimensional submanifold of $\Rspace^\Dim$
for some $d< \Dim$. We let $P$ be a finite point set that samples \M
and suppose furthermore that we have at our disposal a simplicial
complex $K$ whose vertices are the points of $P$. The complex $K$ can
be thought of as some rough approximation of \M as illustrated in Figure \ref{figure:curve-sample-and-complex}.

\begin{figure}[htb]
  \centering
\begingroup%
  \makeatletter%
  \providecommand\color[2][]{%
    \errmessage{(Inkscape) Color is used for the text in Inkscape, but the package 'color.sty' is not loaded}%
    \renewcommand\color[2][]{}%
  }%
  \providecommand\transparent[1]{%
    \errmessage{(Inkscape) Transparency is used (non-zero) for the text in Inkscape, but the package 'transparent.sty' is not loaded}%
    \renewcommand\transparent[1]{}%
  }%
  \providecommand\rotatebox[2]{#2}%
  \newcommand*\fsize{\dimexpr\f@size pt\relax}%
  \newcommand*\lineheight[1]{\fontsize{\fsize}{#1\fsize}\selectfont}%
  \ifx\svgwidth\undefined%
    \setlength{\unitlength}{351.6787262bp}%
    \ifx\svgscale\undefined%
      \relax%
    \else%
      \setlength{\unitlength}{\unitlength * \real{\svgscale}}%
    \fi%
  \else%
    \setlength{\unitlength}{\svgwidth}%
  \fi%
  \global\let\svgwidth\undefined%
  \global\let\svgscale\undefined%
  \makeatother%
  \begin{picture}(1,0.22316263)%
    \lineheight{1}%
    \setlength\tabcolsep{0pt}%
    \put(0,0){\includegraphics[width=\unitlength,page=1]{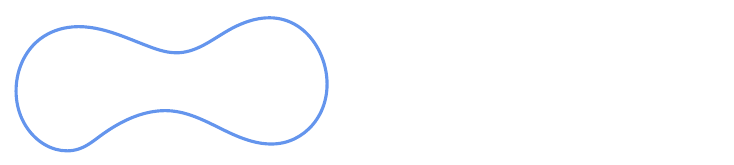}}%
    \put(0.0194129,0.18374856){\color[rgb]{0.39215686,0.58431373,0.92941176}\makebox(0,0)[lt]{\lineheight{1.25}\smash{\begin{tabular}[t]{l}$\M$\end{tabular}}}}%
    \put(0.21354204,0.18845468){\makebox(0,0)[lt]{\lineheight{1.25}\smash{\begin{tabular}[t]{l}$P$\end{tabular}}}}%
    \put(0.71298331,0.20139666){\makebox(0,0)[lt]{\lineheight{1.25}\smash{\begin{tabular}[t]{l}$K$\end{tabular}}}}%
    \put(0,0){\includegraphics[width=\unitlength,page=2]{curve-sample-and-complex.pdf}}%
  \end{picture}%
\endgroup%

  \caption{Left: a $d$-dimensional submanifold $\M$ (for $d=1$) and a noisy sample
    $P$ of $\M$. Right: a simplicial complex $K$ whose vertex set is
    $P$. \label{figure:curve-sample-and-complex}}
\end{figure}

Details on how to derive $K$ from $P$ are given at the end of the
section.  In this section, we describe a convex optimization problem
on the $d$-chains of $K$ and state conditions under which the solution
to that problem is unique and provides a faithful reconstruction of
\M. The concept of faithful reconstruction encapsulates what we mean
by a ``desirable'' reconstruction of \M:

\begin{definition}[Faithful reconstruction]
  Consider a subset $\M \subseteq \Rspace^\Dim$ whose reach is positive,
  and a simplicial complex $D$ with  vertex set in $\Rspace^\Dim$. We
  say that $D$ \emph{reconstructs $\M$ faithfully} (or is a
  \emph{faithful reconstruction} of $\M$) if the following three
  conditions hold:
  \begin{description}
  \item[\styleitem{Embedding:}] $D$ is geometrically realized;
  \item[\styleitem{Closeness:}] $\Shadow D \subseteq \Offset \M r$ for some $0 \leq r < \Reach \M$;
  \item[\styleitem{Homeomorphism:}] the projection map $\pi_\M: \Shadow D \to \M$ is a homeomorphism.
  \end{description}
\end{definition}

We note that when $D$ reconstructs $\M$ faithfully,
  $\Shadow{D}$ and $\M$ are homeomorphic and $D$ is a triangulation of \M.

The rest of the section is organized as
follows. Section~\ref{section:least-norm-problem} presents our convex
optimization problem on the $d$-chains of $K$.
Section~\ref{section:encoding} shows that the feasible set of that
problem contains all
faithful reconstructions of \M in $K$. In
Section~\ref{section:geometric-conditions}, we introduce the necessary
definitions to state our main theorem in
Section~\ref{section:main-theorem}.

Because we have assumed $\M$ to be a compact $C^2$ submanifold of
$\Rspace^\Dim$, the reach of $\M$ is positive and finite
\cite{federer-59}. Afterwards, we denote it as $\reach = \Reach \M$.
Given $m \in \M$, we denote the vector tangent space to $\M$
  at $m$ as $T_m\M$ and the affine tangent space to $\M$ at $m$ as
  $\Tangent m \M$. Clearly, $\Tangent m \M = x + T_x \M$. In the
rest of the section, we assume that \M together with all $d$-simplices
of $K$ have received an arbitrary orientation.  We also assume that
$\Shadow K \subseteq \Offset \M r$ for some $0 \leq r < \Reach \M$ and
that none of the $d$-simplices of $K$ are orthogonal to
$\M$. Precisely, defining the {\em angular deviation} of a simplex
$\sigma$ relatively to \M as
\[
\AngularDeviation \sigma = \max_{m \in \pi_\M(\Conv\sigma)} \angle(\Aff \sigma, \Tangent m \M),
\]
we assume that each $d$-simplex $\sigma \in K$ is such that
$\AngularDeviation{\sigma} < \frac{\pi}{2}$. This allows us to assign
to each $d$-simplex $\sigma \in K$ a sign with respect to \M as follows:
\[
\sign{\sigma}{\M} =
\begin{cases}
  1 \quad \quad &\text{if the orientation of $\sigma$ is consistent with that of \M,}\\
  -1 \quad \quad &\text{otherwise.}
\end{cases}
\]
We refer the reader to Appendix~\ref{appendix:transfering-orientation} for a formal definition of consistency
and more details.

\subsection{Least $\ell_1$-norm problem}
\label{section:least-norm-problem}

We need notation to describe the convex optimization problem that we
are considering.  Let $\omega$ be the weight function which assigns to
each $d$-simplex $\sigma$ of $K$ its Delaunay weight $\omega(\sigma)$
introduced in Section \ref{section:preliminaries}.  We define the {\em
  Delaunay energy} of the chain $\gamma \in C_d(K,\Rspace)$ to be its
$\omega$-weighted $\ell_1$-norm:
\begin{equation*}
  \Edel(\gamma) = \|\gamma\|_{1,\omega} = \sum_{\sigma}  \omega(\sigma) \cdot |\gamma(\sigma)| =
  \sum_{\sigma}\left( \int_{x \in \Conv{\sigma}} - \Power{\sigma}{x} \, dx \right) \cdot  |\gamma(\sigma)|,
\end{equation*}
where $\sigma$ ranges over all $d$-simplices of $K$. The
  Delaunay energy is the objective function of our optimization
  problem. To describe the constraint functions, let $\Indicator{X}: \Rspace^\Dim \to \{0,1\}$ denote the
  indicator function of a subset $X \subseteq \Rspace^\Dim$. Suppose that
  $\Shadow K \subseteq \Offset \M r$ for some $0 \leq r < \Reach \M$
  and that each $d$-simplex $\sigma \in K$ satisfies
  $\operatorname{angularDeviation}_\M(\sigma) < \frac{\pi}{2}$. Given
  $m_0 \in \M$, we assign to each $d$-chain $\gamma$ of $K$ the real
  number:
\[
\Load {m_0} \M K \gamma = \sum_{\sigma \in K^{[d]}} \gamma(\sigma) \sign{\sigma}{\M} \Indicator{\pi_\M(\Conv{\sigma})}(m_0)
\]
and call it the {\em load} of $\gamma \in C_d(K,\Rspace)$ on
  $\M$ at $m_0$.  Roughly, it measures the ``flux'' of the chain
  $\gamma$ above point $m_0 \in \M$. We illustrate its evaluation in
  {Figure~\ref{figure:least-norm-problem}}.

\begin{figure}[htb]
  \centering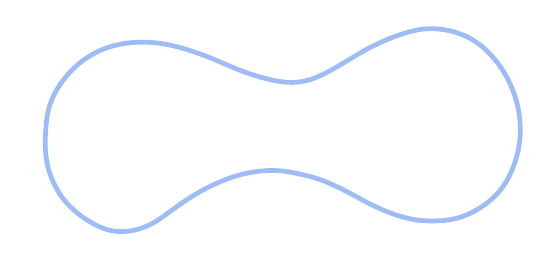
  \caption{A simplicial complex $K$ whose edges have an orientation
    consistent with that of $\M$ and a $1$-cycle $\gamma$ of $K$. To
    evaluate the load of $\gamma \in C_1(K,\Rspace)$ on the curve $\M$
    at $m_0$, one has first to select edges (depicted in blue) for
    which there is a point (blue dot) whose projection onto $\M$ is
    $m_0$ and sum up the coefficients of $\gamma$ on these edges. In
    this example, $\Load {m_0} \M K \gamma=1$ and therefore $\gamma$
    satisfies the constraints of
    Problem~(\ref{problem:reconstruction}).
    \label{figure:least-norm-problem}}
\end{figure}

Letting $m_0$ be a generic\footnote{Generic in the sense that it is
not in the projection on \M of the convex hull of any $(d-1)$-simplex
of $K$.}  point on $\M$, we are interested in the following
optimization problem over the set of chains in $C_d(K,\Rspace)$:

\medskip

\begin{tcolorbox}
  \vspace{-4mm}
  \begin{mini*}
    {\gamma}{\Edel(\gamma)}{}{} \label{problem:reconstruction}
    \addConstraint{\partial \gamma}{=0} \tag{$\star$}
    \addConstraint{\Load {m_0} \M K \gamma}{=1}
  \end{mini*}
\end{tcolorbox}

Problem (\ref{problem:reconstruction}) is a least-norm problem whose
constraint functions $\partial$ and $\load_{m_0,\M,K}$ are clearly
linear. It is therefore a convex optimization problem.
The first constraint $\partial \gamma = 0$ expresses the
fact that we are searching for $d$-cycles. The second constraint
$\Load {m_0} \M K \gamma = 1$ is a
normalization of $\gamma$ and forbids the zero chain to belong to the
feasible set. Two chains that satisfy the constraints are
  depicted in {Figures~\ref{figure:least-norm-problem}} and {\ref{figure:curve-code}}. We shall see
that, under the assumptions of our main theorem, the solution to
Problem~({\ref{problem:reconstruction}}) takes its coordinate values in
$\{0,+1,-1\}$ and is furthermore the code of a faithful reconstruction
of \M.

In Problem~(\ref{problem:reconstruction}), besides the simplicial
complex $K$ that we shall see how to build from $P$, the knowledge of the manifold
$\M$ seems to be required as well for expressing the normalization
constraint. In Section \ref{section:RealisticAlgorithm}, we discuss how to transform
Problem~(\ref{problem:reconstruction}) into an equivalent problem that
does not refer to \M anymore.

\subsection{Faithful reconstructions are encoded in the feasible set}
\label{section:encoding}

Given a subcomplex $D \subseteq K$, we associate to $D$ the $d$-chain
$\code{D}$ of $K$ whose coordinate on the $d$-simplex $\sigma$ is:
\[
\code{D}(\sigma) =
\begin{cases}
  \sign{\sigma}{\M} \quad \quad &\text{if $\sigma \in D^{[d]}$,}\\
  0 \quad \quad  &\text{otherwise.}
  \end{cases}
\]
We note that whenever $D$ is a faithful reconstruction of $\M$, then
$\code{D}$ provides a way of encoding $D$ as a $d$-chain, since $D$
can be recovered straightforwardly from $\code{D}$ by taking the
closure of the support of $\code{D}$. The code of a faithful reconstruction is depicted in {Figure~\ref{figure:curve-code}}.

\begin{figure}[htb]
  \centering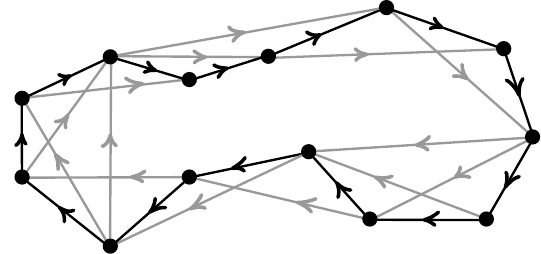
  \caption{A 1-chain of $K$ that encodes a faithful reconstruction of
    curve $\M$. Its coefficients are either 0 (on grey edges) or 1
    (on black edges), assuming the orientation of edges is
    consistent with that of \M. This 1-chain satisfies the
    constraints of Problem~(\ref{problem:reconstruction}).
    \label{figure:curve-code}}
\end{figure}

In this section, we show that, under weak
conditions on $K$, if $D$ is a faithful reconstruction of $\M$, then
$\code{D}$ satisfies the constraints of
Problem~(\ref{problem:reconstruction}).  Indeed, $\Shadow{D}$ being
homeomorphic to $\M$, the generic point $m_0 \in \M$ is covered by the
projection of the convex hull of a unique $d$-simplex $\sigma \in D$
and
\begin{equation*}
  \Load {m_0} \M K {\code{D}} = \code{D}(\sigma) \sign{\sigma}{\M} = 1.
\end{equation*}
The next lemma states conditions under which the constraint $\partial
\code{D} = 0$ is also satisfied.

\begin{lemma}\label{lemma:manifold-as-cycle}
  Let $r, \rho \geq 0$ such that $\rho < \frac{\sqrt{2}}{4}(\reach -
  r)$. Let $K$ be a simplicial complex such that $\Shadow{K} \subseteq
  \Offset \M r$ and whose $d$-simplices are $\rho$-small and have an
  angular deviation smaller than $\frac{\pi}{4}$ relatively to \M. If
  the subcomplex $D \subseteq K$ is a faithful reconstruction of \M,
  then $\code{D}$ is a cycle. 
\end{lemma}

\begin{proof}
  We first prove that for all simplices $\sigma \in D$ and all points $m \in \pi_\M(\Conv \sigma)$, we have that
  \begin{equation}
    \label{eq:projection-simplex-bounded}
    \pi_\M(\Conv \sigma) \subseteq B\left(m, \sin \left(\frac{\pi}{4}\right) \reach\right)^\circ.
  \end{equation}
  Indeed, consider $x,x' \in \Conv \sigma$. Suppose that $m =
  \pi_\M(x)$ and let $m' = \pi_\M(x')$. We know from \cite[page
    435]{federer-59} that for $0 \leq r < \Reach \M$, the projection
  map $\pi_\M$ is $\left( \frac{\reach}{\reach-r} \right)$-Lipschitz for points at distance less than $r$ from
  $\M$. It follows that
  \[
  \| m - m' \| \leq \frac{\reach}{\reach - r} \|x - x'\| \leq \frac{2\rho}{\reach - r}  \reach < \frac{\sqrt{2}}{2} \reach
  \]
  and Inclusion~\eqref{eq:projection-simplex-bounded} follows.

Given a simplicial complex $L$ and a point $x \in \Rspace^\Dim$, we
define the \emph{star} of $x$ in $L$ as the set of simplices
$\Star{x}{L} = \{ \sigma \in L \mid x \in \Conv \sigma \}$.  Since $D$
is a faithful reconstruction of \M, $\Shadow{D}$ is a
$d$-dimensional submanifold. Hence, each $(d-1)$-simplex $\tau \in D$ has exactly two
$d$-cofaces $\sigma_1$ and $\sigma_2$; see Figure \ref{figure:proof-cycle}. Consider a point $x$ in the
relative interior of $\tau$ and its projection $m = \pi_\M(x)$ onto
$\M$. The star of $x$ in $D$ consists of the two $d$-simplices
$\sigma_1$ and $\sigma_2$ and the common $(d-1)$-face $\tau$. It follows
that the set $\pi_{\Tangent {m} \M}(\Star x {D})$ possesses exactly
two $d$-simplices $\sigma_1' = \pi_{\Tangent m \M}(\sigma_1)$ and
$\sigma_2' = \pi_{\Tangent m \M}(\sigma_2)$, and one $(d-1)$-simplex
$\tau' =\pi_{\Tangent m \M}(\tau)$. As we project the $d$-simplex
$\sigma_i = [u_0, \ldots, u_{d}]$, let us preserve the ordering of the vertices,
that is, let $\sigma_i' = [\pi_{\Tangent m \M}(u_0), \pi_{\Tangent
    m \M}(u_1), \ldots, \pi_{\Tangent m \M}(u_d)]$. Let us give to
$\Tangent m \M$ an orientation that is consistent with that of
\M. Inclusion~\eqref{eq:projection-simplex-bounded} allows us to apply
Lemma \ref{lemma:TangentSpaceProjectionIPreservesOrientation} in
Appendix~\ref{appendix:transfering-orientation}:
each $d$-simplex
$\sigma_i'$ has the same orientation with respect to $\Tangent m \M$
than that of $\sigma_i$ with respect to $\M$. Let $s_i =
\sign{\sigma_i'}{\Tangent m \M} = \sign{\sigma_i}{\M}$.

\begin{figure}[htb]
  \centering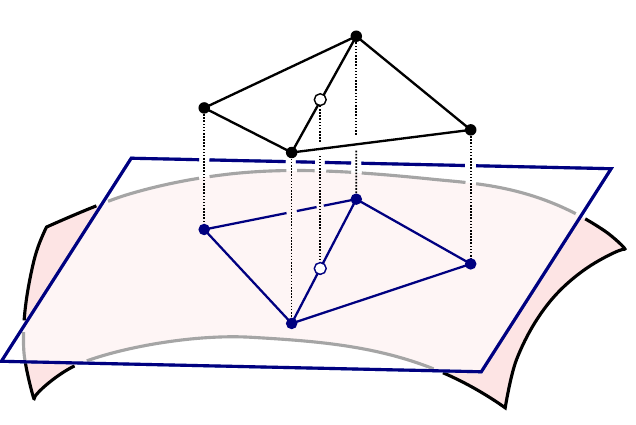
  \caption{Notation for the proof of Lemma \ref{lemma:manifold-as-cycle}. In this example, $s_1 = -1$ and $s_2 = +1$. \label{figure:proof-cycle}}
\end{figure}

We claim that the two geometric $d$-cofaces $\Conv\sigma_1'$ and
$\Conv\sigma_2'$ of $\Conv\tau'$ have disjoint relative interiors.
Indeed, let us denote by $U_x$ an open neighborhood of $x$ in
$\Rspace^\Dim$ and suppose that $U_x$ is sufficiently small so that
its restriction to $\Shadow{K}$ is contained in $\Shadow{\Star x
  {D}}$. For $i=1,2$, let $U^i_x = U_x \cap \Aff \sigma_i$. Note
that the map $\pi_{\Tangent m \M} \circ \pi_\M\restrict{U^i_x}$ is
differentiable and the map $\pi_{\Tangent m \M}\restrict{U^i_x}$ is
affine.  Both maps have equal differential maps at $x$, that is:
\begin{equation}\label{equation:SameDifferentialProjOnTgtPlaneAndThroughProjOnManifold}
\Differential_x \left(\pi_{\Tangent m \M} \circ \pi_\M\restrict{U^i_x}\right)= \Differential_x\left(\pi_{\Tangent m \M}\restrict{U^i_x}\right).
\end{equation}
Let $T_i^+$ denote the set of all vectors parallel to $\Aff \sigma_i$
and pointing inside $\Conv\sigma_i$ after translation at $x$. This set
forms a closed half-space in the vector tangent space to $\Aff
\sigma_i$. Since $\pi_{\Tangent m \M}\restrict{U^i_x}$ is affine, it
coincides, up to a constant, with its differential at $x$ and using
Equation
\eqref{equation:SameDifferentialProjOnTgtPlaneAndThroughProjOnManifold},
we get that
\begin{equation}
  \label{eq:locating-projection-convex-hull}
\Conv\sigma_i' \; \subseteq \;
x + \Differential_x\left(\pi_{\Tangent m\M}\restrict{U^i_x}\right)(T^+_i) \; = \;
x + \Differential_x\left(\pi_{\Tangent m\M} \circ \pi_\M\restrict{U^i_x}\right) (T^+_i).
\end{equation}
Observe that the map $\pi_{\Tangent m \M} \circ \pi_\M\restrict{U^i_x}$,
being the composition of two injective functions, is injective.  It
follows that $\Differential_x \left(\pi_{\Tangent m\M} \circ
\pi_\M\restrict{U^1_x}\right) (T^+_1)$ and $\Differential_x
\left(\pi_{\Tangent m\M} \circ \pi_\M\restrict{U^2_x}\right) (T^+_2)$
are two half-spaces in the vector space $T_m\M$ with disjoint
interiors.  Using Equation \eqref{eq:locating-projection-convex-hull}, we obtain
that $\Conv\sigma_1'$ and $\Conv\sigma_2'$ also have
disjoint interiors, as claimed.

It follows that $\partial (s_1 \sigma_1'+ s_2 \sigma_2')$ is $0$ on $\tau'$, and
consequently $\partial (s_1 \sigma_1+ s_2 \sigma_2)$ is $0$ on
$\tau$. We have shown that $\partial \code{D}=0$.
\end{proof}

\subsection{Geometric conditions on the sample}
\label{section:geometric-conditions}

Recall that our goal is to give conditions under which a solution to
Problem~(\ref{problem:reconstruction}) provides a faithful
reconstruction of \M. To express the conditions that we need, let us
introduce some definitions and notations.

\begin{definition}[Dense sample]
  We say that $P$ is an {\em $\varepsilon$-dense} sample of \M if for
  every point $m \in \M$, there is a point $p \in P$ with $\|p-m\|
  \leq \varepsilon$ or, equivalently, if $\M \subseteq \Offset P
  \varepsilon$.
\end{definition}

\begin{definition}[Accurate sample]
  We say that $P$ is a {\em $\delta$-accurate} sample of \M if for
  every point $p \in P$, there is a point $m \in \M$ with $\|p-m\|
  \leq \delta$ or, equivalently, if $P \subseteq \Offset \M \delta$.
\end{definition}

The {\em separation} of a point set $P$ is
  \[
  \Sep P = \min_{p \neq q \in P} \|p-q\|.
  \]
We recall that the \emph{height} of a simplex $\sigma$ is
\[
\height{\sigma} = \min_{v \in \sigma} d(v,\Aff(\sigma \setminus \{v\})).
\]
The height of $\sigma$ vanishes if
and only if $\sigma$ is degenerate. 
The {\em protection} of a simplex $\sigma$ relatively to a point set $Q$ is
\[
\protection{\sigma}{Q} = \min_{q \in \pi_{\Aff \sigma}(Q \setminus \sigma)} d(q,S(\sigma)).
\]
We stress that our definition of a simplex protection differs slightly
from the one in \cite{BOISSONNAT_2013,boissonnat2018geometric}. We now
associate to a finite point set $P$ and a scale $\rho$ three
quantities that describe the quality of the pair $(P,\rho)$ at dimension $d$:
\begin{align*}
  \Height{P}{\rho} &= \min_{\sigma} \height{\sigma},\\
  \AngularDeviation {P,\rho} &= \max_\sigma \AngularDeviation \sigma,\\
  \ProtectGlob P \rho &= \min_\sigma \protection{\sigma}{P \cap B(c_\sigma,\rho)},
\end{align*}
where the two minima and the maximum are over all $\rho$-small
$d$-simplices $\sigma \subseteq P$. Observe that assuming
$\Height{P}{\rho}>0$ is equivalent to assuming that all $\rho$-small
$d$-simplices of $P$ are non-degenerate.

\begin{definition}[Safety condition]
  Let $\varepsilon$, $\delta$, and $\rho$ be non-negative real
  numbers.  The {\em safety condition} on $(P,\varepsilon,\delta)$ at
  scale $\rho$ is the existence of a real number $\theta \in
  \left[0,\frac{\pi}{6}\right]$ such that:
  \begin{align*}
    \AngularDeviation{P,\rho} \; &\leq \; \frac{\theta}{2} -  \arcsin \left(\frac{\rho+\delta}{\reach}\right), \\ 
    \Sep P \; &> \; 8 (\delta \theta + \rho \theta^2) + 6 \delta + \frac{2\rho^2}{\reach}, \\
    \ProtectGlob{P}{3\rho} \; &> \; 8 (\delta \theta + \rho \theta^2) \left( 1 + \frac{4 d \varepsilon}{\Height{P}{\rho}} \right). 
  \end{align*}
\end{definition}

Roughly speaking, assuming the safety condition on
$(P,\varepsilon,\delta)$ at scale $\rho$ enforces $\rho$-small
$d$-simplices of $P$ to make a sufficiently small angle relatively to
\M. It also enforces $P$ to be both sufficiently separated and
protected at scale $3\rho$. As explained in the companion paper
\cite{AttaliLieutierFlatDelaunay2022}, the safety condition on
$(P,\varepsilon,\delta)$ can be met by considering a
$\left(\frac{20\varepsilon}{21}\right)$-dense
$\left(\frac{\delta}{2}\right)$-accurate point set~$P$ and perturbing
it as described in \cite{AttaliLieutierFlatDelaunay2022}.

\subsection{Main theorem}
\label{section:main-theorem}

In the statement of our main theorem, there is a constant
$\Omega(\Delta_d)$ that depends only upon the dimension $d$ and whose
definition is given in the proof of Lemma \ref{lemma:Wmin-beta}.  Let
$\Cech{P}{r}$ denote the set of simplices of $P$ that are $r$-small,
also known as the \emph{\v Cech complex} of $P$ at scale $r$.

\begin{theorem}[Faithful reconstruction by a variational approach]
  \label{theorem:reconstruction}
  Let $\M$ be a compact orientable $C^2$ $d$-dimensional submanifold
  of $\Rspace^\Dim$ for some $d<\Dim$.  Let $\varepsilon$, $\delta$,
  and $\rho$ be non-negative real numbers such that $\delta \leq
  \varepsilon$ and $16\varepsilon \leq \rho < \frac{\reach}{4}$. Let
  $\Theta = \AngularDeviation{P,\rho}$ and assume that $\Theta \leq
  \frac{\pi}{6}$. Set
  \[
  J = \frac{ (\reach+ \rho)^d  }
  {\left( \reach- \rho \right)^d \:   \left(\cos \Theta  \right)^{\min\{d, N-d\}}} -1.
  \]
  Let $P$ be a $\delta$-accurate $\varepsilon$-dense sample of \M such
  that $\Height{P}{\rho}>0$. Suppose that the safety condition on
  $(P,\varepsilon,\delta)$ is satisfied at scale $\rho$. Suppose furthermore that
  \begin{equation}
    \label{eq:extra-safety-condition}
    \begin{split}
      \ProtectGlob{P}{3\rho}^2 &+ \ProtectGlob{P}{3\rho} \Sep P > \\
      & \max \left\{10 \rho  \Theta (\varepsilon + \rho  \Theta) , \frac{4 J (1+J)}{(d+2)(d-1)! \, \Omega(\Delta_d)} \rho^2 \right\}.
    \end{split}    
  \end{equation}
  Consider a simplicial complex $K$ such that
  \begin{equation}
    \label{eq:complex-condition}
    \Del{P} \cap \Cech P \varepsilon \subseteq K \subseteq \Cech P \rho.    
  \end{equation}
  Then Problem~(\ref{problem:reconstruction}) has a unique
  solution and the closure of the support of that solution is a
  faithful reconstruction of \M.
\end{theorem}

Observe that our main theorem does not require $K$ to be geometrically realized nor
to retain the homotopy type of \M.  One may ask about the feasability
of realizing the assumptions of Theorem~\ref{theorem:reconstruction}.
In Section~\ref{section:PerturbationForSafetyAlgorithm}, we explain how
to apply Moser Tardos Algorithm (\cite{moser2010constructive} and
\cite[Section 5.3.4]{boissonnat2018geometric}) as a perturbation
scheme for enforcing both the safety condition and Condition \eqref{eq:extra-safety-condition} required by
Theorem~\ref{theorem:reconstruction}.

\paragraph{Choosing the simplicial complex $K$.}
Recall that the \v Cech complex of $P$ at scale $r$, denoted as
$\Cech{P}{r}$, is the set of simplices of $P$ that are $r$-small. The
Rips complex of $P$ at scale $r$, denoted as $\Rips{P}{r}$, consists
of all simplices of $P$ with diameter at most $2r$. It is a more
easily-computed version of the \v Cech complex. We stress that our
main theorem applies to any simplicial complex $K$ such that $\Del{P}
\cap \Cech P \varepsilon \subseteq K \subseteq \Cech{P}{\rho}$. Since
$\Cech{P}{r} \subseteq \Rips{P}{r} \subseteq \Cech{P}{\sqrt{2}r}$ for
all $r \geq 0$, it applies in particular to any $K = \Rips{P}{r}$ with
$\varepsilon \leq r \leq \frac{\rho}{\sqrt{2}}$.
This choice of $K$ is well-suited
for applications in high dimensional spaces, while choosing $K = \Del{P} \cap \Cech{P}{r}$ for
any $\varepsilon \leq r \leq \rho$ may be more suited for
applications in low dimensional spaces.

%
%
%
%
%

\section{Proving the main theorem}
\label{section:proof-main-result}

\subsection{Technical lemma}
\label{section:technical-lemma}

The proof of Theorem~\ref{theorem:reconstruction} relies on a
technical lemma which we now state and prove.

\begin{lemma}
  \label{lemma:technical} 
  Let $\D$ be an orientable $d$-dimensional submanifold (with or without boundary) of  $\Rspace^\Dim$ and let $K$ be a
  simplicial complex with vertices in $\Rspace^\Dim$. Assume that there is a continuous function $\varphi: \Shadow{K}
  \to \D$. Suppose that for each
  $d$-simplex $\sigma \in K$, we have two positive weights
  $W(\sigma) \geq W_{\min}(\sigma)$ and that there exists an integrable function $f :
  \D \to \Rspace^+$ such that $W_{\min}(\sigma) = \int_{\varphi(\Conv{\sigma})}
  f$. Consider the $d$-chain $\gamma_{\min}$ on $K$ defined by
  \[
  \gamma_{\min}(\sigma) =
  \begin{cases}
    \sign{\sigma}{\D} \quad \quad &\text{if $W_{\min}(\sigma) = W(\sigma)$,}\\
    0 \quad \quad  &\text{otherwise.}
  \end{cases}
  \] Suppose that
  $\sum_{\sigma \in K^{[d]}} \gamma_{\min}(\sigma) \sign{\sigma}{\D} \Indicator{\varphi(\Conv{\sigma})}(x)=1$, for almost all $x \in \D$. Then,
  $\gamma_{\min}$ is the unique solution to the following optimization problem over the set of chains in $C_d(K,\Rspace)$:
  \begin{tcolorbox}
    \vspace{-4mm}
    \begin{mini*}
      {\gamma}{\|\gamma\|_{1,W}}{}{}
      \label{problem:technical}
      \addConstraint{\sum_{\sigma\in K^{[d]}} \gamma(\sigma) \sign{\sigma}{\D}  \Indicator{\varphi(\Conv{\sigma})}(x)}{=1, \hspace{2mm}}{\text{for almost all } x \in \D}
    \end{mini*}
  \end{tcolorbox}
\end{lemma}

\begin{proof}
  We note that the problem is invariant under change of orientation of
  $d$-simplices in $K$ and thus we may assume that every $d$-simplex
  $\sigma$ in $K$ has an orientation that is consistant with that of
  \D, that is, $\sign{\sigma}{\D}=1$ for all $\sigma \in K^{[d]}$.
  With this assumption, the lemma simply asserts the following.
  Consider the $d$-chain $\gamma_{\min}$ on $K$ defined by
  \[
  \gamma_{\min}(\sigma) =
  \begin{cases}
    1 \quad \quad \text{if $W_{\min}(\sigma) = W(\sigma)$,}\\
    0 \quad \quad \text{otherwise.}
  \end{cases}
  \] Suppose that
  $\sum_{\sigma \in K^{[d]}} \gamma_{\min}(\sigma) \Indicator{\varphi(\Conv{\sigma})}(x)=1$, for almost all $x \in \D$. Then the $\ell_1$-like
  norm $\|\gamma\|_{1,W}$ attains its minimum over all $d$-chains $\gamma$ such
  that
  \begin{equation}
    \label{eq:constraint}
      \sum_{\sigma \in K^{[d]}} \gamma(\sigma) \Indicator{\varphi(\Conv{\sigma})}(x)=1, \quad \text{for almost all } x \in \D
  \end{equation}
  if and only if $\gamma = \gamma_{\min}$.

  We write $\ts = \varphi(\Conv{\sigma})$ throughout the proof for a shorter notation.
  We prove the lemma by showing that for all $d$-chains $\gamma$ on
  $K$ that satisfy constraint~(\ref{eq:constraint}), we have:
  \begin{eqnarray}
    \label{eq:technical1}
    \| \gamma \|_{1,W} \; \geq \;  \| \gamma \|_{1, W_{\min}} \; \geq \; 
    \int_{\D} f \; = \;
    \| \gamma_{\min}
    \|_{1,W_{\min}} \; = \;  \| \gamma_{\min} \|_{1,W},
  \end{eqnarray}
  with the first inequality being an equality if and only if $\gamma =
  \gamma_{\min}$. Clearly, $ \| \gamma \|_{1,W} \geq \| \gamma \|_{1, W_{\min}}$
  because $W(\sigma) \geq W_{\min}(\sigma)$. To obtain the second
  inequality, recall that we have assumed $\sum_{\sigma} \gamma(\sigma)
  \Indicator{\ts}(x)=1$ almost everywhere in \D. We use this to write that:
  \begin{eqnarray}
    \label{eq:technical2}
    \|\gamma\|_{1, W_{\min}}
    \; \geq \;  \sum_{\sigma} \gamma(\sigma) \int_{\ts} f
    \; = \; \sum_{\sigma} \gamma(\sigma) \int_{\D} f \Indicator{\ts}
    \; = \; \int_{\D} f \sum_{\sigma} \gamma(\sigma) \Indicator{\ts}
    \; = \; \int_{\D} f,
  \end{eqnarray}
  where sums are over all $d$-simplices $\sigma$ in $K$. Setting
  $\gamma = \gamma_{\min}$ in (\ref{eq:technical2}), we observe that
  the inequality in (\ref{eq:technical2}) becomes an equality because
  none of the coefficients of $\gamma_{\min}$ are negative by
  construction. It follows that $\int_{\D} f = \| \gamma_{\min}
  \|_{1,W_{\min}}$. Finally, $\| \gamma_{\min} \|_{1,W_{\min}} \; = \;
  \| \gamma_{\min} \|_{1,W}$ because $\gamma_{\min}$ has been defined
  so that for all simplices $\sigma$ in its support, $W_{\min}(\sigma)
  = W(\sigma)$. We have thus established (\ref{eq:technical1}). Suppose
  now that $\gamma \neq \gamma_{\min}$ and let us prove that $ \|
  \gamma \|_{1,W} > \| \gamma \|_{1, W_{\min}}$, or equivalently that
  \[
  \sum_{\sigma \in \Support{\gamma}} | \gamma(\sigma) | \; ( W(\sigma) -
  W_{\min}(\sigma)) > 0.
  \] Since none of the terms in the above sum are negative, it
  suffices to show that there exists at least one simplex $\sigma \in
  \Support{\gamma}$ for which $W(\sigma) > W_{\min}(\sigma)$. By
  contradiction, assume that for all $\sigma \in \Support{\gamma}$,
  $W(\sigma) = W_{\min}(\sigma)$. By construction, we thus have the
  implication: $\gamma(\sigma) \neq 0 \implies
  \gamma_{\min}(\sigma) = 1$, and therefore
  $\Support{\gamma} \subseteq \Support{\gamma_{\min}}$. But, since
  $\sum_{\sigma} \gamma_{\min}(\sigma) \Indicator{\ts}(x)=1$ for
  almost all $x \in \D$ and coefficients of $\gamma_{\min}$ are either
  $0$ or $1$, it follows that for almost all $x \in \D$, point $x$ is
  covered by a unique $d$-simplex in the support of $\gamma_{\min}$.
  Hence, the simplices in $\Support{\gamma_{\min}}$ have pairwise disjoint
  interiors while their union covers \D. Since $\sum_{\sigma}
  \gamma(\sigma) \Indicator{\ts}(x)=1$ for almost all $x \in \D$,
  the simplices in $\Support{\gamma}$ must also cover $\D$ while using
  only a subset of simplices in $\Support{\gamma_{\min}}$. The only
  possibility is that $\gamma = \gamma_{\min}$, yielding a
  contradiction.
\end{proof}

\subsection{Overview of the proof}
\label{section:proof-overview}

We first illustrate the use of the technical lemma by establishing a
simple variant of Theorem~\ref{theorem:reconstruction} in which $\M$
has no curvature. We then pinpoint what has to be modified in the
proof to establish Theorem~\ref{theorem:reconstruction}.

\paragraph{Euclidean setting.}

\begin{theorem}
  \label{theorem:euclidean}
  Let $P$ be a finite point set of $\Rspace^\Dim$ in general position
  and let $K$ be a simplicial complex with vertex set $P$ such that
  $\Del{P} \subseteq K$. Let $\M = \Conv P$ and let $d = \dim(\Aff
  P)$. Then, $\code{\Del P}$ is the unique solution to the following
  optimization problem over the set of chains in $C_d(K,\Rspace)$:
  \begin{tcolorbox}
    \vspace{-4mm}
    \begin{mini*}
      {\gamma}{\Edel(\gamma)}{}{}
      \addConstraint{\sum_{\sigma\in K^{[d]}} \gamma(\sigma) \sign{\sigma}{\M}  \Indicator{\Conv{\sigma}}(x)}{=1, \hspace{2mm}}{\text{for almost all } x \in \M}
    \end{mini*}
  \end{tcolorbox}
\end{theorem}

\begin{proof}
  The proof consists in applying the technical lemma (Lemma
  \ref{lemma:technical}). In preparation for this, we make the
  following definitions.  Let $\D = \M = \Conv P =
  \Shadow{K}$. Clearly, $\D$ is an orientable $d$-dimensional submanifold (with
  boundary).
  Let $\varphi$ be the identity map of $\D$ and let $f: \D \to \Rspace^+$ be the map
  defined by:
  \[
  f(x) = \min_{\sigma}  \left(- \Power{\sigma}{x}\right),
  \]
  where the minimum is taken over all $d$-simplices $\sigma \in K$
  such that $x \in \Conv\sigma$. Finally, for any $\sigma \in K$, we
  let $W(\sigma) = \omega(\sigma)$ be the Delaunay weight of $\sigma$
  and define the weight:
  \begin{align}
    W_{\min}(\sigma) &= \int_{x \in \Conv\sigma} f(x) \, dx.
  \end{align}
  By construction, $W(\sigma) \geq
    W_{\min}(\sigma)$. Because $P$ is in general position, all
    $d$-simplices of $K$ are non-degenerate and $W_{\min}(\sigma)>0$.
  Figure \ref{figure:weights-in-euclidean-setting} depicts the two
  weights associated to a simplex $\beta$.  Consider the $d$-chain
  $\gamma_{\min}$ on $K$:
  \[
  \gamma_{\min}(\sigma) =
  \begin{cases}
    \sign{\sigma}{\D} \quad \quad &\text{if $W_{\min}(\sigma) = W(\sigma)$,}\\
    0 \quad \quad &\text{otherwise.}
  \end{cases}
  \]
  Before applying the technical lemma, we make three observations, one
  per step.

  \Step{1} For any $d$-simplices $\alpha \in \Del{P}$ and $\beta \in K
  \setminus \Del{P}$ and for any $x \in \Conv \alpha \cap \Conv
  \beta$, we have that $-\Power{\alpha}{x} \leq -\Power{\beta}{x}$ and
  the inequality is strict whenever $x \not \in \Conv{(\alpha \cap
    \beta)}$ as illustrated in Figure \ref{figure:weights-in-euclidean-setting}.

  \begin{figure}[htb]
    \centering
\begingroup%
  \makeatletter%
  \providecommand\color[2][]{%
    \errmessage{(Inkscape) Color is used for the text in Inkscape, but the package 'color.sty' is not loaded}%
    \renewcommand\color[2][]{}%
  }%
  \providecommand\transparent[1]{%
    \errmessage{(Inkscape) Transparency is used (non-zero) for the text in Inkscape, but the package 'transparent.sty' is not loaded}%
    \renewcommand\transparent[1]{}%
  }%
  \providecommand\rotatebox[2]{#2}%
  \newcommand*\fsize{\dimexpr\f@size pt\relax}%
  \newcommand*\lineheight[1]{\fontsize{\fsize}{#1\fsize}\selectfont}%
  \ifx\svgwidth\undefined%
    \setlength{\unitlength}{367.75405523bp}%
    \ifx\svgscale\undefined%
      \relax%
    \else%
      \setlength{\unitlength}{\unitlength * \real{\svgscale}}%
    \fi%
  \else%
    \setlength{\unitlength}{\svgwidth}%
  \fi%
  \global\let\svgwidth\undefined%
  \global\let\svgscale\undefined%
  \makeatother%
  \begin{picture}(1,0.44140333)%
    \lineheight{1}%
    \setlength\tabcolsep{0pt}%
    \put(0.65160665,0.03631879){\color[rgb]{1,0,0}\makebox(0,0)[lt]{\lineheight{1.25}\smash{\begin{tabular}[t]{l}$x$\end{tabular}}}}%
    \put(0.7083775,0.24292531){\makebox(0,0)[lt]{\lineheight{1.25}\smash{\begin{tabular}[t]{l}$-\Power{\beta}{x}$\end{tabular}}}}%
    \put(0.72697451,0.15139039){\makebox(0,0)[lt]{\lineheight{1.25}\smash{\begin{tabular}[t]{l}$-\Power{\alpha}{x}$\end{tabular}}}}%
    \put(0,0){\includegraphics[width=\unitlength,page=1]{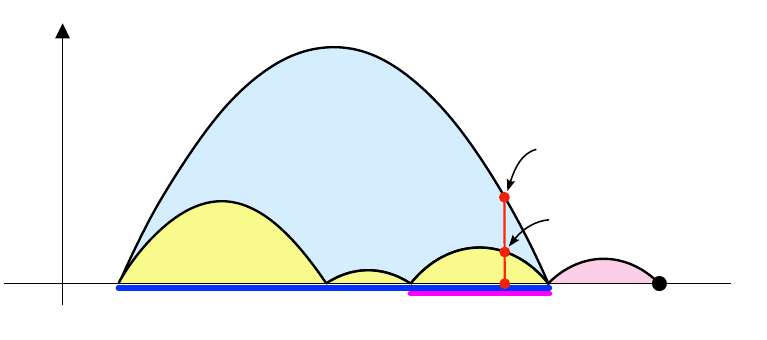}}%
    \put(0.07320115,0.42280968){\makebox(0,0)[lt]{\lineheight{1.25}\smash{\begin{tabular}[t]{l}$\Rspace$\end{tabular}}}}%
    \put(-0.00203143,0.08458691){\makebox(0,0)[lt]{\lineheight{1.25}\smash{\begin{tabular}[t]{l}$\Aff P$\end{tabular}}}}%
    \put(0,0){\includegraphics[width=\unitlength,page=2]{weights-in-euclidean-setting.pdf}}%
    \put(0.38775556,0.03087187){\color[rgb]{0.01568627,0.2,1}\makebox(0,0)[lt]{\lineheight{1.25}\smash{\begin{tabular}[t]{l}$\beta$\end{tabular}}}}%
    \put(0.57600701,0.03236558){\color[rgb]{1,0,1}\makebox(0,0)[lt]{\lineheight{1.25}\smash{\begin{tabular}[t]{l}$\alpha$\end{tabular}}}}%
  \end{picture}%
\endgroup%

    \caption{The two weights $W(\beta)$ and $W_{\min}(\beta)$ defined
      in the proof of Theorem~\ref{theorem:euclidean} can be depicted
      as the volume of the blue region and the volume of the the yellow region,
      respectively. The blue region is the part of the subgraph of
      $-\operatorname{Power}_{\beta}$ lying above $\Conv \beta$ and
      the yellow region is the part of the subgraph of $f$ lying above
      $\Conv\beta$. \label{figure:weights-in-euclidean-setting}}
  \end{figure}

  \Step{2} For every point $x \in \D$, we thus have that
  $f(x) = - \Power{\alpha}{x}$, where $\alpha$ is any $d$-simplex of
  $\Del{P}$ whose convex hull contains $x$.

  \Step{3} For all simplices $\sigma \in K$, the following property
  holds: $W_{\min}(\sigma)=W(\sigma)$ if and only if $\sigma$ is a
  Delaunay $d$-simplex of $P$.

  \medskip
    \noindent Hence, $\gamma_{\min} = \code{\Del{P}}$ and applying
    Lemma \ref{lemma:technical} yields the result.
\end{proof}

\paragraph{Adapting the proof to the submanifold setting.}

In the submanifold setting, we let $\M$ be an orientable $C^2$
$d$-dimensional submanifold of $\Rspace^N$ and let $P$ be a finite
$\delta$-accurate $\varepsilon$-dense sample of $\M$. We also consider
a simplicial complex $K$ with vertex set $P$ and let $\rho \geq 0$ be
a scale parameter. Assuming that $\M$, $P$ $\varepsilon$, $\delta$ and
$\rho$ satisfy the assumptions of
{Theorem~\ref{theorem:reconstruction}}, we now give an overview of our
proof of Theorem~\ref{theorem:reconstruction}.

The proof consists in applying the technical lemma, following for this
the same steps as in the proof of
Theorem~\ref{theorem:euclidean}. However, the different steps are now
more involved as is the definition of the various objects required to
apply the technical lemma.  First of all, we introduce in
Section~\ref{section:Delloc-complex} a simplicial complex, called the
{\em Delloc complex} of $P$ at scale $\rho$ and denoted as $\Delloc P
\rho d$. This complex is going to act as the counterpart of the Delaunay
complex of $P$ in the Euclidean setting. In particular, we aim at
showing that the unique solution to
Problem~(\ref{problem:reconstruction}) is $\code{\Delloc P \rho d}$
(instead of $\code{\Del P}$ in the Euclidean setting).

We first show in Section~\ref{section:Delloc-complex} that, under the
assumptions of Theorem~\ref{theorem:reconstruction}, $\Shadow{\Delloc
  P \rho d}$ is a faithful reconstruction of \M
(Theorem~\ref{theorem:homeomorphism-from-sampling-conditions}) and
$\Delloc{P}{\rho}{d} \subseteq K$
(Remark~\ref{remark:contained-in-Delaunay}). The set $\D =
\Shadow{\Delloc P \rho d}$ is thus an orientable $d$-dimensional
submanifold of $\Rspace^\Dim$, which is going to play the role of
$\Conv P$ in the Euclidean setting. Next we define $\varphi:
\Shadow{K} \to \D$ which maps any $y \in \Shadow{K}$ to the point $x
\in \D$ such that $\pi_\M(x) = \pi_\M(y)$. We then define the map $f :
\D \to \Rspace^+$ by
\[
f(x) = \min_\sigma  \left(- \Power{\sigma}{y}\right),
\]
where the minimum is taken over all $d$-simplices $\sigma \in K$ and all points $y \in \Conv \sigma$ such that $\pi_\M(x) = \pi_\M(y)$.
Finally, for any $\sigma \in K$, we
  let $W(\sigma) = \omega(\sigma)$ be the Delaunay weight of $\sigma$
  and define the weight:
  \begin{align}
    W_{\min}(\sigma) &= \int_{x \in \varphi(\Conv\sigma)} f(x) \, dx.
  \end{align}
  By construction, $W(\sigma) \geq W_{\min}(\sigma)$ and because every
  $d$-simplex $\sigma \in K$ is non-degenerate,
  $W_{\min}(\sigma)>0$. Figure
  \ref{figure:weights-in-submanifold-setting} depicts the two weights
  associated to a simplex $\beta$.

\begin{figure}[htb]
  \centering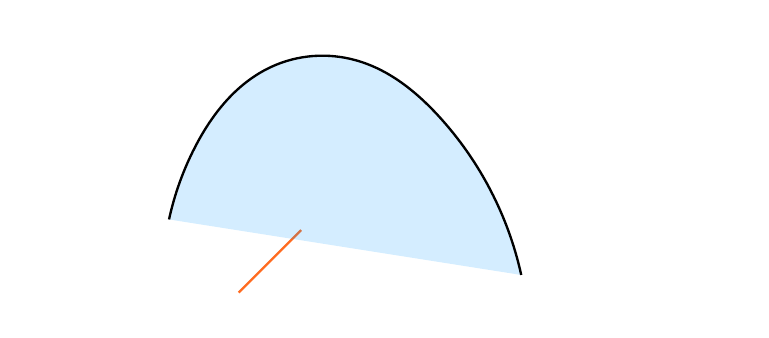
  \caption{The two weights $W(\beta)$ and $W_{\min}(\beta)$ used in
    the proof of Theorem \ref{theorem:reconstruction} can be
    depicted as the volume of the blue region and the volume of the yellow region,
    respectively. The complex $\Delloc P \rho d$ is formed of the five
    black dots, the brown segment and the three green segments. The
    green polygonal chain represents $\varphi(\Conv \beta)$.
    \label{figure:weights-in-submanifold-setting}}
\end{figure}

The first step in the proof of Theorem \ref{theorem:euclidean} shows
that $-\Power{\alpha}{x} \leq -\Power{\beta}{x}$ for any $d$-simplices
$\alpha \in \Del{P}$ and $\beta \subseteq P$ and for any $x \in \Conv
\alpha \cap \Conv \beta$.  One difficulty in the submanifold setting
is that now $d$-simplices of $K$ do not lie anymore in the same
$d$-dimensional affine space as illustrated in
Figure~\ref{figure:weights-in-submanifold-setting}. Nonetheless, we
can still compare $\Power \alpha x$ and $\Power \beta y$ for two
$d$-simplices $\alpha \in \Delloc P \rho d$ and $\beta \subseteq P$,
and for two points $x \in \Conv \alpha$ and $y \in \Conv \beta$,
assuming that they share the same projection onto $\M$, {\em i.e.}
assuming that $\pi_\M(x) = \pi_\M(y)$. This is the goal of
Section~\ref{section:ComparingPowerDistances}.

In a second step, we establish that for every point $x \in \D$, we
have that $f(x) = - \Power \alpha x$, where $\alpha$ is any
$d$-simplex of $\Delloc P \rho d$ whose convex hull contains $x$
(Lemma~\ref{lemma:Wmin-alpha} in Section~\ref{section:final}).

In a third step, we check that we have defined the two weight
functions $W$ and $W_{\min}$ in such a way that $W(\sigma) =
W_{\min}(\sigma)$ if and only if $\sigma$ belongs to the Delloc
complex of $P$. This is done in Section~\ref{section:final} thanks to
Lemmas~\ref{lemma:Wmin-alpha} and \ref{lemma:Wmin-beta}. The tricky
part consists in showing that $W_{\min}(\beta) < W(\beta)$ whenever
$\beta$ is not in the Delloc complex of $P$. Indeed, to compare the
two quantities, we make a change of variable which possibly can
jeopardize the strict inequality but which we are able to
counterbalance, thanks to the conditions we are assuming on $P$ (and
in particular a sufficient protection of $P$).

Finally, all the elements are put together in Section \ref{section:final} where one
can find the proof of the main theorem.

%
%
%
%
%

\subsection{Delloc complexes}
\label{section:Delloc-complex}

In this section, we define the Delloc complex. We then recall a key
result established in the companion paper
\cite{AttaliLieutierFlatDelaunay2022}: when the Delloc complex is
computed over a finite point set $P$ that samples some $d$-dimensional
submanifold of $\Rspace^\Dim$, it provides a faithful reconstruction
of that submanifold. Incidentally, under the right assumptions, the
Delloc complex coincides with the flat Delaunay complex
\cite{AttaliLieutierFlatDelaunay2022} and the tangential Delaunay
complex \cite{boissonnat2014manifold,boissonnat2018geometric}. Since
all the results in this paper are based on the property for a simplex
to belong to the Delloc complex, we find it more enlightening to
formulate the results of this paper using the Delloc complex.

\begin{figure}[htb]
  \centering
\begingroup%
  \makeatletter%
  \providecommand\color[2][]{%
    \errmessage{(Inkscape) Color is used for the text in Inkscape, but the package 'color.sty' is not loaded}%
    \renewcommand\color[2][]{}%
  }%
  \providecommand\transparent[1]{%
    \errmessage{(Inkscape) Transparency is used (non-zero) for the text in Inkscape, but the package 'transparent.sty' is not loaded}%
    \renewcommand\transparent[1]{}%
  }%
  \providecommand\rotatebox[2]{#2}%
  \newcommand*\fsize{\dimexpr\f@size pt\relax}%
  \newcommand*\lineheight[1]{\fontsize{\fsize}{#1\fsize}\selectfont}%
  \ifx\svgwidth\undefined%
    \setlength{\unitlength}{280.40221447bp}%
    \ifx\svgscale\undefined%
      \relax%
    \else%
      \setlength{\unitlength}{\unitlength * \real{\svgscale}}%
    \fi%
  \else%
    \setlength{\unitlength}{\svgwidth}%
  \fi%
  \global\let\svgwidth\undefined%
  \global\let\svgscale\undefined%
  \makeatother%
  \begin{picture}(1,0.58166534)%
    \lineheight{1}%
    \setlength\tabcolsep{0pt}%
    \put(0,0){\includegraphics[width=\unitlength,page=1]{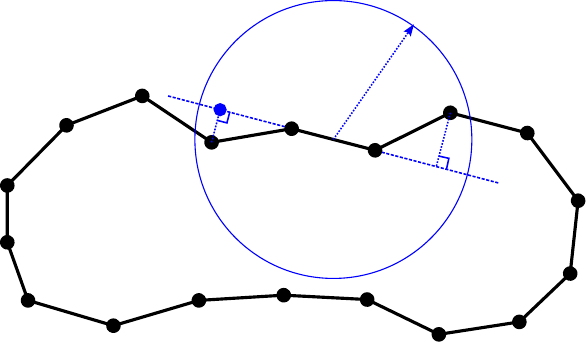}}%
    \put(0.49371723,0.3835589){\makebox(0,0)[lt]{\lineheight{1.25}\smash{\begin{tabular}[t]{l}$a$\end{tabular}}}}%
    \put(0.6336196,0.347204){\makebox(0,0)[lt]{\lineheight{1.25}\smash{\begin{tabular}[t]{l}$b$\end{tabular}}}}%
    \put(0.55059336,0.30172593){\color[rgb]{0,0,1}\makebox(0,0)[lt]{\lineheight{1.25}\smash{\begin{tabular}[t]{l}$c_{ab}$\end{tabular}}}}%
    \put(0.37518303,0.41751067){\color[rgb]{0,0,1}\makebox(0,0)[lt]{\lineheight{1.25}\smash{\begin{tabular}[t]{l}$x$\end{tabular}}}}%
    \put(0.73211069,0.25626046){\color[rgb]{0,0,1}\makebox(0,0)[lt]{\lineheight{1.25}\smash{\begin{tabular}[t]{l}$y$\end{tabular}}}}%
    \put(0,0){\includegraphics[width=\unitlength,page=2]{Delloc-complex-1d.pdf}}%
    \put(0.71984609,0.54495933){\color[rgb]{0,0,1}\makebox(0,0)[lt]{\lineheight{1.25}\smash{\begin{tabular}[t]{l}$\rho$\end{tabular}}}}%
    \put(0,0){\includegraphics[width=\unitlength,page=3]{Delloc-complex-1d.pdf}}%
  \end{picture}%
\endgroup%

  \caption{A set of black dots and its $1$-dimensional Delloc complex
    at scale $\rho$. The edge $ab$ belongs to the Delloc complex
    because ab is a Delaunay edge of the four points $\{a,b,x,y\}$
    obtained by projecting the four black dots inside $B(c_{ab},\rho)$
    onto the line passing through $a$ and $b$.}
  \label{figure:Delloc-complex}
\end{figure}

\paragraph{Definition.}
Afterwards, $P$ designates a finite set of points in $\Rspace^\Dim$,
$d$ designates an integer in $[0,\Dim)$ and $\rho \geq 0$ designates a scale
  parameter.

\begin{definition}[Delloc complex]
  We say that a simplex $\sigma$ is \emph{delloc} in $P$ at
  scale $\rho$ if
  \[
  \sigma \in \Del{\pi_{\Aff \sigma}(P \cap B(c_\sigma,\rho))},
  \]
  where $c_\sigma$ denotes the center of the smallest $N$-ball
  enclosing $\sigma$.  The $d$-dimensional {\em Delloc complex} of $P$
  at scale $\rho$, denoted by $\Delloc P \rho d$, is the set of
  $d$-simplices that are delloc in $P$ at scale $\rho$ together with
  all their faces; see Figure \ref{figure:Delloc-complex}.
\end{definition}

\begin{remark} \label{remark:Gabriel}
  It is easy to see that if $2R(\sigma)\leq \rho$, then the smallest
  circumsphere $S(\sigma)$ of $\sigma$ is contained in
  $B(c_\sigma,\rho)$. It follows that a delloc simplex $\sigma$ in $P$
  at scale $\rho$ is also a \emph{Gabriel simplex} of $P$, by which we
  mean that $S(\sigma)$ does not enclose any point of $P$ in its
  interior. In particular, when $2R(\sigma) \leq \rho$, then $\sigma$ is a
  Delaunay simplex of $P$.
\end{remark}

\paragraph{Key result.} We now recall a key result established in the companion paper
\cite{AttaliLieutierFlatDelaunay2022} and which gives condition under
which the Delloc complex of $P$ is a faithful reconstruction of \M.

\begin{theorem}[Faithful reconstruction by a geometric approach]
  \label{theorem:homeomorphism-from-sampling-conditions}
  Let $\varepsilon$, $\delta$, and $\rho$ be non-negative real numbers
  such that $\delta \leq \varepsilon$ and $16\varepsilon \leq \rho <
  \frac{\reach}{4}$. Let $P$ be a $\delta$-accurate
  $\varepsilon$-dense sample of \M such that
  $\Height{P}{\rho}>0$. Suppose that the safety condition on
  $(P,\varepsilon,\delta)$ is satisfied at scale $\rho$.  Then,
  $\Delloc P \rho d$ is a faithful reconstruction of \M. Furthermore,
  for all $d$-simplices $\sigma \in \Delloc P \rho d$, we have
  $R(\sigma) \leq \varepsilon$.
\end{theorem}

\begin{remark}
  \label{remark:contained-in-Delaunay}
  Under the assumptions of
  Theorem~\ref{theorem:homeomorphism-from-sampling-conditions},
  Remark~\ref{remark:Gabriel} implies that the Delloc complex of $P$
  at scale $\rho$ is a subset of the Delaunay complex of $P$, that is,
  $\Delloc P \rho d \subseteq \Del{P}$. It follows that under the
  assumptions of Theorem
  \ref{theorem:homeomorphism-from-sampling-conditions}:
  \[
  \Delloc P \rho d  \subseteq \Del{P} \cap \Cech P \varepsilon
  \]
  and therefore any simplicial complex $K$ that satisfies the
  assumptions of Theorem~\ref{theorem:reconstruction} contains
  $\Delloc P \rho d$.
\end{remark}

\subsection{Comparing power distances}
\label{section:ComparingPowerDistances}

The goal of this section is to relate the two maps $\Power \alpha x$
and $\Power \beta y$ for two $d$-simplices $\alpha \in \Delloc P \rho
d$ and $\beta \subseteq P$, and for two points $x \in \Conv \alpha$
and $y \in \Conv \beta$, such that $\pi_\M(x) = \pi_\M(y)$. The main
result of the section is stated in the following lemma and proved at
the end of the section.  We recall that given a non-degenerate simplex
$\alpha$ and a point $x \in \Aff \sigma$, the (normalized) barycentric
coordinates of $x$ relatively to the simplex $\alpha$ are real numbers
$\{\lambda_a\}_{a \in \alpha}$ such that $x = \sum_{a \in \alpha}
\lambda_a a$ and $\sum_{a \in \alpha} \lambda_a = 1$. We write
\[
\BaryCoord x \alpha a = \lambda_a
\]

\begin{lemma}
  \label{lemma:final-bound-on-power}
  Let $\varepsilon, \delta, \rho \geq 0$ such that $0 \leq
  2\varepsilon \leq \rho$, and $16 \delta \leq \rho \leq
  \frac{\reach}{3}$. Suppose that $P \subseteq \Offset \M \delta$. Let
  $\pro = \ProtectGlob P {3\rho}$, $\sepvar = \Sep{P}$, and $\Theta =
  \AngularDeviation{P,\rho}$. Assume that $\Theta \leq \frac{\pi}{6}$ and
  \[
  10 \rho\, \Theta  \left( \varepsilon + \rho\, \Theta \right) < \pro^2 + \pro \sepvar.
  \]
  Then, for every non-degenerate $\varepsilon$-small $d$-simplex $\alpha \in \Delloc P \rho d$, every
  non-degenerate $\rho$-small $d$-simplex $\beta \subseteq P$, every $x \in \Conv \alpha$, and every $y
  \in \Conv \beta$ such that $\pi_\M(x) = \pi_\M(y)$:
  \[
  - \Power{\beta}{y} + \Power{\alpha}{x} \geq \frac{1}{2} \left(\pro^2+\pro \sepvar \right)
  \, \sum_{b \in \beta \setminus \alpha} \BaryCoord y \beta b.
  \]
\end{lemma}

\begin{figure}[htb]
  \centering
\begingroup%
  \makeatletter%
  \providecommand\color[2][]{%
    \errmessage{(Inkscape) Color is used for the text in Inkscape, but the package 'color.sty' is not loaded}%
    \renewcommand\color[2][]{}%
  }%
  \providecommand\transparent[1]{%
    \errmessage{(Inkscape) Transparency is used (non-zero) for the text in Inkscape, but the package 'transparent.sty' is not loaded}%
    \renewcommand\transparent[1]{}%
  }%
  \providecommand\rotatebox[2]{#2}%
  \newcommand*\fsize{\dimexpr\f@size pt\relax}%
  \newcommand*\lineheight[1]{\fontsize{\fsize}{#1\fsize}\selectfont}%
  \ifx\svgwidth\undefined%
    \setlength{\unitlength}{205.4530731bp}%
    \ifx\svgscale\undefined%
      \relax%
    \else%
      \setlength{\unitlength}{\unitlength * \real{\svgscale}}%
    \fi%
  \else%
    \setlength{\unitlength}{\svgwidth}%
  \fi%
  \global\let\svgwidth\undefined%
  \global\let\svgscale\undefined%
  \makeatother%
  \begin{picture}(1,0.60454553)%
    \lineheight{1}%
    \setlength\tabcolsep{0pt}%
    \put(0,0){\includegraphics[width=\unitlength,page=1]{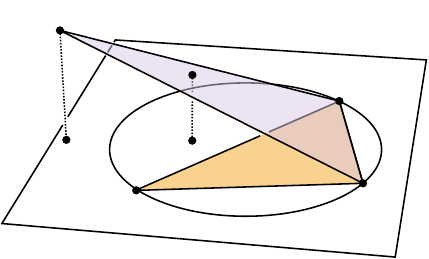}}%
    \put(0.13652484,0.55949999){\makebox(0,0)[lt]{\lineheight{1.25}\smash{\begin{tabular}[t]{l}$b_2 \in \beta \setminus \alpha$\end{tabular}}}}%
    \put(0.77718267,0.40160682){\makebox(0,0)[lt]{\lineheight{1.25}\smash{\begin{tabular}[t]{l}$b_0 \in \alpha \cap \beta$\end{tabular}}}}%
    \put(0.03210333,0.09345754){\makebox(0,0)[lt]{\lineheight{1.25}\smash{\begin{tabular}[t]{l}$\Aff \alpha$\end{tabular}}}}%
    \put(0.53845872,0.19976665){\makebox(0,0)[lt]{\lineheight{1.25}\smash{\begin{tabular}[t]{l}$\alpha$\end{tabular}}}}%
    \put(0.40149788,0.42815431){\makebox(0,0)[lt]{\lineheight{1.25}\smash{\begin{tabular}[t]{l}$y$\end{tabular}}}}%
    \put(0.57974912,0.35161283){\makebox(0,0)[lt]{\lineheight{1.25}\smash{\begin{tabular}[t]{l}$\beta$\end{tabular}}}}%
    \put(0.86166906,0.13724405){\makebox(0,0)[lt]{\lineheight{1.25}\smash{\begin{tabular}[t]{l}$b_1 \in \alpha \cap \beta$\end{tabular}}}}%
    \put(0.25837013,0.14430126){\makebox(0,0)[lt]{\lineheight{1.25}\smash{\begin{tabular}[t]{l}$a$\end{tabular}}}}%
  \end{picture}%
\endgroup%

  \caption{Notation for the proof of Lemma \ref{lemma:basic-bound-on-power}. \label{figure:basic-bound-on-power}}
\end{figure}

\begin{lemma}
  \label{lemma:basic-bound-on-power}
  Let $\alpha$ and $\beta$ be two non-degenerate abstract
  $d$-simplices in $\Rspace^N$ such that $\alpha \in \Del{\pi_{\Aff
      \alpha}(\alpha \cup \beta)}$. Let $\pro = \protection{\alpha}{\beta}$.
  Then for every $y \in \Conv \beta$, we have
  \[
  \Power{\beta}{y} \leq \Power{\alpha}{\pi_{\Aff\alpha}(y)} - (\pro^2 + 2 \pro R(\alpha)) \sum_{b \in \beta \setminus \alpha} \BaryCoord y \beta b.
  \]
\end{lemma}

\begin{proof}
  See Figure \ref{figure:basic-bound-on-power}. Let $Z(\alpha)$ be the
  radius of the $(d-1)$-dimensional circumsphere of $\alpha$. Clearly,
  $\|a-Z(\alpha)\| = R(\alpha)$ for all $a \in \alpha$. Since $\alpha \in \Del{\pi_{\Aff \alpha}(\alpha \cup \beta)}$
  and $\pro = \protection{\alpha}{\beta}$, we get:
  \begin{align*}
    (R(\alpha) + \pro)^2 &\leq \| \pi_{\Aff\alpha}(b) - Z(\alpha)\|^2,
     \qquad \text{for all $b \in \beta\setminus \alpha$},\\
    R(\alpha)^2 &= \|\pi_{\Aff\alpha}(b)-Z(\alpha)\|^2, \qquad \text{for all $b \in \beta \cap \alpha$}.      
  \end{align*}
  Let $\mu_b = \BaryCoord y \beta b$ and note that $\mu_b \geq 0$. Multiplying both sides of each equation above by $\mu_b$ and summing over all $b \in \beta$, we obtain:
  \begin{equation}
  \label{eq:proof-bounding-power}
  R(\alpha)^2 +  (\pro^2 + 2\pro R(\alpha))\sum_{b \in \beta \setminus \alpha}\mu_b
  \leq  \sum_{b \in \beta} \mu_b \| \pi_{\Aff\alpha}(b) - Z(\alpha)\|^2.    
  \end{equation}
  For short, write $y' = \pi_{\Aff\alpha}(y)$ and $\beta' =
  \pi_{\Aff\alpha}(\beta)$.  Noting that $y' = \sum_{b \in \beta}
  \mu_b b'$ and applying Lemma
  \ref{lemma:power-expression-with-additional-point} with $z = Z(\alpha)$, we get that
  \[
  \Power{\beta'}{y'} = \| y' - Z(\alpha) \|^2 - \sum_{b \in \beta} \mu_b \| \pi_{\Aff \alpha}(b) - Z(\alpha) \|^2.
  \]
  Substracting $\|y'-Z(\alpha)\|^2$ from both sides of (\ref{eq:proof-bounding-power}) and using the
  above expression, we obtain
  \begin{align*}
    -\Power{\alpha}{y'} +  (\pro^2 + 2 \pro R(\alpha))\sum_{b \in \beta \setminus \alpha}\mu_b
    &\leq  - \Power{\beta'}{y'}.
  \end{align*}
  Applying Lemma \ref{lemma:power-expression-with-additional-point} again,  with $Z = y'$ and $Z=y$ respectively, we get that:
  \[
    - \Power{\beta'}{y'} = \sum_{b \in \beta} \mu_b \| \pi_{\Aff \alpha}(b) - \pi_{\Aff \alpha}(y) \|^2
    \leq  \sum_{b \in \beta} \mu_b \| b - y \|^2 = - \Power{\beta}{y},
  \]
  which concludes the proof.
\end{proof}

\begin{lemma}
  \label{lemma:advanced-bound-on-power}
  Let $\alpha$ and $\beta$ be two non-degenerate abstract
  $d$-simplices in $\Rspace^\Dim$ such that $\alpha \in \Del{\pi_{\Aff
      \alpha}(\alpha \cup \beta)}$. Let $\pro =
  \protection{\alpha}{\beta}$ and $\Theta =
  \AngularDeviation{\alpha}$.  Suppose that both $\Conv\alpha$ and
  $\Conv\beta$ are contained in the
  $\left(\frac{\rho}{4}\right)$-tubular neighborhood of \M with $\rho
  < 4 \reach$. Suppose furthermore that $\beta$ is $\rho$-small. If
  $\Theta \leq \frac{\pi}{6}$ and
  \[
  5\rho \sin(\Theta) \left( 2 R(\alpha) + \frac{\rho}{2} \sin(\Theta) \right)
    < \pro^2 + 2 \pro R(\alpha),
  \]
  then for every $x \in \Conv\alpha$ and every
  $y \in \Conv\beta$ with $\pi_\M(x) = \pi_\M(y)$, we have
  \[
  \Power{\beta}{y} \leq \Power{\alpha}{x} - \frac{1}{2} (\pro^2 + 2 \pro R(\alpha)) \sum_{b \in \beta \setminus \alpha} \BaryCoord y \beta b.
  \]
\end{lemma}

\begin{proof}
  Consider a point $x \in \Conv\alpha$ and a point $y \in \Conv\beta$
  with $\pi_\M(x) = \pi_\M(y)$. We distinguish two cases depending on whether $y$ belongs to $\Conv{(\alpha \cap \beta)}$ or not.

  \medskip \noindent First, assume that $y \in \Conv{(\alpha \cap
    \beta)}$. In that case, we claim that the only possibility is that
  $x = y$. Indeed, assume for a contradiction that this is not the
  case. Then, we would have two distinct points $x \neq y$ of
  $\Conv\alpha$ that share the same projection onto $\M$, showing that
  $\angle(\Aff\alpha,\Tangent {\pi_\M(x)} \M) = \frac{\pi}{2}$ for
  some $x \in \Conv\alpha$ and contradicting our assumption that
  $\Theta < \frac{\pi}{6}$. Hence, $x = y \in \Conv{(\alpha \cap
    \beta)}$. We claim that furthermore $\Power \alpha x = \Power
  \beta y$. Indeed,
  Lemma~\ref{lemma:power-expression-with-additional-point} implies
  that when $x$ is an affine combination of points in $\alpha$, that
  is, when $x = \sum_{a\in\alpha} \lambda_a a$ with $\sum_a \lambda_a
  = 1$, then $\Power{\alpha}{x} = - \sum_{a\in\alpha} \lambda_{a}
  \|x-a\|^2$. In particular, if $x$ belongs to the convex hull of a
  face of $\alpha$, the expression of the power distance depends only
  upon the vertices of that face. It follows that
  \[
  \Power \alpha x = \Power {\alpha \cap \beta} x = \Power {\alpha\cap\beta} y = \Power \beta y. 
  \]
  Since $y \in \Conv{(\alpha \cap \beta)}$, we have $\sum_{b \in \beta\setminus\alpha} \BaryCoord y \beta b = 0$ and combining this with the above equality, we get the desired inequality.

  \begin{figure}[htb]
  \centering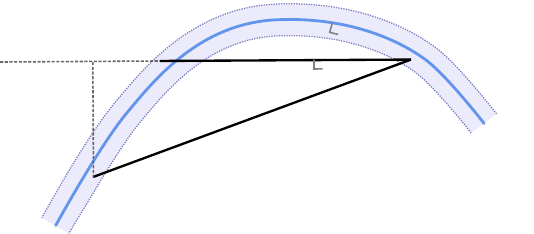
  \caption{Right: Notation for the proof of Lemma
  \ref{lemma:advanced-bound-on-power}. We know by Lemma \ref{lemma:small-tubular-neigborhood} that if $\alpha,\beta \subseteq \Offset \M \delta$ with $\delta \leq \frac{\rho}{16}$, then $\Conv \alpha, \Conv \beta \subseteq \Offset \M {\frac{\rho}{4}}$ as assumed in Lemma \ref{lemma:advanced-bound-on-power}. \label{figure:advanced-bound-on-power}}
  \end{figure}

  \medskip \noindent Second, assume that $y \in \Conv\beta \setminus
  \Conv{(\alpha\cap\beta)}$; see Figure
  \ref{figure:advanced-bound-on-power}. Write $\mu_b =
  \BaryCoord y \beta b$ and note that $\mu_b \geq 0$. Letting $y' =
  \pi_{\Aff\alpha}(y)$, we know by Lemma~\ref{lemma:basic-bound-on-power} that:
  \begin{equation}
    \label{eq:compare-powers-1}
    \Power{\beta}{y} \leq \Power{\alpha}{y'} -  (\pro^2 + 2 \pro R(\alpha))\sum_{b \in \beta \setminus \alpha} \mu_b.
  \end{equation}
    Because $y \not \in \Conv{(\alpha \cap \beta)}$, we have $\sum_{b
      \in \alpha \cap \beta} \mu_b \neq 1$ and therefore $\sum_{b \in
      \beta \setminus \alpha} \mu_b \neq 0$. First, suppose that
    $y=x$. In that case, $y'=x$ and the result follows
    immediately. Second, suppose that $y \neq x$. We claim that in
    that case we also have $y \neq y'$. Indeed, if we were to have
    that $y = y'$, then both $x$ and $y$ would belong to $\Aff\alpha$
    and since $\pi_\M(x) = \pi_\M(y)$, this would mean that
    $\angle(\Aff\alpha,\Tangent {\pi_\M(x)} \M) = \frac{\pi}{2}$ for
    $x \in \Conv\alpha$, contradicting our assumption that
    $\Theta < \frac{\pi}{6}$. Thus, $x \neq y$ and $y \neq
    y'$, and we can define the angle $\theta = \angle
    xyy'$. Noting that $\theta
    \leq \angle(\Aff\alpha,\Tangent{\pi_\M(x)}{\M}) \leq \Theta$
    and $\|x-y'\| = \|x-y\| \sin\theta$, and recalling that $Z(\alpha)$ is
    the radius of the $(d-1)$-dimensional circumsphere of $\alpha$, we
    have:
  \begin{align}
    \Power{\alpha}{y'} - \Power{\alpha}{x} &= \| y' - Z(\alpha) \|^2 - \| x - Z(\alpha) \|^2 \nonumber
    = \DotProd{(y' - x)}{(x+y'-2Z(\alpha))} \nonumber \\
    &\leq \|x - y'\| \cdot \left( \| x - Z(\alpha) \| + \| y' - Z(\alpha) \| \right) \nonumber \\
    &\leq \|x - y'\| \cdot \left( 2 \| x - Z(\alpha) \| + \| x - y' \| \right) \nonumber \\
    &\leq \|x - y\| \sin(\theta) \left( 2 R(\alpha) + \|x - y\| \sin(\theta) \right) \label{eq:compare-powers-2}.
  \end{align}
  Writing $m = \pi_\M(x) = \pi_\M(y)$, we have $\|x-y\| \leq \|x-m\| +
  \|m-y\| \leq \frac{\rho}{2}$. Summing up Inequalities (\ref{eq:compare-powers-1}) and (\ref{eq:compare-powers-2}), we get
  \begin{equation*}
    \Power{\beta}{y} - \Power{\alpha}{x} \leq
    - \underbrace{(\pro^2 + 2 \pro R(\alpha)) \sum_{b \in \beta \setminus \alpha} \mu_b}_{A}
    + \underbrace{\|x - y\| \sin(\Theta) \left( 2 R(\alpha) + \frac{\rho}{2} \sin(\Theta) \right)}_{B}.
  \end{equation*}
  To establish the lemma in the second case, it suffices to show that $2B < A$, that is, 
  \begin{equation}
    \label{eq:raw-condition-for-bounding-power}
  2\|x - y\| \sin \Theta \cdot \left( 2 R(\alpha) + \frac{\rho}{2} \sin \Theta \right)
  < (\pro^2 + 2 \pro R(\alpha)) \sum_{b \in \beta \setminus \alpha} \mu_b.    
  \end{equation}
  We consider two subcases:

  \smallskip \noindent \underline{Subcase 1: $\alpha \cap \beta =
    \emptyset$.}  In that case, $\sum_{b \in \beta \setminus \alpha}
  \mu_b = 1$, and because $2\|x-y\|  \leq 4\rho \leq 5 \rho$, one can see that
  (\ref{eq:raw-condition-for-bounding-power}) follows from our assumptions.
  
    \smallskip \noindent \underline{Subcase 2: $\alpha \cap \beta \neq
      \emptyset$.} In that case, we know that there exists a point $u
    \in \Conv{(\beta \cap \alpha)}$ and a point $v \in
    \Conv{(\beta\setminus\alpha)}$ such that $y = \sum_{b \in \beta \cap
      \alpha} \mu_b u + \sum_{b \in \beta \setminus \alpha} \mu_b
    v$; see Figure \ref{figure:advanced-bound-on-power}. Furthermore, letting $v' = \pi_{\Aff\alpha}(v)$ we have
    \[
    \sum_{b \in \beta \setminus \alpha} \mu_b = \frac{\|y-u\|}{\|v-u\|} \geq \frac{\|y-y'\|}{\|v-u\|} \geq \frac{\|x - y\|\cos\theta}{\Diam\beta}
    \geq \frac{\|x-y\| \cos \theta}{2\rho} \geq \frac{\sqrt{3}}{4\rho} \cdot \|x-y\|.
    \]
    Again,  (\ref{eq:raw-condition-for-bounding-power}) follows from our assumptions.
\end{proof}

The next lemma says that if a subset $\sigma \subseteq \Rspace^\Dim$
is sufficiently small and sufficiently close to a subset $A
\subseteq \Rspace^\Dim$ compare to the reach of $A$, then the convex
hull of $\sigma$ is not too far away from $A$.

\begin{lemma}
  \label{lemma:small-tubular-neigborhood}
  Let $16\delta \leq \rho \leq \frac{\Reach A}{3}$. If the subset $\sigma
  \subseteq \Offset A \delta$ is $\rho$-small, then $\Conv \sigma
  \subseteq \Offset A {\frac{\rho}{4}}$.
\end{lemma}

\begin{proof}
  Let $\reach = \Reach A$. Applying Lemma 14 in \cite{socg10-convex},
  we get that $\Conv \sigma \subseteq \Offset A r$ for $r = \reach -
  \sqrt{(\reach - \delta)^2 - \rho^2}$. Since $\delta \leq
  \frac{\rho}{16}$, we deduce that $\frac{r}{\reach} \leq 1 -
  \sqrt{\left(1 - \frac{\rho}{16\reach}\right)^2 -
    \frac{\rho}{\reach}^2}$ and since for all $0 \leq t \leq
  \frac{1}{3}$ we have $1 - \sqrt{\left(1-\frac{t}{16}\right)^2-t^2}
  \leq \frac{t}{4}$, we obtain the result.
\end{proof}

We are now ready to prove Lemma \ref{lemma:final-bound-on-power}.

\begin{proof}[Proof of Lemma \ref{lemma:final-bound-on-power}]
  Let $\alpha$ be a non-degenerate $\varepsilon$-small $d$-simplex of
  $\Delloc P \rho d$. Because $\alpha \in \Delloc P \rho d$, we have that
  $\alpha \in \Del{\pi_{\Aff \sigma}(P \cap B(c_\sigma,\rho))}$, and
  because $\alpha$ is $\varepsilon$-small, we have that
  $B(Z(\alpha),R(\alpha)) \subseteq B(c_\sigma,2\varepsilon) \subseteq
  B(c_\sigma,\rho)$ and consequently $\alpha \in \Del{\pi_{\Aff
      \sigma}(P \cap B(c_\sigma,3\rho))}$.

  Let $\beta$ be a non-degenerate $\rho$-small $d$-simplex of $P$. Assume that
  $\pi_\M(\alpha) \cap \pi_\M(\beta) \neq \emptyset$ and let us show
  that $\beta \subseteq P \cap B(c_\sigma,3\rho)$. Suppose that $x \in
  \Conv\alpha$ and $y \in \Conv\beta$ share the same projection $m$
  onto \M, that is, $m = \pi_\M(x) = \pi_\M(y)$.  Since both $\alpha$
  and $\beta$ are $\rho$-small,
  Lemma~\ref{lemma:small-tubular-neigborhood} implies that both $\Conv
  \alpha$ and $\Conv \beta$ are contained in the
  $\left(\frac{\rho}{4}\right)$-tubular neighborhood of $\M$ and in
  particular $\|x - y\| \leq \|x - m\| + \|m - y\| \leq \frac{\rho}{4}
  + \frac{\rho}{4} \leq \frac{\rho}{2}$. For all vertices $b \in
  \beta$, we thus have
  \[
  \|c_\alpha - b \| \leq \|c_\alpha - x\| + \| x - y \| +\|y - b\| \leq \varepsilon + \frac{\rho}{2} + 2\rho \leq 3\rho,
  \]
  showing that $\beta \subseteq P \cap B(c_\sigma,3\rho)$. Hence, we
  get that $\alpha \in \Del{\pi_{\Aff\sigma}(\alpha \cup \beta)}$ and
  can easily see that $\pro = \ProtectGlob P {3\rho} \leq \tilde\pro =
  \protection \alpha \beta$. Let $\tilde\Theta =
  \AngularDeviation{\alpha} \leq \Theta$. To apply Lemma
  \ref{lemma:advanced-bound-on-power}, we need to verify that
  \[
  5\rho \sin(\tilde\Theta) \left( 2 R(\alpha) + \frac{\rho}{2} \sin(\tilde\Theta) \right)
  < \tilde\pro^2 + 2 \tilde\pro R(\alpha).
  \]
  Since $\frac{\sepvar}{2} \leq R(\alpha) \leq \varepsilon$ and $\sin
  t \leq t$ for all $t \geq 0$, this follows from:
  \[
  10 \rho  \Theta \left( \varepsilon + \rho
  \Theta \right) < \pro^2 + \pro \sepvar,
  \]
  which is a consequence of our hypotheses.
\end{proof}

\subsection{Final}
\label{section:final}

Suppose that $K$ is a simplicial complex with vertex set $P$. Write $D
= \Delloc P \rho d$, $\D = \Shadow D$ and $\K = \Shadow{K}$ for short.
In this section, we prove our main theorem by applying
Lemma~\ref{lemma:technical}. This requires us to define two maps
$\varphi : \K \to \D$ and $f : \D \to \Rspace^+$, two weights
$W(\alpha)$ and $W_{\min}(\alpha)$ for each $d$-simplex $\alpha \in
K$, and to check that these maps and weights satisfy the requirements
of Lemma~\ref{lemma:technical}. For each $\alpha \in K$, let
$W(\alpha)=\omega(\alpha)$ be the Delaunay weight of $\alpha$. To be
able to define $\varphi$, $f$, and $W_{\min}$, we assume that the
following conditions are met:
\begin{itemize}
\item $D$ is a faithful reconstruction of \M;
\item For every $d$-simplex $\sigma \subseteq K$, the map
  $\restr{\pi_\M}{\Conv \sigma}$ is well-defined and injective.
\end{itemize}
These conditions are easily derived from the assumptions of the main
theorem.  We are now ready to introduce additional notation. Consider
a subset $X \subseteq \Rspace^\Dim$ and suppose that the map
$\restr{\pi_\M}{X}$ is well-defined and injective. Then it is possible
to define a bijective map $\pi_{X \to \M}: X \to \pi_\M(X)$. Because
$D$ is a faithful reconstruction of \M, the map $\pi_{\D\to\M}$ is
well-defined and bijective. Similarly, for every $d$-simplex $\sigma
\in K$, the map $\pi_{\Conv\sigma\to\M}$ is well-defined and
bijective. We now introduce the map $\varphi: \K \to \D$ defined by
$\varphi = [\pi_{\D \to \M}]^{-1} \circ \pi_\M$ and let $f: \D
\to \Rspace^+$ be the map defined by:
\begin{equation}
  \label{eq:g}
  f(x) = \min_{\sigma}  \left(- \Power{\sigma}{[\pi_{\Conv\sigma\to\M}]^{-1}\circ \pi_\M(x)}\right),
\end{equation}
where the minimum is taken over all $d$-simplices $\sigma \in K$ such
that $x \in \varphi(\Conv\sigma)$. Note that $f(x)$ can be defined
equivalently as the minimum of $- \Power{\beta}{y}$ over all
$d$-simplices $\beta \in K$ and all points $y \in \Conv \beta$ such
that $\pi_\M(x) = \pi_\M(y)$. Given a $d$-simplex $\sigma \in K$, we
associate to $\sigma$ the weight:
\begin{align}
W_{\min}(\sigma) &= \int_{x \in \varphi(\Conv\sigma)} f(x) \, dx.
\end{align}

\begin{lemma}
  \label{lemma:Wmin-alpha}
  Under the assumptions of Theorem \ref{theorem:reconstruction}:
  \begin{itemize}
  \item For every $d$-simplex $\alpha \in D$ and every point $x
    \in \Conv \alpha$, we have $f(x) = - \Power \alpha x$.
  \item For every $d$-simplex $\alpha \in D$, we have
    $W_{\min}(\alpha) = W(\alpha)$.
  \end{itemize}
\end{lemma}

\begin{proof}
  Consider a $d$-simplex $\alpha \in D$, a $d$-simplex $\beta \in K$,
  $x \in \Conv\alpha$ and $y \in \Conv\beta$ such that $\pi_\M(x) =
  \pi_\M(y)$. Applying Lemma \ref{lemma:final-bound-on-power}, we obtain that $\Power
  \beta y \leq \Power \alpha x$ or equivalently
  $
  \Power \beta {[\pi_{\Conv\beta\to\M}]^{-1}  \circ \pi_\M(x)} \leq \Power \alpha x
  $
  and therefore $f(x) = - \Power \alpha x$.  To establish the second
  item of the lemma, notice that for all $\alpha \in D$, the
  restriction of $\varphi$ to $\Conv\alpha$ is the identity function,
  $\varphi_{| \Conv\alpha} = \Id$ and therefore $\varphi(\Conv\alpha)
  = \Conv{\alpha}$. Since we have just established that $f(x) = -
  \Power \alpha x$, we get that
\[
W_{\min}(\alpha) = \int_{x \in \varphi(\Conv\alpha)} f(x) \, dx = \int_{x \in \Conv \alpha} - \Power \alpha x \, dx =  \omega(\alpha) = W(\alpha),
\]
which concludes the proof.
\end{proof}

\begin{lemma}
  \label{lemma:Wmin-beta}
  Under the assumptions of Theorem \ref{theorem:reconstruction}, for
  every $d$-simplex $\beta \in K \setminus D$, we have
  $W_{\min}(\beta) < W(\beta)$.
\end{lemma}

\begin{proof}
  We need some notation. Given $\alpha$ and $\beta$ in $K$, we write
  $\RestConv\beta\alpha$ for the set of points $y \in \Conv{\beta}$
  for which there exists a point $x \in \Conv\alpha$ such that
  $\pi_\M(x) = \pi_\M(y)$.  We define the map $\varphi_{\beta \to
    \alpha} : \RestConv\beta\alpha \to \RestConv\alpha\beta$ as
  $\varphi_{\beta \to \alpha}(y) = [\pi_{\Conv\alpha\to\M}]^{-1} \circ
  \pi_{\Conv\beta\to\M}(y)$. Note that $\varphi_{\beta\to\alpha}$ is
  invertible and its inverse is $\varphi_{\alpha\to\beta}$. Also, note
  that $J$ in Theorem \ref{theorem:reconstruction} has been chosen
  precisely so that one can apply Lemma
  \ref{lemma:BoundsOnJacobianPhiPi1Pi2} in Appendix~\ref{appendix:J} and guarantee that
  $|\det(\Differential\varphi_{\beta \to \alpha})(y)| \in [\frac{1}{1 + J}, 1 + J]$ for all
  $\alpha,\beta \in K$ and all $y \in \RestConv\beta\alpha$.  Consider
  a $d$-simplex $\beta \in K \setminus D$. By Lemma
  \ref{lemma:Wmin-alpha}, $f(x) = - \Power\alpha x$ and therefore:
  \begin{eqnarray*}
    W_{\min}(\beta) &=& \sum_{\alpha \in D^{[d]}} \int_{x \in \RestConv\alpha\beta} - \Power{\alpha}{x} \, dx.    
  \end{eqnarray*}
  For any convex combination $y$ of points in $\beta$, let $\{ \Comb
  \beta b y\}_{b \in \beta}$ designate the family of non-negative
  real numbers summing up to 1 such that $y = \sum_{b \in \beta}
  \Comb \beta b y b$. 
  Plugging in the upper bound on $- \Power{\alpha}{x}$ provided by Lemma
  \ref{lemma:final-bound-on-power}, letting
  \[
  c = \frac{1}{2} \left(\pro^2 + \pro \sepvar \right),
  \]
  and making the change of variable $x = \varphi_{\beta \to
    \alpha}(y)$, we upper bound $W_{\min}(\beta)$ as follows:
  \begin{eqnarray*}
     W_{\min}(\beta) &\leq& \sum_{\alpha \in D^{[d]}} \int_{x \in \RestConv\alpha\beta}
     \left[-\Power{\beta}{\varphi_{\alpha \to \beta}(x)} - c \sum_{b \in \beta \setminus \alpha} \Comb{b}{\beta}{\varphi_{\alpha \to \beta}(x)} \right] \, dx\\
     &=& \sum_{\alpha \in D^{[d]}} \int_{y \in \RestConv\beta\alpha}
     \left[-\Power{\beta}{y} - c \sum_{b \in \beta \setminus \alpha} \Comb{b}{\beta}{y} \right]
     |\det(\Differential\varphi_{\beta \to \alpha})(y)| \, dy \\
     &\leq& (1+J) W(\beta)
       - (1 + J)^{-1} c \sum_{\alpha \in D^{[d]}}
       \int_{y \in \RestConv\beta\alpha}  \sum_{b \in \beta \setminus \alpha} \Comb{b}{\beta}{y} \, dy.
  \end{eqnarray*}
  A key observation is that, because $\beta \neq \alpha$, then $\beta
  \setminus \alpha \neq \emptyset$. Therefore the sum $\sum_{b
    \in \beta \setminus \alpha} \Comb{b}{\beta}{y}$ does not vanish
  and is always lower bounded by $\inf_{b \in \beta}
  \Comb{b}{\beta}{y}$. Associating the quantity
  \[
  \Omega(\beta) = \int_{y  \in \Conv\beta} \inf_{b \in \beta} \Comb{b}{\beta}{y} \, dy,
  \]
  to $\beta$ we thus obtain that $W_{\min}(\beta) \leq (1 + J) W(\beta) - \left(1 +
  J\right)^{-1} c\, \Omega(\beta)$. Hence, $W_{\min}(\beta) <
  W(\beta)$ as long as
  \begin{equation}
    \label{eq:key-condition}
  J W(\beta) < (1 + J)^{-1} c\,  \Omega(\beta).  
  \end{equation}
  Using a change of variable, it is not too difficult to show that
  $\Omega(\beta) = d!  \Volume\beta \Omega(\Delta_d)$, where $\Delta_d
  = \{ \lambda \in \Rspace^d \mid \sum_{i=1}^d \lambda_i \leq 1;
  \lambda_i \geq 0, i=1,2,\ldots, d \}$ represents the standard
  $d$-simplex. Remark that $\Omega(\Delta_d)$ is a constant that
  depends only upon the dimension $d$ and is thus universal. Plugging in
  $\Omega(\beta) = d!  \Volume\beta \Omega(\Delta_d)$ on the right
  side of (\ref{eq:key-condition}), and the expression of $W(\beta) =
  \omega(\beta)$ given by Lemma \ref{lemma:weight} on the left side of
  (\ref{eq:key-condition}), and recalling that $\beta$ is
  $\rho$-small, we find that condition
  (\ref{eq:key-condition}) is implied by the following condition:
  \begin{equation*}
    \label{eq:J-condition}
    J  \rho^2
    <  (1 + J)^{-1} \frac{ (d+2) (d-1)!}{4}
     \left(\pro^2 + \pro \sepvar \right) \,  \Omega(\Delta_d),
  \end{equation*}
  which we have assumed to hold.
\end{proof}

\begin{proof}[Proof of Theorem \ref{theorem:reconstruction}]
    Let $D = \Delloc P \rho d$, $\D = \Shadow{D}$ and $\K =
    \Shadow{K}$. Theorem~\ref{theorem:homeomorphism-from-sampling-conditions}
    ensures that $\D$ is a $d$-manifold and $\pi_\M : \D \to \M$ is a
    homeomorphism. Give to $\D$ an orientation that is consistent
    with that of \M, that is, $\sign{\sigma}{\D}=\sign{\sigma}{\M}$ for all $\sigma \in K^{[d]}$.
    Define $\varphi : \K \to \D$, $f : \D \to \Rspace^+$,
    $W$, and $W_{\min}$ as explained at the beginning of the
    section. Consider the $d$-chain $\gamma_{\min}$ on $K$:
  \[
  \gamma_{\min}(\sigma) =
  \begin{cases}
    \sign{\sigma}{\M} \quad \quad &\text{if $W_{\min}(\sigma) = W(\sigma)$,}\\
    0 \quad \quad &\text{otherwise.}
  \end{cases}
  \] By Lemma \ref{lemma:Wmin-alpha} and Lemma \ref{lemma:Wmin-beta}, the following
  property holds: for all $\sigma \in K$, 
  $W_{\min}(\sigma) = W(\sigma)$ if and only if $\sigma$ is a
  $d$-simplex of $D$.  It follows that $\gamma_{\min} =
  \code{D}$. Furthermore, we have $\sum_{\sigma\in K^{[d]}}
  \gamma_{\min}(\sigma) \sign{\sigma}{\M}\Indicator{\varphi(\Conv{\sigma})}(x) =
  \sum_{\sigma\in D^{[d]}} \Indicator{\Conv{\sigma}}(x) = 1$ for
  almost all $x \in \D$. Recalling that $W = \omega$ and therefore
  $\|\gamma\|_{1,W}=\Edel(\gamma)$, and applying Lemma
  \ref{lemma:technical}, we deduce that $\gamma_{\min}=\code{D}$ is the unique
  solution to the following optimization problem over the set of
  chains in $C_d(K,\Rspace)$:

  \medskip
  \begin{tcolorbox}
    \vspace{-4mm}
    \begin{mini*}
      {\gamma}{\Edel(\gamma)}{}{}
      \label{problem:technical}
      \addConstraint{\sum_{\sigma\in K^{[d]}} \gamma(\sigma) \sign{\sigma}{\M}  \Indicator{\varphi(\Conv{\sigma})}(x)}{=1, \hspace{2mm}}{\text{for almost all } x \in \D}
      \tag{$\circledast$}
    \end{mini*}
  \end{tcolorbox}

  \noindent We now claim that the feasible set of Problem~(\ref{problem:technical})
  contains the feasible set of
  Problem~(\ref{problem:reconstruction}). Indeed, consider a $d$-chain
  $\gamma$ that satisfies the constaints of Problem~(\ref{problem:reconstruction}), that is, such that
  \begin{equation*}
    \begin{cases}
      \partial \gamma = 0,\\
      \sum_{\sigma\in K^{[d]}} \gamma(\sigma)\sign{\sigma}{\M}\Indicator{\pi_\M(\Conv{\sigma})}(m_0) = 1.
    \end{cases}
  \end{equation*}
  Then, by Lemma~\ref{lemma:PushForwardOfCycleIsConstant} in Appendix~\ref{appendix:practical-conditions}, we obtain
  that $\gamma$ also satisfies the following constraint:
  \begin{equation*}
    \label{eq:one-constraint}
    \sum_{\sigma\in K^{[d]}} \gamma(\sigma)\sign{\sigma}{\M}\Indicator{\pi_\M(\Conv{\sigma})}(m)=1, \quad \text{for almost all } m \in \M,
  \end{equation*}
  which is equivalent to the constaint of
  Problem~(\ref{problem:technical}). Since the unique solution to
  Problem~(\ref{problem:technical}) is $\code D$,
  Theorem~\ref{theorem:homeomorphism-from-sampling-conditions}
  guarantees that $D = \Delloc P \rho d$ is a faithful reconstruction
  of \M. By Lemma \ref{lemma:manifold-as-cycle}, $\code{D}$ is thus a
  cycle. Hence, the unique solution $\code{D}$ to
  Problem~(\ref{problem:technical}) also satisfies the containts of
  Problem~(\ref{problem:reconstruction}) and, because the feasible set
  of Problem~(\ref{problem:technical}) contains the feasible set of
  Problem~(\ref{problem:reconstruction}), $\code{D}$ is also the
  unique solution to Problem~(\ref{problem:reconstruction}).
\end{proof}

\section{Practical aspects}
\label{section:PracticalAlgorithm}

In this section, we discuss practical aspects.

\subsection{Transforming the problem into a realistic algorithm}
\label{section:RealisticAlgorithm}

Besides the complex $K$ that one can build from $P$, Problem
(\ref{problem:reconstruction}) seems to require the knowledge of \M
for expressing the normalization constraint $\Load {m_0} \M {K} \gamma
= 1$.  What we call a {\em realistic} algorithm is an algorithm that
takes only the point set $P$ as input.  In this section, we explain
how to transform Problem (\ref{problem:reconstruction}) into an
equivalent problem that does not refer to \M anymore, thus providing a
realistic algorithm. Roughly, we simply replace the constraint $\Load
{m_0} \M {K} \gamma = 1$ by a constraint of the form $\LoadLoc {x_0} \Pi
{\Sigma} \gamma = 1$, where $x_0$ is a point ``close'' to \M, $\Pi$ is
a $d$-flat that ``roughly approximates'' \M near $x_0$ and $\Sigma$
are simplices of $K$ ``close'' to $x_0$. To make this idea precise, we
use the following localized version of the load:
\[
\LoadLoc {x_0} \Pi \Sigma \gamma = \sum_{\sigma \in \Sigma^{[d]}} \gamma(\sigma) \sign{\sigma}{\Pi} \Indicator{\pi_\Pi(\Conv{\sigma})}(x_0)
\]
and state conditions in Lemma \ref{lemma:RealisticProblemEquivalence}
(see below) under which Problem~(\ref{problem:reconstruction}) is
equivalent to the problem obtained by replacing the constraint $\Load
{m_0} \M {K} \gamma = 1$ with the constraint $\LoadLoc {x_0} \Pi {\Sigma}
\gamma = 1$. Given a point $x \in \Rspace^\Dim$ and $r \geq 0$, let us
introduce the subset of $K$:
\[
K[x,r] = \{
\sigma \in K \mid \Conv \sigma \cap B(x,r) \neq \emptyset \}.
\]
Note that $K[x,r]$ is not necessarily a simplicial
complex.

\begin{lemma}\label{lemma:RealisticProblemEquivalence}
  Suppose $0 \leq \rho \leq \frac{\reach}{25}$. Consider a point $x_0
  \in \Offset \M \rho$ and a $d$-dimensional affine space $\Pi$
  passing through $x_0$. Suppose that $\angle(\Pi,\Tangent {\pi_\M(x_0)}
  \M) \leq \frac{\pi}{8}$ and that the orientation of $\Pi$ is
  consistent with that of $\Tangent {\pi_\M(x_0)} \M$. Then, Problem
  (\ref{problem:reconstruction}) is equivalent to the following problem
  \begin{tcolorbox}
    \vspace{-4mm}
    \begin{mini*}
      {\gamma}{\Edel(\gamma)}{}{}  \label{problem:practical}
      \addConstraint{\partial \gamma}{=0}  \tag{$\star\star$}
      \addConstraint{\LoadLoc {x_0} \Pi {K[x_0,4\rho]} \gamma}{=1}
    \end{mini*}
  \end{tcolorbox}
\end{lemma}

\begin{proof}
  This is a direct consequence of Lemma~\ref{lemma:PracticalNormalizationByApproximateTangentSpace} in Appendix~\ref{appendix:practical-conditions}.
\end{proof}

Observe that the conditions on the $d$-flat $\Pi$ in the above lemma
are rather mild. Indeed, we only require $\Pi$ to pass through a point
$x_0$ such that $d(x_0,\M) \leq \frac{\reach}{25}$ and
$\angle(\Pi,\Tangent {\pi_\M(x_0)} \M) \leq \frac{\pi}{8}$. Hence, $\Pi$
only needs to be what we could call a {\em rough approximation} of \M
near $x_0$.  In practice, we may take for $x_0$ any point $p_0 \in P$ and
for $\Pi$ the $d$-dimensional affine space $T_{p_0}(P, \rho)$ passing
through $p_0$ and parallel to the $d$-dimensional vector space
$V_{p_0}(P, \rho)$ defined as follows: $V_{p_0}(P, \rho)$
is spanned by the eigenvectors associated to
the $d$ largest eigenvalues of the inertia tensor of $(P \cap B(p_0,
\rho)) - c$, where $c$ is the center of mass of $P \cap B(p_0,
\rho)$. By Lemma~\ref{lemma:PerturbationAndLLL_pseudoTangent} in Appendix~\ref{appendix:subsection:PCA}, for
$\frac{\rho}{\reach}$ small enough and $\varepsilon < \frac{\rho}{16}$,
we have
\[ 
\angle( T_{p_0}(P, \rho), T_{\pi_\M(p_0)}\M) \leq \frac{\pi}{8}.
\]
See Appendix~\ref{appendix:subsection:PCA} for more
details. Hence, the assumptions of the above lemma hold for $x_0=p_0$
and $\Pi = T_{p_0}(P, \rho)$. This shows that the normalization
constraint in Problem~(\ref{problem:reconstruction}) can be replaced
by a constraint whose definition depends only upon the point set $P$,
thus providing a realistic algorithm.

\subsection{Perturbing the data set for ensuring the safety conditions}\label{section:PerturbationForSafetyAlgorithm}

In this section, we assume that $P_0$ is a $\delta_0$-accurate
$\varepsilon_0$-dense sample of \M and perturbe it to obtain a point
set $P$ that satisfies the assumptions of our main theorem. For this,
we use the Moser Tardos Algorithm \cite{moser2010constructive} as a
perturbation scheme in the spirit of what is done in \cite[Section
  5.3.4]{boissonnat2018geometric}.

The perturbation scheme is parametrized with real numbers $\rho \geq
0$, $r_{\operatorname{pert.}} \geq 0$, $\mathrm{Height}_{\min}> 0$, and
$\mathrm{Prot}_{\min}> 0$. To describe it, we need some notations and
terminology. Let $T_{p_0}^*= T_{p_0}(P_0,3\rho)$ be the $d$-dimensional
affine space passing through $p_0$ and parallel to the $d$-dimensional
vector space $V_{p_0}(P_0,3\rho)$.
To each point $p_0 \in P_0$, we associate a
perturbed point $p \in P$, computed by applying a sequence of
elementary operations called reset. Precisely, given a point $p \in P$
associated to the point $p_0 \in P_0$, the {\em reset} of $p$ is the
operation that consists in drawing a point $q$ uniformely at random in
$T_{p_0}^* \cap B(p_0,r_{\operatorname{pert.}})$ and assigning $q$ to
$p$. Finally, we call any of the two situations below a {\em bad
  event}:
\begin{description}
\item[\styleitem{Violation of the height condition:}] There exists a $\rho$-small $d$-simplex
$\sigma \subseteq P$ such that $\height{\sigma}<\mathrm{Height}_{\min}$;
\item[\styleitem{Violation of the protection condition:}] There exists a pair $(p,\sigma)$ made of a point $p \in P$
  and a $d$-simplex $\sigma \subseteq P \setminus \{p\}$ such that $p \in
  B(c_\sigma,3\rho)$ and $\protection{\sigma}{\{p\}} \leq \mathrm{Prot}_{\min}$.
\end{description}
In both situations, we associate to the bad event $E$ a set of points
called the points {\em correlated} to $E$. In the first
situation, the points correlated to $E$ are the $d+1$ vertices of
$\sigma$ and in the second situation, they are the $d+2$ points of
$\{p\} \cup \sigma$.

\bigskip

\noindent \fbox{
    \begin{minipage}{0.97\textwidth}
      \noindent {\bf Moser-Tardos Algorithm: }\\
      \texttt{
      1. For each $p_0 \in P_0$, compute the $d$-dimensional affine space $T_{p_0}^*$ \\
      2. For each point $p \in P$, reset $p$ \\
      3. WHILE (some bad event $E$ occurs):\\
      --------------- For each point $p$ correlated to $E$, reset $p$ \\
      ----- END WHILE\\
      4. Return $P$
      }
    \end{minipage}
}

\bigskip

Roughly speaking, in our context, the Moser Tardos Algorithm reassigns
new coordinates to any point $p \in P$ that is correlated to a bad
event as long as a bad event occurs. A beautiful result from
\cite{moser2010constructive} tells us that if bad events are
  mostly independent from one another and have each a sufficiently
  small probability to occur, then the Moser-Tardos Algorithm
terminates and does so in a number of steps that is expected
to be linear in the size of $P_0$. Precisely, suppose that each
  bad event is independent of all but at most $\Gamma$ of the other
  bad events and the probability of a bad event is at most
  $\varpi$. Then, the result in \cite{moser2010constructive} tells us
  that the Moser-Tardos algorithm terminates with expected time
  $O(\sharp P_0)$ whenever
\begin{equation}
  \label{eq:moser-tardos}
  \varpi \leq \frac{1}{e(\Gamma+1)},
\end{equation}
where $e$ is the base of natural logarithms.
Using this result, one can establish the following lemma, the proof of which is beyond the scope of this paper:

\begin{lemma}\label{lemma:PerturbationTunedForTheorem}
Let $\varepsilon_0 \geq 0$, $\eta_0 >0$, and $\rho =
C_{\operatorname{ste}} \varepsilon_0$, where $C_{\operatorname{ste}}
  \geq 32$.  Let $\delta_0 = \frac{\rho^2}{\reach}$,
$r_{\operatorname{pert.}} = \frac{\eta_0 \varepsilon_0}{20}$,
$\varepsilon = \frac{21}{20}\varepsilon_0$, and $\delta = 2\delta_0$.
 There are
positive constants $c_0$, $c_1$, and $c_2$ that depend only upon
$\eta_0$, $C_{\operatorname{ste}}$, and $d$ such that if
$\frac{\rho}{\reach} < c_0$ then, given a point set $P_0$
such that $\M \subseteq \Offset {(P_0)} {\varepsilon_0}$, $P_0
\subseteq \Offset \M {\delta_0}$, and $\Sep{P_0} > \eta_0
\varepsilon_0$, the point set $P$ obtained after resetting each of its points
satisfies $\M \subseteq \Offset {P} {\varepsilon}$, $P
\subseteq \Offset \M {\delta}$, and $\Sep{P} > \frac{9}{10} \eta_0
\varepsilon_0$.  Moreover, whenever we apply the Moser-Tardos Algorithm with
$\mathrm{Height}_{\min} = c_1 \left( \frac{\rho}{\reach}
\right)^{\frac{1}{3}} \rho $ and $\mathrm{Prot}_{\min} =c_2 \left(
\frac{\rho}{\reach} \right)^{\frac{1}{3}} \rho$, the algorithm terminates with
expected time $O( \sharp P_0)$ and returns a point set $P$ that is a 
$\delta$-accurate $\varepsilon$-dense sample of \M and that satisfies the assumptions of 
Theorem \ref{theorem:reconstruction}.
\end{lemma}

We only sketch the proof of Lemma \ref{lemma:PerturbationTunedForTheorem}
below. 

\begin{proof}[Sketch of proof]
  The proof consists in applying the Moser Tardos theorem
  \cite{moser2010constructive}. In other words, we show that
  Condition~\eqref{eq:moser-tardos} holds, for a well-chosen upper
  bound $\varpi$ on the probability of each bad event and a
  well-chosen upper bound $\Gamma$ on the number of bad events to
  which each bad event is dependent upon. Upper bounds $\varpi$ and
  $\Gamma$ are obtained by adapting the proof of a similar simpler result
  presented in the appendix of
  \cite{AttaliLieutierFlatDelaunay2022}. The intuition is that thanks
  to Lemma~\ref{lemma:PerturbationAndLLL_pseudoTangent} in Appendix~\ref{appendix:subsection:PCA}, one can
  compute from the sample $P_0$ a local approximation
  $T_{p_0}(P_0,3\rho)$ of a local tangent space with accuracy
  $\mathcal{O} \left( \frac{\rho}{\reach} \right)$.  It follows that,
  if $\frac{\rho}{\reach} $ is small enough, the volume, in
  $\Pi_{p_0\in P_0} T_{p_0}(P_0,3\rho)$, for which a height or
  protection condition is violated, can be made arbitrary small, and
  Condition~\eqref{eq:moser-tardos} required for Moser-Tardos
  algorithm to terminate will be met.

  When the Moser-Tardos algorithm terminates, we thus have two positive constants
  $c_1$ and $c_2$ such that
 \begin{align}
   \Height{P}{\rho} &> c_1 \left( \frac{\rho}{\reach}
   \right)^{\frac{1}{3}} \rho, \label{eq:control-height}\\
   \ProtectGlob P {3\rho} &> c_2 \left( \frac{\rho}{\reach} \right)^{\frac{1}{3}}  \rho. \label{eq:control-protection} 
 \end{align}
 For short, write
 \begin{eqnarray*}
   \pro &=& \ProtectGlob P {3\rho},\\
   \sepvar &=& \Sep{P},\\
   \Theta &=& \AngularDeviation{P,\rho}.
 \end{eqnarray*}
 Let us check that the
 safety assumptions of Theorem \ref{theorem:reconstruction} are then
 satisfied. For this, we need to show that one can find $\theta \in
 \left[0,\frac{\pi}{6}\right]$ such that:
    \begin{align}
      \Theta \; &\leq \; \frac{\theta}{2} -  \arcsin \left(\frac{\rho+\delta}{\reach}\right), \label{eq:safety-Theta}\\
      \sepvar \; &> \; 8 (\delta \theta + \rho \theta^2) + 6 \delta + \frac{2\rho^2}{\reach}, \label{eq:safety-sep} \\
      \pro \; &> \; 8 (\delta \theta + \rho \theta^2) \left( 1 + \frac{4 d \varepsilon}{\hei} \right), \label{eq:safety-pro} \\
      \pro^2 + \pro \sepvar &> \; \max \left\{10 \rho  \Theta (\varepsilon + \rho  \Theta) , \frac{4 J (1+J)}{(d+2)(d-1)! \, \Omega(\Delta_d)} \rho^2 \right\}, \label{eq:safety-additional}
    \end{align}
    where
    \[
    J = \frac{ (\reach+ \rho)^d  }
    {\left( \reach- \rho \right)^d \:   \left(\cos \Theta  \right)^{\min\{d, N-d\}}} -1.
    \]
    By Lemma~\ref{lemma:general-angle-bound} in Appendix~\ref{appendix:submanifold-generalities}, we obtain that
    \[
    \Theta \leq \arcsin\left( \frac{2 d}{\hei} \left(\frac{3 \rho^2}{\reach}+\delta \right) \right).
    \]
    Hence, since there exists a positive constant $c_1$ such that
    $\hei > c_1 \left( \frac{\rho}{\reach} \right)^{\frac{1}{3}}
    \rho$, we deduce that there exists a positive constant $c_3$ such
    that for $\frac{\rho}{\reach}$ small enough we have
    \[
    \Theta  \leq c_3 \left( \frac{\rho}{\reach} \right)^{\frac{2}{3}}.
    \]
    Let $\theta = 3 \Theta$ and observe that for $\frac{\rho}{\reach}$
    small enough, $\theta \in \left[0,\frac{\pi}{6}\right]$. With this
    choice of $\theta$ and using $\sepvar >
    \frac{\eta_0}{C_{\operatorname{ste}}} \rho$, $\pro > c_2 \left(
    \frac{\rho}{\reach} \right)^{\frac{1}{3}} \rho$, $\varepsilon =
    \frac{21}{20 C_{\operatorname{ste}}} \rho$, $\delta =
    \frac{2\rho^2}{\reach}$, and $\hei > c_1 \left(
    \frac{\rho}{\reach} \right)^{\frac{1}{3}} \rho$, it is easy to
    check that for $\frac{\rho}{\reach}$ small enough, Inequalities
    \eqref{eq:safety-Theta}, \eqref{eq:safety-sep},
    \eqref{eq:safety-pro}, and \eqref{eq:safety-additional} hold and
    therefore the safety assumptions of Theorem \ref{theorem:reconstruction} are
    met.
\end{proof}

\section{Conclusion}

We have shown that the submanifold reconstruction problem can be
recast as a weighted $\ell_1$-norm minimization problem under linear
constraints and as such is solvable by linear programming.

In the future, it would be interesting to study variants of this
minimization problem. For instance, one could imagine constraining the
solution to be a homology representative $d$-cycle (instead of a
normalized $d$-cycle).  Indeed, when $\M$ is orientable and connected,
its $d$-homology group with real coefficients is one-dimensional, and
the homology class of a normalized generator of it is called the
\emph{manifold fundamental class}. Furthermore, it is known that if
$K$ is either the \v Cech complex of $P$ or the Vietoris-Rips complex
of $P$, then $K$ and $\M$ are homotopy equivalent, assuming that $P$
samples sufficiently densely and accurately the manifold $\M$ and for
a careful choice of the scale parameter of these complexes
\cite{chazal08:_smoot_manif_recon_from_noisy,chazal2009sampling,attali2013vietoris,niyogi08:_findin_homol_of_subman_with,kim2019homotopy}.
Hence, it follows that the $d$-homology group of $K$ is also
one-dimensional. One representative of a generator of the fundamental
class of $\M$ can then be obtained, up to a multiplicative constant,
by taking any non-boundary $d$-cycle $\gamma_0$ of $K$ (performing for
this standard linear algebra operations on the boundary operators
$\partial_d$ and $\partial_{d+1}$ of $K$). Thus, a variant to our
problem can be expressed as follows: among all the $d$-chains of $K$
homologous to $\gamma_0$, search for the one with smallest Delaunay
energy.  We believe that it would be possible to adapt the proof
presented in the paper and establish conditions under which the
solution to this variant is also a $d$-chain which carries a
triangulation of $\M$ (namely, the Delloc complex of $P$).




\appendix


\clearpage
\section{A linear programming formulation}
\label{appendix:linear-programming}

For the sake of completeness, we recall in this appendix how
minimizing a weighted $\ell^1$-norm under linear constraints can be expressed
as a linear programming problem through slack variables. Consider the following minimization problem:
  \begin{mini*}
    {}{ \sum_i \omega_i | x_i | }{}{}
    \addConstraint{A x}{=b.}
  \end{mini*}
Here $x \in \Rspace^n$ is the variable, and $\omega_1, \omega_2, \ldots, \omega_n \in \Rspace$, $A \in \Rspace^{k \times n}$, $b \in \Rspace^k$ are parameters.
This problem is equivalent to the linear programming problem:
  \begin{mini*}
    {}{ \sum_i \omega_i s_i }{}{}
    \addConstraint{A x}{=b}
    \addConstraint{x_i}{\leq s_i, \quad}{i = 1,\ldots,n}
    \addConstraint{x_i}{\geq - s_i, \quad}{i = 1,\ldots,n,}
  \end{mini*}
where the variables are $x \in \Rspace^n$ and $s \in \Rspace^n$. Each
$s_i$ is called a {\em slack variable}.


\clearpage
\section{Angle between affine spaces}\label{appendix:AnglebetweenAffineSpaces}


In this appendix, we start by recalling how the angle between two
vector spaces and two affine spaces are defined, following for this
\cite{BSMF_1875__3__103_2}. The appendix presents classical results
(see \cite{BSMF_1875__3__103_2} and also the Wikipedia page untitled
"Angles between flats"), except for Lemma
\ref{lemma:RotationBetweenVectorSubspaces} which provides an explicit
expression of the path of the orthonormal frame.

\begin{definition}\label{definition:DefinitionAngleAffineSpaces}
The {\em angle} between two vector subspaces $V_1$ and $V_2$ of same Euclidean space 
is defined as:
\begin{equation}\label{eq:DefinitionAngleAffineSpaces}
\angle V_1, V_2  \defunder{=} \sup_{\Above{v_1\in V_1}{\|v_1\|=1}} \inf_{\Above{v_2\in V_2}{\|v_2\|=1}}  \angle v_1, v_2.
= \max_{\Above{v_1\in V_1}{\|v_1\|=1}} \min_{\Above{v_2\in V_2}{\|v_2\|=1}}  \angle v_1, v_2
\end{equation}
The {\em angle} between affine subspaces $A_1$ and $A_2$ is defined as the angle between their associated vector spaces.
\end{definition}

One gets trivially the  equivalent definition:
\begin{equation}\label{eq:DefinitionAngleAffineSpacesAlternative}
\angle V_1, V_2 = \inf \big\{ \theta \geq 0, \:  \forall v_1 \in V_1 \setminus \{0\}, \: \exists v_2 \in V_2 \setminus \{0\}, \: \angle v_1, v_2 \leq \theta \big\}.
\end{equation}
Since when $\dim V_1= \dim V_2$ there is an isometry (mirror symmetry) 
 that swaps $V_1$ and $V_2$ and preserves angles, we get the following implication:
\[
\dim V_1= \dim V_2 \implies \angle V_1, V_2 = \angle V_2, V_1,
\]
and one gets from \eqref{eq:DefinitionAngleAffineSpacesAlternative} and the triangular inequality on angles between vectors that:
\[
\angle V_1, V_3 \leq \angle V_1, V_2 + \angle V_2, V_3.
\]

\begin{lemma}[Minimal angle corresponds to orthogonal projection]\label{lemma:MinumalAngleAndOrthogonalPtrojection}
Let $V \subseteq  \Rspace^N$ be a vector subspace and $\pi_V$ the orthogonal projection on $V$.
Let $v' \in \Rspace^N$ a vector such that $\|v'\| = 1$, $\pi_V(v') \neq 0$ and $\theta = \min_{v^\star \in V, \| v^\star \| = 1} \angle v^{\star}, v'$.
Then 
\[
\arg \min_{v^\star \in V, \| v^\star \| = 1} \angle v^{\star}, v'  = \frac{1}{\cos \theta} \, \pi_V( v').
\]
\end{lemma}
\begin{proof}
One has by definition of $\pi_V$:
\begin{equation}\label{eq:ByDefinitionOfProjection}
\pi_V( v') =  \arg \min_{v^\star \in V} ( v^{\star} - v' )^2 
\end{equation}
and since $\|v'\|= 1$ and $\pi_V( v')- v'$ is orthogonal to $V$, one has $(\pi_V( v') -v')^2 = \left( \sin   \angle \pi_V( v'), v'  \right)^2$.
Also, for any vector  $v^\star \neq 0$,  one has $\left( \sin   \angle v^\star , v'  \right)^2 = \min_{\lambda} \left(\lambda v^\star - v' \right)^2$ so that
 \eqref{eq:ByDefinitionOfProjection} implies that $\pi_V(v')$ (as well as all its positively collinear vectors) minimises $v^{\star} \rightarrow \sin^2   \angle v^{\star}, v' $ in $V$.
 It follows that $\pi_V(v')$ is collinear to $\arg \min_{v^\star \in V, \| v^\star \| = 1} \angle v^{\star}, v'$ and since its norm is $\cos \theta$ we get the result.
\end{proof}

If $V$ is a vector subspace of $\Rspace^N$, and $\pi_V$ the orthogonal projection on $V$ then it is well known that:
\begin{itemize}
\item $\pi_V$ is self-adjoint and therefore its matrix in any orthonormal frame is symmetric.
\item $\pi_V \circ \pi_V = \pi_V$
\item the kernel of $\pi_V$ is the vector space normal to $V$ and its restriction to $V$ is the identity. 
\end{itemize}
Let $V$ and $V'$ be two $d$-dimensional vector subspaces of $\Rspace^N$ such that $\theta = \angle V',V < \pi/2$ and let $\pi_V$ and $\pi_{V'}$ be their corresponding  
orthogonal projections. Thanks to Lemma \ref{lemma:MinumalAngleAndOrthogonalPtrojection},  and since the cosinus function is decreasing on $[0, \pi/2]$,
one has:
\[
\cos \theta = \min_{v' \in V', \|v'\|=1}  \| \pi_V(v') \|
\]
So that, since $v' \in V' \Rightarrow v' =  \pi_{V'} (v')$:
\[
(\cos \theta)^2 = \min_{v' \in V', \|v'\|=1} \langle  \pi_V \circ \pi_{V'} (v') | \pi_V \circ \pi_{V'} (v') \rangle
\]
Denoting $M_V$ and $M_{V'}$ the respective symmetric matrix of $\pi_V$ and $\pi_{V'}$ in some orthonormal basis and since $M_V M_V = M_V$, we obtain:
\begin{align*}
(\cos \theta)^2 &= \min_{v' \in V', \|v'\|=1}   \left( M_V M_{V'} v' \right)^t M_V M_{V'} v' \\
&= \min_{v' \in V', \|v'\|=1}  v'^t  M_{V'}^t  M_V M_V M_{V'} v' \\
&= \min_{v' \in V', \|v'\|=1}  v'^t  M_{V'}^t  M_V M_{V'} v' 
\end{align*}
Since  $M_{V'}^t= M_{V'}$ and $v'\in V'$, we obtain that $M_{V'} v'  = v'$ and therefore
\[
(\cos \theta)^2 =  \min_{v' \in V', \|v'\|=1}  v'^t  M_{V'}  M_V v'.
\]
Let $A_{V'} : V \rightarrow V'$ be the restriction of $M_{V'}$ to $V$ 
and  $A_V : V' \rightarrow V$ the restriction 
of $M_V$ to $V'$. One has:
\begin{equation}\label{eq:LowerBoundRayleighQuotient}
(\cos \theta)^2 =  \min_{v' \in V', \|v'\|=1}  v'^t  A_{V'}  A_V v'
\end{equation}

Since $M_{V'}^t  M_V M_V M_{V'} : \Rspace^N \rightarrow V' \subseteq   \Rspace^N$
 is self-adjoint, so is its restriction  $C' = A_{V'}  A_V : V' \rightarrow V'$.
 It follows that $C'$ has $d$ (counting multiplicities)  real eigenvalues,
 associated to $d$  eigenvectors of $C'$ making an orthogonal basis of $V$,  and \eqref{eq:LowerBoundRayleighQuotient}
 is the Rayleigh quotient of $C'$ which gives that the smallest eigenvalue of $C'$ is $(\cos \theta)^2$.
 Since $\theta < \pi/2$ we have that all eigenvalues of $C'$ are positive. In particular, $C'$ is invertible.

It follows that  $A_V$ and $A_{V'}$ have rank $d$ and 
are also invertible so that $C= A_V A_{V'}   : V \rightarrow V$
is also invertible. If $v'_i$ is an eigenvector of $C'$ with eigenvalue $\lambda_i$, 
then $ C' v'_i = A_{V'}  A_V v'_i = \lambda_i v'_i$ and:
\begin{equation}\label{eq:EigenVectorOfC}
A_V A_{V'}  A_V v'_i = \lambda_i  A_V v'_i
\end{equation}
Since $A_V$ is invertible, $ A_V v'_i \neq 0$ and \eqref{eq:EigenVectorOfC} says that
$v_i =  A_V v'_i $ 
 is an eigenvector of $C= A_V A_{V'}$ with eigenvalue $\lambda_i$:
\[
C  v_i =A_V A_{V'}  v_i = \lambda_i  v_i
\]
Also, since $A_V $ and $A_{V'}$ have their $L^2$ operator norms upper bounded by $1$, so is the operator norm of $C$ and $C'$.
We have shown that:
\begin{lemma}\label{lemma:CorrespndanceEigenVectorsCPrimeC}
The orthogonal projection $A_V : V' \to V$ sends an orthogonal basis of $V'$ made of eigenvectors of $C'= A_{V'}  A_V$ 
to an orthogonal basis of $V$ made of eigenvectors of $C= A_V A_{V'}$ with the same eigenvalues.
These eigenvalues are included in $\left[ \left( \cos  \theta \right)^2 , 1 \right]$ with the smallest one being equal to $\left( \cos  \theta \right)^2$.
\end{lemma}

\begin{lemma}[Rotation between two vector spaces]\label{lemma:RotationBetweenVectorSubspaces}
Let $V$ and $V'$ be $d$-dimensional vector subspaces of Euclidean space such that the angle  $\theta = \angle V, V'$ satisfies:
\[
0 < \angle V, V' < \frac{\pi}{2}
\]
and $d' = d - \dim( V \cap V')$.
Then there is an orthonormal  basis $v_1,\ldots ,v_{d'}, v'_1, \ldots, v'_{d'}, w_1, \ldots, w_{d-d'}$ 
and a sequence of angles $\theta_1 \geq \theta_2 \geq \ldots, \theta_{d'} > 0$
such that $\theta_1 = \theta$,
\[
v_1, \ldots , v_{d'} , w_1, \ldots, w_{d-d'} 
\]
is a basis of $V$ and:
\[
\cos \theta_1 v_1 + \sin \theta_1 v'_1, \ldots ,\cos \theta_{d'} v_{d'} + \sin \theta_{d'} v'_{d'} , w_1, \ldots, w_{d-d'} 
\]
is a basis of $V'$.
\end{lemma}
\begin{proof}
We first claim that, for $v \in V$ one has:
\begin{equation}\label{eq:EigenSpace1IsIntersection}
 C v = v \iff A_{V'} v = v \iff v \in V \cap V'
\end{equation}

Indeed, if $v \in V \cap V'$ one has trivially $A_{V'} v  = v$ and $C v = A_V A_{V'} v = A_V v = v$. In the other direction,
if $C V = A_V A_{V'} v =  v$, since the operator norm of $A_{V'}$ is $1$, one must have $\| A_{V'} v \| \leq \| v \|$
and, since the operator norm of $A_V$ is $1$, one must have $\|  v \| = \| A_V A_{V'} v \| \leq \| A_{V'} v  \|$.
Therefore one has $\| A_{V'} v \| = \| v \|$. 
But since $\| A_{V'} v \| = \| v \| \cos \angle v, A_{V'}v$ 
we get $\cos\angle v, A_{V'}v = 1$ and $\angle v, A_{V'}v = 0$. This with $\| A_{V'} v \| = \| v \|$ gives  $A_{V'}v = v \in V \cap V'$.

It follows from \eqref{eq:EigenSpace1IsIntersection} that the eigenspace of $C$ corresponding to the eigenvalue $1$ coincides with $V \cap V'$.
We sort the eigenvalues of $C$  in increasing order (see Lemma \ref{lemma:CorrespndanceEigenVectorsCPrimeC}):
\[
\left( \cos  \theta \right)^2 = \lambda_1 \leq \lambda_2  \leq \ldots \leq \lambda_d,
\]
with $\lambda_{d'+1}= \ldots =\lambda_d = 1$ for $d' = d - \dim( V \cap V')$.

For any $k$,  $1\leq k \leq d'$, we define $v_k$ as an unit eigenvector 
associated\footnote{For multiple eigenvalues we choose orthogonal unit vectors as the eigenspace basis.} 
to  the eigenvalue $\lambda_k$.
For any $l$, $1 \leq l \leq d-d'$ we define $w_l \in V\cap V'$  such that $w_1,\ldots,w_{d-d'}$ is an orthonormal 
basis of $V\cap V'$, where $w_l$ is associated to the eigenvalue $\lambda_{d'+l} = 1$ by \eqref{eq:EigenSpace1IsIntersection}.
Then $(v_1, \ldots , v_{d'} , w_1, \ldots, w_{d-d'})$ is an orthonormal basis of $V$ and,
from  \eqref{eq:EigenSpace1IsIntersection},  $(w_l)_{ d'\leq l\leq d}$ is an orthonormal  of $V \cap V'$.

For $1\leq k \leq d'$ define $\theta_k$ as $\theta_k = \angle v_k, A_{V'}v_k$. We have seen that $\theta_1 = \theta$ and one has $0 < \theta_k \leq \theta< \pi/2$ and:
\[
\theta_k = \angle v_k, A_{V'}v_k =  \angle \lambda_k v_k, A_{V'}v_k = \angle C v_k, A_{V'}v_k = \angle A_V (A_{V'}v_k), A_{V'}v_k
\]
It follows that:
\[
\lambda_k \| v_k \| =\| C v_k \| = \cos \theta_k \| A_{V'}v_k \| =  \cos^2 \theta_k \|v_k \| 
\]
So that $\left( \cos \theta_k \right)^2= \lambda_k$.

We define now for $1\leq k \leq d'$:
\[
v'_k =  \frac{ A_{V'}v_k - \langle v_k, \,  A_{V'}v_k \rangle v_k}{\left\| A_{V'}v_k - \langle v_k, \, A_{V'}v_k \rangle v_k \right\|}
\]
where the denominator is no zero since $\theta_k > 0$.
One has by construction that $\|v'_k\| = 1$, $v'_k$ is orthogonal to $v_k$ and:
\begin{equation}\label{eq:ProjectionVkExpressedWithVkPrime}
\frac{A_{V'}v_k}{\left\| A_{V'}v_k \right\|} = \cos \theta_k v_k + \sin \theta_k v'_k
\end{equation}

Since $\left( \left( \frac{A_{V'}v_k}{\left\| A_{V'}v_k \right\|} \right)_{1\leq k \leq d'}, w_1,\ldots, w_{d-d'} \right)$ are unit eigenvectors of $C'$
they form an orthonormal basis of $V'$.
 In order to complete the proof, it remains to prove that, for $k'\neq k$, $v_k$ is orthogonal to $v'_{k'}$.
 One has from \eqref{eq:ProjectionVkExpressedWithVkPrime}, 
 using that $A_V A_{V'} v_{k}  = \cos^2   \theta_{k} v_{k}$ 
 and $\left\| A_{V'}v_{k} \right\| =  \cos   \theta_{k}$:
\begin{align*}
 \sin \theta_{k} v'_{k} &= \frac{A_{V'}v_{k}}{\left\| A_{V'}v_{k} \right\|} - \cos \theta_{k} v_{k}  
 =   \frac{A_{V'}v_{k}}{\left\| A_{V'}v_{k} \right\|} - \cos \theta_{k}   \frac{A_{V} \left(A_{V'}v_{k} \right)}{\cos^2   \theta_{k} } \\
 &= \frac{A_{V'}v_{k}}{\left\| A_{V'}v_{k} \right\|} - A_{V} \left( \frac{A_{V'}v_{k}}{\left\| A_{V'}v_{k} \right\|} \right)
 \end{align*}
Recall thet $A_V$ is the orthogonal projection on $V$ and the last equality shows then  that $ v'_{k}$ is orthogonal to $V$ 
and therefore orthogonal to any $v_{k'}$ for $1\leq k' \leq d'$.
\end{proof}


\clearpage
\section{$C^2$-submanifold of Euclidean space}
\label{appendix:submanifold-generalities}

We recall the definition of a smooth submanifold; see for instance \url{https://maths-people.anu.edu.au/~andrews/DG/DG_chap3.pdf}.
 
\begin{definition}
A subset  $\M$ of $\Rspace^N$ is a $d$-dimensional $C^2$ submanifold if for 
every point $x$ in $\M$ there exists a neighborhood $V$ of $x$ in $\Rspace^N$,
 an open set $U \subseteq  \Rspace^d$ and a $C^2$  map $\xi : U \rightarrow \Rspace^N$ such that
 $\xi$  is a homeomorphism onto $\M \cap V$, and the differential $\Differential_y \xi$  is injective for every $y \in U$.
\end{definition}
A $d$-dimensional $C^2$ submanifold is a $C^2$ manifold topologically
embedded in $\Rspace^N$ but the converse is not true in general.  For
instance, the square is not a submanifold of the plane. Yet, it
is a $C^\infty$ manifold topologically embedded in the
plane. Precisely, the square can be defined as the image of the circle
$\Sspace^1$ (which is a $C^\infty$ manifold) by a $C^\infty$ map that
is non-regular ({\em i.e.} whose derivative is non-injective)
precisely at the four corners of the square.

A compact $C^2$ submanifold has positive reach \cite{federer-59}. 
Moreover one has (see for example Lemma 4 and following paragraph in \cite{boissonnat2019reach} or 
Proposition 2.3 in \cite{aamari2019estimating}) that:

\begin{lemma}[The inverse of the reach bounds the curvature]\label{lemma:InverReachBoundsCurvature}
If $\M \subseteq  \Rspace^N$ is a $C^2$ $d$-dimensional submanifold with reach $ \Reach (\M) >0$ then $1/ \Reach (\M)$
bounds the (absolute values of the ) principal curvatures at $m \in \M$ in the direction $v$ for any  vector $v$ in the space normal to $\M$ at $m$.
In particular  $1/ \Reach (\M)$ bounds the principal curvatures when $\M$ has codimension $1$.
\end{lemma}

The following lemma,  due to Federer, bounds the distance of a point $q\in \M$ to the tangent space at a point $p \in \M$. 
It holds for any set with positive reach and in particular  for $C^2$  submanifolds.
\begin{lemma}[Distance to tangent space, Theorem 4.8(7) of \cite{federer-59}] \label{lemma:distanceToTangent}
Let $p,q \in \M \subseteq  \Rspace^N$ such that $\|p-q\| < \Reach(\M)$. We have 
\begin{align} 
\sin \angle \left([pq], T_p \M \right) \leq \frac{\|p-q\|}{2\,  \Reach (\M)},
\label{eq:AngleSegmentToT}
\end{align} 
and 
\begin{align}
d(q , T_p \M )  \leq \frac{\|p-q\|^2}{2\, \Reach(\M) }.
\label{eq:DistToTangent}
\end{align}
\label{TangentSpaceLine1}
\end{lemma}

Next lemma bounds the angle variation for $C^2$ manifolds (a slightly weaker condition si given for $C^{1,1}$ manifolds in the same paper):
\begin{lemma}[Corollary  3 in \cite{boissonnat2019reach}]\label{lemma:TangentSpaceVariation}
For any $p,q \in \M$, we have
\begin{align}
\sin \left(\frac{\angle (T_p\M, T_q \M )}{2} \right)
\leq \frac{\|p-q\|}{2\Reach{\M}} .
\nonumber
\end{align}
\end{lemma}

Using Lemmas \ref{lemma:distanceToTangent} and \ref{lemma:TangentSpaceVariation} we can show that the projection on a tangent space
defines a local chart for $\M$.
Indeed, if $m \in \M$ then for any $p,q  \in  \M \cap B\big(m,  \sin (\pi/4) \Reach{\M} \big)$, \eqref{eq:AngleSegmentToT} gives:
\[
\sin \angle ([pq], T_p \M) \leq \frac{\|p-q\|}{2\,  \Reach (\M)} \leq   \frac{2 \sin (\pi/4) \Reach{\M}}{2\,  \Reach (\M)} = \sin (\pi/4)
\]
So that:
\begin{align}
\angle ([pq], T_p \M) \leq  \pi/4, \label{eq:BounfAnglePQTpM} 
\end{align}
and Lemma \ref{lemma:TangentSpaceVariation} gives:
\[
\sin \left(\frac{\angle (T_p\M, T_m\M )}{2} \right)
\leq \frac{\|p-m\|}{2\Reach{\M}} \leq \frac{\sin (\pi/4) \Reach{\M} }{2\Reach{\M}} =  \frac{\sin (\pi/4) }{2}
\]
so that:
\[
\angle (T_p\M, T_m\M ) \leq  2 \arcsin \left(\frac{\sin (\pi/4) }{2}\right) <  \pi/4
\]
Summing with \eqref{eq:BounfAnglePQTpM} gives: 
\[
\angle ([pq], T_m \M) < \pi/2
\]
and this in turn implies:
\[
p\neq q \Rightarrow \pi_{T_m\M} (p) \neq \pi_{T_m\M} (q)
\]
We have shown that the restriction of $\pi_{T_m\M}$ to $\M \cap B(m,  \sin (\pi/4) \Reach{\M}$ is injective, 
and, by Invariance of Domain Theorem (\cite{brouwer1912invarianz}),
it is an homeomorphism on its image, which gives us:
\begin{lemma}[Projection on tangent space is a local chart]\label{lemma:TangentSpaceProjectionIsAMap}
If $\M$ is a compact $C^2$ submanifold of Euclidean space, and $m\in \M$ then, identifying $T_m\M$ with $\Rspace^d$ through a given frame, 
the restriction of $\pi_{T_m\M}$ to $\M \cap B\left(m,  \sin (\pi/4) \Reach{\M}\right)^\circ$ is a local chart of $\M$.
Moreover, for any $p \in \M \cap B\left(m,  \sin (\pi/4) \Reach{\M}\right)^\circ$, one has:
\[
\angle (T_p\M, T_m\M  )  <  \pi/4
\]
\end{lemma}

Following \cite[Notation 1.1]{boissonnat2020local}, we call Thickness of a $k$-simplex:
\[
t(\sigma) = \frac{h}{k L}
\]
where $h$ is the smallest altitude of $\sigma$ and $L$ the length of the longest edge.
Then we have (adaptation of \cite[Section IV.15]{whitney2005geometric}, proven in \cite[Lemma 8.11]{boissonnat2018geometric}):
\begin{lemma}[Angle between  simplices  tangent space]\label{lemma:UpperBoundAngleTangentSpaceSimplex}
If $\M$ is a compact $C^2$ submanifold of Euclidean space, $\sigma$ a $d$-simplex with vertices in $\M$, and $p$ a vertex of $\sigma$:
then:
\[
\angle \Aff{\sigma} , T_p\M \leq \frac{L}{t(\sigma) \Reach{\M}} = \frac{d L^2}{h \Reach{\M}}
\]
\end{lemma}

We shall also need the Whitney angle bound established in \cite{BOISSONNAT_2013}.

\begin{lemma}[Whitney angle bound {\cite[Lemma 2.1]{BOISSONNAT_2013}}]
  \label{lemma:whitney-angle-bound}
  Consider a $d$-dimensional affine space $H$ and a simplex $\sigma$
  such that $\dim\sigma \leq d$ and $\sigma \subseteq \Offset H t$ for
  some $t \geq 0$. Then
  \[
  \sin \angle(\Aff\sigma,H) \leq \frac{2 t \dim(\sigma)}{\height{\sigma}}
  \]
\end{lemma}

Building on these results, we derive yet another  bound between the affine
space spanned by a simplex and a nearby tangent space.

\begin{lemma}
  \label{lemma:general-angle-bound}
  Consider a non-degenerate $\rho$-small simplex $\sigma \subseteq
  \Offset \M \delta$ with $16\delta \leq \rho \leq
  \frac{\reach}{3}$. Let $x \in \Conv \sigma$. Then,
  \[
  \angle(\Aff \sigma, \Tangent {\pi_\M(x)} \M) \leq
  \arcsin \left(
  \frac{2\dim(\sigma)}{\height{\sigma}}
  \left(\frac{3\rho^2}{\reach} + \delta \right) \right).
  \]
\end{lemma}

\begin{proof}
  Let $v \in \sigma$. Write $v^\ast = \pi_\M(v)$ and $x^\ast =
  \pi_\M(x)$.  We know from \cite[page 435]{federer-59} that for $0
  \leq h < \Reach{\M}$, the projection map $\pi_\M$ onto $\M$ is
  $\left( \frac{\reach}{\reach - h} \right)$-Lipschitz for points at
  distance less than $h$ from $\M$. By Lemma
  \ref{lemma:small-tubular-neigborhood}, both $x$ and $v$ belong to
  $\Offset \M h$ for $h = \frac{\rho}{4}$, we thus have
  \begin{equation*}
    \| v^\ast - x^\ast \| \leq \frac{\reach}{\reach - \frac{\rho}{4}} \times \| v - x\| 
    \leq  \frac{\reach}{\reach - \frac{\reach}{3\times 4}} \times 2\rho = \frac{24}{11} \rho \leq \sqrt{6} \rho
  \end{equation*}
  Applying Lemma \ref{lemma:distanceToTangent}, we get that
  \[
    d(v,\Tangent {x^\ast} \M) \leq d(v^\ast,\Tangent {x^\ast} \M) + \|v - v^\ast\| 
    \leq \frac{\|v^\ast - x^\ast\|^2}{2\reach} + \delta 
    \leq \frac{3\rho^2}{\reach} + \delta.
  \]
  Hence, $\sigma \subseteq \Offset {(\Tangent {x^\ast} \M)} t$ for $t =
  \frac{3\rho^2}{\reach} + \delta$ and applying Whitney angle bound
  (Lemma \ref{lemma:whitney-angle-bound}), we conclude that $\sin
  \angle(\Aff\sigma,\Tangent {x^\ast} \M) \leq
  \frac{2\dim(\sigma)}{\height{\sigma}} \left(\frac{3\rho^2}{\reach} +
  \delta \right) $.
\end{proof}

%

\clearpage
\section{Smooth local parametrization of the normal bundle}

Let $\M \subseteq  \Rspace^N$ be a $C^2$ $d$-dimensional submanifold  with reach greater than $\rho >0$.
If $x \in \M$, the vector spaces $T_x\M$ and $N_x\M$ are the respective tangent and normal spaces to $\M$ at $x$.
To avoid confusion, we denote by $\TanSpace_x\M$ and $\NorSpace_x\M$ the corresponding
affine subspaces of ambient $\Rspace^N$:  $\TanSpace_x\M = x + T_x\M$ and $\NorSpace_x\M =x + N_x\M$.

We denote by $\M^{\oplus \rho}$  the $\rho$-tubular neighborhood of $\M$.
The normal bundle of $\M$ is denoted $N\M$ and defined as:
\[
N\M \defunder{=} \coprod_{x\in \M} N_x\M
\]
where $\coprod$ is the disjoint union. The normal bundle restricted to radius $\rho$ is denoted $N^\rho\M$:
\[
N^\rho\M \defunder{=} \coprod_{x\in \M} \left\{ (x, v) \in N_x\M,\, \| v \| \leq \rho \right\}
\]

We know from \cite[Item (13) of Theorem 4.8]{federer-59} that, since $\rho < \Reach{\M}$,
the map $\Psi : \M^{\oplus \rho} \rightarrow  N^\rho\M \subseteq   \Rspace^N \times \Rspace^N$ defined by:
\[
\Psi(y) = \left( \pi_\M(y), y-\pi_\M(y) \right)
\]
is a Lipschitz  homeomorphism between $\M^{\oplus \rho}$  and $N^\rho\M$
whose inverse $\sigma: N^\rho\M \rightarrow \M^{\oplus \rho}$ defined as
\begin{equation}\label{eq:DefinitionFunctionSigma}
\sigma(x,v) = x+v
\end{equation}
is Lipschitz as well.
We have \cite{foote1984regularity}  that inside  $\M^{\oplus \rho}$  the distance function $d_\M$  is $C^2$ and  the projection $\pi_\M$ is $C^1$.
It follows that $\Psi$ is $C^1$ and since its inverse is Lipschitz, its Jacobian cannot be singular.
It follows that $\Psi$ is a $C^1$ diffeomorphims between $\M^{\oplus \rho}$ and $N^\rho\M$.

For $x\in \M$ we consider an open neighborhood $U_0$ of $0$ in $\Rspace^d$ and a $C^2$
injective map $\xi: U_0 \rightarrow \M$,  regular in the sense that
    differential $\left(\frac{d \xi}{d u^j}(u) \right)_{ j=1,d}$  has rank $d$ for any $u\in U_0$,
and such that  $\xi(0) = x$.
Let us choose $\xi$ such that, moreover, $\left( \frac{d \xi}{d u^j} \restrict{u=0} \right)_{ j=1,d}$ forms an orthonormal basis of $T_x\M$. 

Also, we consider, as in the proof of Lemma 6.3 of \cite{milnor2016morse}, a set of $(N-d)$ $C^1$-smooth 
vector valued maps $(w_k)_{k=1 \ldots N-d}$,
where $w_k: U_0 \rightarrow \Rspace^N$ and such that, for any $u\in U_0$,   $(w_k(u))_{k=1, N-d}$
is an orthonormal basis of $N_{\xi(u)}\M$. 

 As done by J.Milnor in \cite[paragraph 6, proof of Lemma 6.3]{milnor2016morse}, $\xi$ and $(w_k)_{k=1 \ldots N-d}$
 defines a local trivialization of the normal bundle $N\M$, that is a chart $\phi$ of 
 $N\M$ in the neighborhood of $(x,0)$
 where the parameter $\left(u^1,\ldots,u^d, t^1,\ldots, t^{N-d} \right) \in U_0 \times  \B^{N-d}(0, \rho)$ corresponds to the point:
 \begin{align*}
\phi(u,t) &=  \phi\left(u1,\ldots,u^d, t^1,\ldots, t^{N-d} \right) \\ 
&=  \left( \xi \left(u^1,\ldots,u^d\right), \: \sum _{k=1}^{N-d} t^k  w_k (u^1,\ldots,u^d)\right) \in N^\rho \M
 \end{align*}
 and
\begin{equation}\label{eq:SigmaComposedWithPhi}
 \sigma \circ \phi \left(u,t \right)  = \xi \left(u1,\ldots,u^d \right) + \sum _{k=1}^{N-d} t^k  w_k (u^1,\ldots,u^d) \in \M^{\oplus \rho} \subseteq  \Rspace^N
\end{equation}
 Derivating \eqref{eq:SigmaComposedWithPhi} at $u=0$ gives:
 \begin{align}
 \frac{d}{d u^j}\restrict{u=0} \sigma \circ \phi \left(u,t \right) &=  \frac{d \xi}{d u^j}\restrict{u=0} + \sum _{k=1}^{N-d} t^k  \frac{d w_k}{d u^j}\restrict{u=0}\label{eq:DerivativeSigmaPhiWrtU}\\
 \frac{d}{d t^k}\restrict{u=0}  \sigma \circ \phi \left(u,t \right) &=  w_k (0) \label{eq:DerivativeSigmaPhiWrtT}
 \end{align}

In order to express the derivative of  $\sigma: N^\rho\M \rightarrow \M^{\oplus \rho}$
 at point $(x,v) \in N^\rho\M$ we need, besides the chart $\phi$ of a neighborhood of $(x,v)$, 
 a chart $\hat{\phi}$ of a neighborhood of $\sigma(x,v) = x + v$ in euclidean space.
 A natural choice for $\hat{\phi}$ is
 \[
 \hat{\phi}  \big( y \big) = \hat{\phi}  \big( y^1, \ldots y^N \big) = x+v + \sum_{k=1}^{d}  y^k   \frac{d \xi}{d u^k}\restrict{u=0} + \sum_{k=1}^{N-d} y^{d+k} w_k
 \]
 Observe that, since $\big ( \frac{d \xi}{d u^1}\restrict{u=0}\ldots,  \frac{d \xi}{d u^d}\restrict{u=0}, w_1, \ldots, w_{N-d} \big)$ is an orthonormal  basis,
 seeing  $N^\rho\M$ and $\M^{\oplus \rho}$ as Riemannian manifolds, the metric tensor associated to the 
charts $\phi$ and $ \hat{\phi}$ respectively at $(x,v)$ and $\sigma(x,v) = x +v$ are the identity matrix.

One has:
\begin{equation}\label{equation:ExpressionDiffDXiAndW}
\Differential \restrict{y=0} \hat{\phi} = \Big ( \frac{d \xi}{d u^1}\restrict{u=0}\ldots,  \frac{d \xi}{d u^d}\restrict{u=0}, w_1, \ldots, w_{N-d} \Big)
\end{equation}
 and since te columns of \eqref{equation:ExpressionDiffDXiAndW} are unitary and pairwise orthogonal, 
 taking the inner product of \eqref{eq:DerivativeSigmaPhiWrtU} and \eqref{eq:DerivativeSigmaPhiWrtT} 
  with the columns of $\Differential \restrict{u=0} \hat{\phi} $, namely 
$\left(  \frac{d \xi}{d u^j}\restrict{u=0} \right)_{j=1,d}$ and $\left( w_k \right)_{k=1,N-d}$, gives an expression of the Jacobian of 
$\sigma: N^\rho\M \rightarrow \M^{\oplus \rho} \subseteq   \Rspace^N$
at the point  $\phi(0,t)$ in the charts $\phi$  and $\hat{\phi}$  as:
\begin{equation}\label{eq:JacobianSigma}
\Differential\restrict{u=0}\, \left( {\hat{\phi}}^{-1} \circ \sigma \circ \phi \right) =
\begin{pmatrix}
\left\langle \frac{d \xi}{d u^i},\, \frac{d \xi}{d u^j} \right\rangle + \sum _{k=1}^{N-d} t^k  \left\langle  \frac{d w_k}{d u^i} ,\, \frac{d \xi}{d u^j} \right\rangle & 
\sum _{k=1}^{N-d} t^k  \left\langle  \frac{d w_k}{d u^i} ,\, w_l \right\rangle \\
0& \UnitMatrix
\end{pmatrix}
\end{equation}
Where in both $N^\rho\M \subseteq  \Rspace^N$ and $\M^{\oplus \rho}$ the tangent space is decomposed as direct sum of the tangential and orthogonal fibers.

Denoting by $\mathrm{I}$ the first  fundamental form of $\M$:
\[
\mathrm{I} = \left\langle \frac{d \xi}{d u^i},\, \frac{d \xi}{d u^j} \right\rangle.
\]
Since we have chosen $\xi$ in such a way that
$\left( \frac{d \xi}{d u^j}\restrict{u=0} \right)_{j=1,d}$ is an orthonormal basis of $T_x\M$, we have at $x = \xi(0)$ that $\mathrm{I} = \UnitMatrix$.

For a unit vector 
$v\in N_x\M$  with 
 $v = \sum _{k=1}^{N-d} t^k_v w_k$ we can call, following J. Milnor again, the second fundamental form in the direction $v$:
\[
\mathrm{II}_v =  \sum _{k=1}^{N-d} t^k_v  \left\langle  \frac{d w_k}{d u^i} ,\, \frac{d \xi}{d u^j} \right\rangle.
\]
We also denote  the ``torsion''  term by $T_v$:
\[
T_v = \sum _{k=1}^{N-d} t^k_v  \left\langle  \frac{d w_k}{d u^i} ,\, w_l \right\rangle
\]
In fact we claim that,  without loss of generality, one can choose the maps $(w_k)_{k=1, N-d}$
in such a way that  $T_v=0$. Indeed, since, for any $u\in U_0$,  $ \left\langle  w_k(u) ,\, w_l(u) \right\rangle = \delta_{k,l}$ where 
$\delta_{k,l}$, the Kronecker delta, is constant, one has:
\[
0 =  \frac{d }{d u^i}   \left\langle  w_k ,\, w_l \right\rangle =  \left\langle  \frac{d w_k}{d u^i} ,\, w_l \right\rangle +  \left\langle  w_k,\,  \frac{d w_l}{d u^i}  \right\rangle
\]
in other words, $ \left\langle  \frac{d w_k}{d u^i} ,\, w_l \right\rangle$ is antisymmetric and can then be seen as an infinitesimal rotation, i.e.,  formally,
 an element of the Lie algebra of $SO(N-d)$.
It results that if we replace $w$ by $w'$ defined by :
\[
w'(u^1,\ldots, u^d) =   \exp \left( - \sum_{i=1}^d  u_i  \left\langle  \frac{d w_k}{d u^i}_{|_{u=0}} ,\, w_l \right\rangle \right) w(u^1,\ldots, u^d), 
\]
 we get:
\[
 \left\langle  \frac{d w'_k}{d u^i}_{|_{u=0}} ,\, w'_l \right\rangle =0 .
\]
Since the $\exp$ term is a rotation depending smoothly ($C^\infty$) on the $u_i$ and is applied to the basis $w$, $w$ can be replaced by $w'$ which proves the claim.

We have that for $\tau \in [0, \rho]$,  at the point $(x, \tau v) \in N^\rho\M$ with $\sigma(x, \tau v) = x + \tau v \in \M^{\oplus \rho}$ 
 the Jacobian \eqref{eq:JacobianSigma} can be expressed as, with the assumption $T_v=0$:
\begin{equation}\label{eq:JacobianSigmaWithFundForms}
\begin{pmatrix}
\UnitMatrix + \tau  \mathrm{II}_v& T_v \\
0& \UnitMatrix
\end{pmatrix}
=
\begin{pmatrix}
\UnitMatrix + \tau  \mathrm{II}_v& 0 \\
0& \UnitMatrix
\end{pmatrix}
\end{equation}

Thanks to Lemma \ref{lemma:InverReachBoundsCurvature},
$\|\mathrm{II}_v\|_2$, the operator $L^2$ norm  of $\mathrm{II}_v$, 
which is the maximal extrinsic curvature of $\M$, is upper 
bounded by the  inverse of the reach of $\M$:
\begin{equation}
\forall v \in N_x\M, \|v\| = 1, \: \left\|  \mathrm{II}_v \right\|_2 \leq \frac{1}{\Reach{\M}}
\end{equation}
Since in \eqref{eq:JacobianSigmaWithFundForms} one has $\tau \leq \rho < \Reach{\M}$, 
we get that $ \left\|  \tau \mathrm{II}_v \right\|_2 \leq \frac{\tau}{\Reach{\M}} < 1$ and $\UnitMatrix + \tau  \mathrm{II}_v$, 
the upper left bloc in  \eqref{eq:JacobianSigmaWithFundForms},
 is inversible.  It follows that matrix  \eqref{eq:JacobianSigmaWithFundForms} is inversible:
\begin{equation}\label{eq:JacobianPsi}
\Differential\restrict{y=0}\, \left( \phi^{-1} \circ \Psi \circ \hat{\phi} \right) =
\begin{pmatrix}
\UnitMatrix + \tau  \mathrm{II}_v& 0 \\
0& \UnitMatrix
\end{pmatrix}^{-1} =
\begin{pmatrix}
\left( \UnitMatrix + \tau  \mathrm{II}_v\right)^{-1}& 0 \\
0& \UnitMatrix
\end{pmatrix}
\end{equation}
By local inversion Theorem,  \eqref{eq:JacobianPsi} gives us the Jacobian of $\Psi = \sigma^{-1}$ at $x+ \tau v  \in \M^{\oplus \rho}$. 
In particular, the first row of \eqref{eq:JacobianPsi}:
\begin{equation}\label{eq:JacobianPiM}
\begin{pmatrix}
\left( \UnitMatrix + \tau  \mathrm{II}_v\right)^{-1}& 0 
\end{pmatrix}
\end{equation}
is the Jacobian of $\xi^{-1} \circ \pi_\M \circ \hat{\phi} $ at $x+ \tau v \in \M^{\oplus \rho}$.

\clearpage
\section{Restriction of $\pi_\M$ to a $d$-dimensional affine space in a neighborhood of $x+ \tau v  \in \M^{\oplus \rho}$}
\label{appendix:J}

Let $\Pi  \subseteq  \Rspace^N$ denote a  $d$-dimensional affine space that contains the point $y = x+ \tau v  \in \M^{\oplus \rho}$
for $x \in\M$, $v \in N_x\M$ a unit vector and $\tau < \rho$.
If $\left( y, e_1, \ldots, e_d \right)$ is an orthonormal frame of $\Pi$ centered at $y$. 
This frame defines a parametrization $\Rspace^d \rightarrow \Pi$ defined by 
\begin{equation}\label{equation:OrthonormalPlaneCoordinateSystem}
\left( z^1,\ldots, z^d \right) \mapsto y + \sum_{i=1}^d z^i e_i.
\end{equation}

Each vector $e_i$ can be decomposed uniquely as a sum  $e_i= e_i^T + e_i^N$ where $e_i^T \in T_x\M$ and $e_i^N \in N_x\M$. 
From the definition \eqref{eq:DefinitionAngleAffineSpaces} one has:
\begin{equation}\label{eq:ProjectionFromAngle}
\min_{ z \in \Rspace^d, \|z\| = 1} \left\| \sum_{i=1}^d z^i e_i^T \right\| 
=\min_{ z \in \Rspace^d, \|z\| = 1} \left\| \pi_{T_x\M} \left(\sum_{i=1}^d z^i e_i \right) \right\| = \cos \angle \Pi, T_x\M
\end{equation}

Using \eqref{eq:JacobianPiM} we get:
\begin{equation}\label{eq:JacobianProjectionFromAffineSpace}
\frac{d}{d z^j}\restrict{z=0} \: \xi^{-1} \circ \pi_\M \left(  y + \sum_{i=1}^d z^i e_i \right) 
= \left( \UnitMatrix + \tau  \mathrm{II}_v\right)^{-1} e_j^T 
\end{equation}
Since $\tau \leq \rho$ and the eigenvalues of $  \mathrm{II}_v$ are the principal curvatures, by Lemma \ref{lemma:InverReachBoundsCurvature}
they are upper bounded by $\reach^{-1}$, the inverse of the reach of \M.
The eigenvalues of the symmetric matrix $\left( \UnitMatrix + \tau  \mathrm{II}_v\right)^{-1} $ are therefore included in:
 \[
 \left[ \frac{\reach}{\reach+ \rho}, \, \frac{\reach}{\reach- \rho} \right]
 \]
 and therefore its determinant included by:
 \[
 \left[ \left( \frac{\reach}{\reach+ \rho} \right)^d , \,   \left( \frac{\reach}{\reach- \rho} \right)^d \right]
 \]
 From \eqref{eq:ProjectionFromAngle} we have that the determinant of 
\[
\left( z^1,\ldots, z^d \right) \mapsto  \pi_{T_x\M} \left( y + \sum_{i=1}^d z^i e_i \right) = \sum_{i=1}^d z^i e_i^T
\]
is included in:
\[
 \left[ \big(  \cos \angle \Pi, T_x\M \big)^d , \,  1 \right]
 \]
In fact, a finer analysis exploiting Lemma \ref{lemma:RotationBetweenVectorSubspaces}
 allows to improve this bound to:
 \[
 \left[ \big(  \cos \angle \Pi, T_x\M \big)^{\min(d, N-d)} , \,  \cos \angle \Pi, T_x\M  \right],
 \]
since, when $N< 2d$,  $T_x\M$  and  the vector space associated to $\Pi$ have a common subspace of dimension at least $2d-N$.

 Therefore, the determinant of the differential of  $\left( z^1,\ldots, z^d \right) \mapsto  \xi^{-1} \circ \pi_\M \left( y + \sum_{i=1}^d z^i e_i \right)$
 is bounded by:
 \begin{align}
& \det \left( \frac{d}{d z^j}\restrict{z=0} \: \xi^{-1} \circ \pi_\M \left(  y + \sum_{i=1}^d z^i e_i \right)  \right) \label{equation:determinantProjectionOnMRestrictedDFlat}\\
& \quad \in  \left[ \left(\frac{\reach}{\reach+ \rho} \right)^d \:   \left(\cos \angle \Pi, T_x\M  \right)^{\min(d, N-d)}  , \:   \left( \frac{\reach}{\reach- \rho} \right)^d   \cos \angle \Pi, T_x\M  \right] 
\nonumber
 \end{align}

If $\angle \Pi, T_x\M < \pi/2$, it follows from \eqref{eq:ProjectionFromAngle} that the set of vectors $\left(e_i^T\right)_{i=1,d}$
spans $T_x\M$. Since  $\left( \UnitMatrix + \tau  \mathrm{II}_v\right)^{-1}$ has full rank $d$ we get that  
$\frac{d}{d z^j}\restrict{z=0} \: \xi^{-1} \circ \pi_\M \left(  y + \sum_{i=1}^d z^i e_i \right)$
has full rank. We have proven that:
\begin{lemma} \label{lemma:ProjectionFromNonOrthogonalFlatIsAChart}
Let $\Pi  \subseteq  \Rspace^N$ denote a  $d$-dimensional affine space that contains the point $y  \in \M^{\oplus \rho}$.
If $\angle \Pi, T_{\pi_\M(y)}\M < \pi/2$, then, in some open neighborhood $U_y$ of $y$ in $\Pi$, 
$\pi_\M\restrict{U_y}$ is a $C^1$-diffeomorphism on its image.
 \end{lemma}
 

We consider now $\Pi_1,\Pi_2 \subseteq  \Rspace^N$, two  $d$-dimensional affine space,
where, for $i=1,2$ the affine space  $\Pi_i$ contains  the point $y_i = x+ \tau_i v_i  \in \M^{\oplus \rho}$
for $x \in\M$, $v_i \in N_x\M$ a unit vector and $\tau_i < \rho$.
We assume that, for $i=1,2$, one has $\angle \Pi_i,  T_{\pi_\M(y)}\M < \pi/3$.

By Lemma \ref{lemma:ProjectionFromNonOrthogonalFlatIsAChart}, 
the projection on $\M$ restricted to  some
neighborhood of $y_i$ in $\Pi_i$, is an homeomorphism.
For  $i=1,2$, let  $U_i$ be is an open neighborhood
of $y_i$ in $\Pi_i$ such that $\pi_\M\restrict{U_i}$ is an homeomorphism on its image and $U_i \subseteq  \Offset \M \rho$.
Assume moreover that $\pi_\M(U_1)= \pi_\M(U_2)$.
Then one can defines an homomorphism
$\varphi_{1 \to 2} : U_1 \to U_2$  as:
\[
\varphi_{1 \to 2} \defunder{=}  \left( \pi_{\M}\restrict{U_2}  \right)^{-1}   \circ  \pi_{\M}\restrict{U_1}
\]

  
One has, for  a chart $\xi: U_0 \rightarrow \M$, where $U_0 \subseteq  \Rspace^d$, and 
is such that  $\xi(U_0)  \supset \pi_\M(U_1)= \pi_\M(U_2)$, that:
 \begin{align*}
\varphi_{1 \to 2} &=  \left( \pi_{\M}\restrict{U_2}  \right)^{-1}   \circ  \pi_{\M}\restrict{U_1} \\
  &  \left( \pi_{\M}\restrict{U_2}  \right)^{-1} \circ \xi \circ \xi^{-1}  \circ \pi_{\M}\restrict{U_1} 
  \end{align*}
Choosing some  coordinate systems  for $\Pi_1$ and $\Pi_2$
defined by respective orthonormal frame $(y_1, e_{11}, \ldots,e_{1d})$ and $(y_2, e_{21}, \ldots,e_{2d})$,
 as in \eqref{equation:OrthonormalPlaneCoordinateSystem},
  we denote by 
  $D \varphi_{1 \to 2}$ the matrix of the derivative of $\varphi_{1 \to 2}$ in theses coordinate systems
  we have:
\begin{equation*}
  D\varphi_{1 \to 2}(y) =     \left( \frac{d}{d z_2^j}\restrict{z=0} \: \xi^{-1} \circ \pi_\M \left(  y + \sum_{i=1}^d z_2^i e_{2i} \right)  \right)^{-1} 
   \left( \frac{d}{d z_1^j}\restrict{z=0} \: \xi^{-1} \circ \pi_\M \left(  y_1 + \sum_{i=1}^d z_1^i e_{1i} \right)  \right) 
\end{equation*}
So that, using  \eqref{equation:determinantProjectionOnMRestrictedDFlat} ,
we get:
\begin{equation*}
\det  D\varphi_{1 \to 2}(y) \in \left[ 
\frac{\left(\frac{\reach}{\reach+ \rho} \right)^d \:   \left(\cos \angle \Pi_1, T_x\M  \right)^{\min(d, N-d)}   }
{  \left( \frac{\reach}{\reach- \rho} \right)^d   \cos \angle \Pi_2, T_x\M } , 
\frac{  \left( \frac{\reach}{\reach- \rho} \right)^d   \cos \angle \Pi_1, T_x\M } 
{\left(\frac{\reach}{\reach+ \rho} \right)^d \:   \left(\cos \angle \Pi_2, T_x\M  \right)^{\min(d, N-d)}   }
\right]    
\end{equation*}
In other words:

\begin{equation*}
\det  D\varphi_{1 \to 2}(y) \in \left[ 
\frac{\left( \reach- \rho \right)^d \:   \left(\cos \angle \Pi_1, T_x\M  \right)^{\min(d, N-d)}   }
{  (\reach+ \rho)^d  \cos \angle \Pi_2, T_x\M } , 
\frac{ (\reach+ \rho)^d  \cos \angle \Pi_1, T_x\M  } 
{\left( \reach- \rho \right)^d \:   \left(\cos \angle \Pi_2, T_x\M  \right)^{\min(d, N-d)}   }
\right]    
\end{equation*}

To sum up, we have proven:
\begin{lemma}\label{lemma:BoundsOnJacobianPhiPi1Pi2}
Let $\Pi_1,\Pi_2 \subseteq  \Rspace^N$, be two  $d$-dimensional affine spaces,
where, for $i=1,2$   $\Pi_i$ contains  the point $y_i \in \M^{\oplus \rho}$
with  $\pi_\M(y_i) = x \in\M$.
For  $i=1,2$, let  $U_i  \subseteq  \Pi_i \cap \Offset \M \rho$ be   an open subset of 
$\Pi_i$ such that $\pi_\M(U_1)= \pi_\M(U_2)$
and, for $i=1,2$,  $\forall z \in U_i, \angle \Pi_i,  T_{\pi_\M(z)}\M < \theta < \pi/3$.

Then, the Jacobian, of the map $\varphi_{1 \to 2} \defunder{=}  \left( \pi_{\M}\restrict{U_2}  \right)^{-1}   \circ  \pi_{\M}\restrict{U_1}$,
 taking as chart for $U_1$ and $U_2$ the coordinates associated to an orthonormal frame, is bounded by:
\begin{equation}\label{equation:BoundsOnJacobianPhiPi1Pi2}
\left| \det  D\varphi_{1 \to 2}(y)   \right| \in \left[ 
\frac{\left( \reach- \rho \right)^d \:   \left(\cos \theta \right)^{\min(d, N-d)}   }
{  (\reach+ \rho)^d } , 
\frac{ (\reach+ \rho)^d  } 
{\left( \reach- \rho \right)^d \:   \left(\cos \theta  \right)^{\min(d, N-d)}   }
\right]    
\end{equation}
\end{lemma}

\begin{remark}\label{remark::BoundsOnJacobianPhiPi1Pi2Asymptotic}
The bound \eqref{equation:BoundsOnJacobianPhiPi1Pi2} can be expressed as:
\begin{equation}\label{equation:BoundsOnJacobianPhiPi1Pi2Asymptotic}
\left| \det  D\varphi_{1 \to 2}(y)   \right| \in \left[ (1+J)^{-1}, 1+J \right]    
\end{equation}
with 
\[
J = \frac{ (\reach+ \rho)^d  } 
{\left( \reach- \rho \right)^d \:   \left(\cos \theta  \right)^{\min(d, N-d)}   } - 1
\]
where, as $\frac{\rho}{\reach} \rightarrow 0$, one has:
\begin{equation}\label{equation:BoundsOnJacobianPhiPi1Pi2AsymptoticExpressionOfJ}
J = \mathcal{O} \left( \frac{\rho}{\reach}  \right)
\end{equation}

\end{remark}
%
%

\clearpage
\section{Transfering orientation}
\label{appendix:transfering-orientation}

In this section, we start by recalling what it means for a manifold to
be orientable. Given a manifold with a prescribed orientation, we then
explain how to orient a simplex consistently with the manifold
(Definition \ref{definition:SimplexOrientationTransfert}). Finally, we
provide conditions under which the property for a simplex to be
consistently oriented with a manifold is preserved under projection
onto a nearby tangent space (Lemma
\ref{lemma:TangentSpaceProjectionIPreservesOrientation}).

\begin{definition}[Manifold orientation]\label{definition:ManifoldOrientation}
  An {\em orientation} of a $C^1$-manifold \M consists of an atlas
  $\left\{ \left( U_i \subseteq \M, \psi_i:U_i \to \Rspace^d \right)\right\}_{i \in I}$
  such that:
  \[
  U_i \cap U_j \neq  \emptyset \: \implies \: \forall m \in U_i \cap U_j, \: \det \Differential \left(\psi_j \circ \psi_i^{-1} \right) (\psi_i(m)) > 0.
  \]
  \M is said to be {\em orientable} if such an atlas exists.
\end{definition}

\begin{definition}[Chart consistent with the manifold orientation]\label{definition:ChartOrientation}
If $U$ is an open subset of the oriented manifold \M, a local chart $\psi : U \rightarrow \Rspace^d$ is called {\em consistent} 
with the orientation of \M  defined by the atlas $\left\{\left(U_i, \psi_i \right)\right\}_{i \in I}$
if the following implication holds:
\[
U_i \cap U \neq  \emptyset \: \implies \:  \forall m \in U_i \cap U, \: \det \Differential \left(\psi \circ \psi_i^{-1} \right) (\psi_i(m)) >0.
\]
\end{definition}

Let us associate to each oriented non-degenerate abstract $d$-simplex
$\sigma = \left[ u_0, u_1, \ldots, u_d \right] \subseteq \Rspace^N$
the linear map $\psi_\sigma$ defined by
\begin{equation*}\label{equation:FrameAssociatedToSimplex}
  \psi_\sigma :
  \begin{cases}
    \Aff \sigma &\longrightarrow \quad \Rspace^d\\
    x &\longmapsto \quad (t^1, \ldots, t^d),
  \end{cases}
\end{equation*}
where $t^1, \ldots, t^d$ are the coordinates of $x$ in the frame
$(u_0, u_1 - u_0, \ldots, u_d - u_0)$, that is, real numbers such that
$x = u_0 + \sum_{i=1}^d t^i (u_i - u_0)$. We note that the orientation
of $\sigma$ induces a natural orientation of the affine subspace
$\Aff{\sigma}$ defined by the atlas formed of a unique chart $\{(\Aff \sigma, \psi_\sigma)\}$. Let
us now define what we mean for a simplex to have an orientation
consistent with that of a manifold.

\begin{definition}[Consistent orientation between a simplex and a manifold]
  \label{definition:SimplexOrientationTransfert}
Let $\M \subseteq \Rspace^N$ be an orientable manifold, whose reach
is greater than $\rho >0$. Let $\sigma 
\subseteq \Rspace^\Dim$ be an oriented non-degenerate $d$-simplex such that
$\Conv{\sigma} \subseteq \M^{\oplus \rho}$, and suppose that
\[
\max_{x \in
  \Conv\sigma}\angle (\Aff{\sigma}, \Tangent{\pi_\M(x)}{\M}) <
\frac{\pi}{2}.
\]
We say that the orientation of $\sigma$ is {\em
  consistent} with a given orientation of $\M$ if there exists a point
$x \in \Conv \sigma$ and an open neighborhood $U_x$ of $x$ in
$\Aff\sigma$ such that the chart
\begin{equation}\label{equation:chart-simplex-manifold}
\psi_\sigma \circ \left( \pi_\M\restrict{U_x} \right)^{-1}: \pi_{\M}(U_x) \subseteq \M \rightarrow  \psi_\sigma(U_x) \subseteq \Rspace^d
\end{equation}
 is consistent with the orientation  of $\M$.
 \end{definition}

\begin{remark}
By Lemma \ref{lemma:ProjectionFromNonOrthogonalFlatIsAChart}, the
chart defined in \eqref{equation:chart-simplex-manifold} is indeed a
valid $C^1$-chart for $\M$.
\end{remark}

\begin{remark}
It is easy to see that the above definition does not depend upon the
choice of $x$ inside $\Conv \sigma$. Indeed, suppose that the
orientation of \M is defined by the atlas $\left\{ (U_i, \psi_i)
\right\}_{i \in I}$. Consider a point $m$ such that $m \in U_i \cap
\pi_\M(U_x) \neq \emptyset$ for some $i \in I$. The orientation
consistency between $\sigma$ and $\M$ is determined by the sign of
\[
\det \Differential \left(\psi_\sigma \circ \left( \pi_\M\restrict{U_x} \right)^{-1} \circ (\psi_i)^{-1} \right)(\psi_i(m)).
\]
Thanks to Lemma \ref{lemma:ProjectionFromNonOrthogonalFlatIsAChart},
the above determinant does not vanish as $x$ moves in $\Conv{\sigma}$,
and, since $\Conv{\sigma}$ is connected and the determinant is
continuous, its sign is constant over $\Conv{\sigma}$.
\end{remark}

Given two $d$-dimensional vector subspaces $V$ and $V'$ of $\Rspace^N$
such that $\angle(V,V') < \frac{\pi}{2}$,
Lemma~\ref{lemma:RotationBetweenVectorSubspaces} implicitly gives an
expression of the matrix of the orthogonal projection
$\pi_{V'}\restrict{V}$ on $V'$ restricted to $V$ in an orthonormal
basis.  In particular, it says that the matrix of $\pi_{V'}\restrict{V}$
has a full rank and therefore is invertible. Observe that if $B$ is a
basis of $V$ that defines an orientation of $V$, then
$\pi_{V'}\restrict{V} (B)$ is a basis of $V'$ that induces an
orientation of $V'$. This allows us to transfer the orientation of $V$
to $V'$ as follows.

\begin{definition}[Consistent orientation between vector spaces]
  Let $V$ and $V'$ be two $d$-dimensional vector subspaces of
  $\Rspace^N$.  Let $B$ ({\em resp.} $B'$) be a basis of $V$ ({\em
    resp.} $V'$). We say that $(V',B')$ {\em has a consistent
    orientation through projection with} $(V,B)$ if the
  determinant of the matrix of $\pi_{V'}\restrict{V}$ with respect to the bases $B$ and $B'$ is positive.
\end{definition}

As seen in the proof of Lemma
\ref{lemma:RotationBetweenVectorSubspaces}, the map
$\pi_V\restrict{V'} \circ \pi_{V'}\restrict{V} : V \rightarrow V$
admits a positive definite matrix in an orthonormal basis and
therefore preserves the orientation. It follows that the relation
``has a consistent orientation through projection with'' is symmetric
and could have been defined either by saying (as in the definition
above) that the determinant of the matrix of $\pi_{V'}\restrict{V}$
with respect to bases $B$ and $B'$ is positive, or by saying that the
determinant of the matrix of $\pi_{V}\restrict{V'}$ with respect to
bases $B'$ and $B$ is positive.

However, the relation ``has a consistent orientation through
projection with'' is not transitive in general.  For example, if
$\Delta_1, \Delta_2,\Delta_3$ are three $1$-dimensional vector spaces
in $\Rspace^2$, each oriented by basis $B_1= \{(1,0)\}$,  basis
$B_2=\{(\cos \frac{\pi}{3}, \sin \frac{\pi}{3})\}$, and  basis $B_3
=\{(\cos \frac{2\pi}{3}, \sin \frac{2\pi}{3})\}$ respectively, then
$(\Delta_1, B_1)$ has a consistent orientation through projection with
$(\Delta_2, B_2)$, as do $(\Delta_2, B_2)$ with $(\Delta_3, B_3)$, but
not $(\Delta_1, B_1)$ with $(\Delta_3, B_3)$.  However, one has the
following lemma, useful for the proof of Lemma
\ref{lemma:TangentSpaceProjectionIPreservesOrientation}:

\begin{lemma}[Making consistent orientation between vector spaces transitive]
  \label{lemma:WhenOrientationPropagationIsTransitive}
Let $V$, $V_1$, $V_2$, and $V_3$ be $d$-dimensional vector subspaces of $\Rspace^N$
such that:
\[
\angle(V, V_i) < \frac{\pi}{4}, \quad \textrm{for $i \in \{1,2,3\}$}.
\]
Then, in any basis $B_1$ of $V_1$ the matrix of the linear map from $V_1$ to $V_1$:
\begin{equation*}\label{equation:CyclicProjectionBetween3Spaces}
\pi_{V_1}\restrict{V_3} \circ \pi_{V_3}\restrict{V_2} \circ \pi_{V_2}\restrict{V_1}
\end{equation*}
has a positive determinant. Equivalently, assuming $B_1$, $B_2$, and
$B_3$ are bases defining orientation of $V_1$, $V_2$, and $V_3$
respectively, we have the following: If $(V_1,B_1)$ has a consistent
orientation through projection with $(V_2,B_2)$ and $(V_2,B_2)$ has a
consistent orientation through projection with $(V_3,B_3)$, then so does
$(V_1,B_1)$ with $(V_3,B_3)$.
\end{lemma}

\begin{proof}
The lemma holds trivially when $V_1=V_2=V_3=V$. In this case, the
matrix associated to $\pi_{V_1}\restrict{V_3} \circ
\pi_{V_3}\restrict{V_2} \circ \pi_{V_2}\restrict{V_1}$ is the identity
and its determinant is $1$.  Using Lemma
\ref{lemma:RotationBetweenVectorSubspaces}, it is easy to build a
basis $B_i(t)$ of a $d$-dimensional vector space $V_i(t)$ for each $i
\in \{1,2,3\}$, parametrized by $t \in [0,1]$ and continuous as a
function of $t$, such that $B_i(0)$ is a basis of $V$, $B_i(1)$ a
basis of $V_i$, and $\angle( V, V_i(t)) < \frac{\pi}{4}$ for all $t\in
[0,1]$.  In this condition, the determinant of:
\begin{equation*}\label{equation:CyclicProjectionBetween3SpacesVarying}
\pi_{V_1(t)}\restrict{V_3(t)} \circ \pi_{V_3(t)}\restrict{V_2(t)} \circ \pi_{V_2(t)}\restrict{V_1(t)}
\end{equation*}
is a continuous function of $t$ which equals $1$ when $t=0$. Since for
any $i,j \in \{1,2,3\}$ one has $\angle( V_i(t), V_j(t)) <
\frac{\pi}{2}$, each projection has a full rank and therefore so does
their composition. Thus, the determinant does not vanish and must
remain positive for all $t\in [0,1]$.
\end{proof}

Consider a compact $C^2$ $d$-dimensional submanifold \M of the
Euclidean space $\Rspace^N$ whose reach is positive. For $m \in \M$,
write
\[
\M_m = \M \cap B\left(m,  \sin \left(\frac{\pi}{4}\right) \Reach{\M}\right)^\circ.
\]
Suppose that the orientation of the tangent space $\TanSpace_m\M$ is defined by
the chart $\psi: \TanSpace_m\M \to \Rspace^d$. By Lemma
\ref{lemma:TangentSpaceProjectionIsAMap}, the restriction of $\psi
\circ \pi_{\TanSpace_m\M}$ to $\M_m$ is a local chart of $\M$.  We say that
the orientation of the tangent space $\TanSpace_m\M$ is {\em consistent} with
the orientation of \M if the local chart $\psi \circ \pi_{\TanSpace_m\M}: \M_m \to \Rspace^d$ is
consistent with the orientation of \M.

\begin{lemma}[Projection on a tangent space gives the orientation]
  \label{lemma:TangentSpaceProjectionIPreservesOrientation}
Let $\M$ be a compact, orientable, $C^2$ $d$-submanifold of the
Euclidean space $\Rspace^N$ with a reach greater than $\rho>0$. Let
$m\in \M$ and assume that $T_m\M$ has orientation that is consistent with
that of \M. Let $\sigma = [v_0,\ldots,v_d]$ be an oriented,
non-degenerate $d$-simplex such that $\Conv{\sigma} \subseteq
\M^{\oplus \rho}$, $\pi_\M(\Conv{\sigma}) \subseteq B\left(m, \sin
(\frac{\pi}{4}) \Reach{\M}\right)$, and $\angle( \Aff{\sigma} , T_m\M)
< \frac{\pi}{4}$.  Then the orientation of $\sigma$ is consistent with
the orientation of $\M$ if and only if the orientation
of \[\pi_{T_m\M}(\sigma) =
[\pi_{T_m\M}(v_0),\ldots,\pi_{T_m\M}(v_d)]\] is consistent with the
orientation of $T_m\M$.
\end{lemma}

\begin{proof}
  For $x \in \Conv \sigma$, let us denote by $\tilde{x}=\pi_\M(x)$ the
  projection of $x$ on $\M$.  Since $x\in \Conv{\sigma}$, we have that
  $\tilde{x} \in \M \cap B\left(m, \sin (\frac{\pi}{4})
  \Reach{\M}\right)$ and applying Lemma
  \ref{lemma:TangentSpaceProjectionIsAMap}, one obtains that:
\begin{equation*}
\angle (T_m\M, T_{\tilde{x}}\M )< \frac{\pi}{4}.
\end{equation*}
Using the above inequality and our assumption that $\angle(
\Aff{\sigma} , T_m\M) < \frac{\pi}{4}$, we can apply Lemma
\ref{lemma:WhenOrientationPropagationIsTransitive} with vector spaces
$V = V_1 = T_m\M$, $V_2 = \Aff{\sigma}$, and $V_3 = T_{\tilde{x}}\M$.
It follows that we can choose orthonormal bases for $T_m\M$,
$\Aff{\sigma}$, and $T_{\tilde{x}}\M$ which have a consistent
orientation through projection. These bases define respective frames
of $\TanSpace_m\M$, $\Aff{\sigma}$, and $\TanSpace_{\tilde{x}}\M$
centered at $m$, $x$, and $\tilde{x}$, respectively.  These frames
define coordinate systems, or charts, for $\TanSpace_m\M$,
$\Aff{\sigma}$, and $\TanSpace_{\tilde{x}}\M$. We can always choose
those charts so that the induced orientation of $\TanSpace_m\M$ is
consistent with the orientation of \M.

Suppose that the orientation of $\pi_{T_m\M}(\sigma) =
[\pi_{T_m\M}(v_0),\ldots,\pi_{T_m\M}(v_d)]$ is consistent with the
orientation of $T_m\M$ and let us prove that the orientation of
$\sigma = [v_0,\ldots,v_d]$ is consistent with the orientation of
\M. By definition, we thus need to show that for an open neighborhood
$U_x$ of $x$ in $\Aff \sigma$, we have
\begin{equation}\label{equation:TauOrientedAsMBrut}
  \det \Differential \left( \psi_\sigma \circ \left( \pi_\M\restrict{U_x} \right)^{-1}
  \circ \left( \psi \circ \pi_{\TanSpace_m\M}\restrict{\M_m}\right)^{-1} \right)  \left( \psi \circ \pi_{\TanSpace_m\M}(\tilde{x}) \right) >0
\end{equation}
Without loss of generality, we may assume $\psi = \operatorname{Id}$,
identifying $\TanSpace_m\M$ with $\Rspace^d$.
Because of our choice of frames for $\TanSpace_m\M$, $\Aff{\sigma}$,
and $\TanSpace_{\tilde{x}}\M$, we obtain immediately that the
orientation of $\sigma$ is consistent with the orientation of $\Aff
\sigma$. Thus, $\det \Differential \psi_\sigma(y) > 0$ for all $y \in
\Aff \sigma$ and  \eqref{equation:TauOrientedAsMBrut} is
equivalent to
\begin{equation}\label{equation:TauOrientedAsM}
\det \Differential \left( \left( \pi_\M\restrict{U_x} \right)^{-1} \circ \left( \pi_{\TanSpace_m\M}\restrict{\M_m}\right)^{-1} \right)  \left( \pi_{\TanSpace_m\M}(\tilde{x}) \right) >0.
\end{equation}
which in turn is equivalent to
\begin{equation}\label{equation:TauOrientedAsMInverse}
\det \Differential \left( \pi_{\TanSpace_m\M}\restrict{\M_m}  \circ \pi_\M\restrict{U_x} \right) \left( x \right) >0.
\end{equation}

Since $\M$ and $\TanSpace_{\tilde{x}}$ are tangent at $\tilde{x}$, we obtain, by
using the chart $\pi_{T_{\tilde{x}}\M}\restrict{\M_{\tilde{x}}}$  for $\M$ in a neighborhood $\M_{\tilde{x}}$ of $\tilde{x}$ in $\M$:
\begin{equation}\label{equation:TDiffPiTYTildeFromMIs1}
\Differential \left(  \pi_{\TanSpace_{\tilde{x}}\M}\restrict{\M_{\tilde{x}}}  \right) \left( \tilde{x}\right)= \UnitMatrix.
\end{equation}

Since \eqref{eq:JacobianProjectionFromAffineSpace} assumes the
projection on $\TanSpace_{\tilde{x}}\M$ equipped with the orthonormal
basis $(e_1,\ldots,e_d)$ of $T_{\tilde{x}}\M$ as a chart for $\M$ in a
neighborhood $\M_{\tilde{x}}$ of $\tilde{x}$, it gives us an
expression of the derivative of
$\pi_{\TanSpace_{\tilde{x}}\M}\restrict{\M_{\tilde{x}}} \circ
\pi_\M\restrict{U_x}$:
\[
\Differential \left( \pi_{\TanSpace_{\tilde{x}}\M}\restrict{\M_{\tilde{x}}} \circ  \pi_\M\restrict{U_x}\right)  (x)  = \left( \UnitMatrix + \lambda  \mathrm{II}_v(\tilde{x})\right)^{-1} e_j^T  = \left( \UnitMatrix + \lambda  \mathrm{II}_v(\tilde{x})\right)^{-1}  \pi_{T_{\tilde{x}}\M}\restrict{U_x},
\]
where $v$ is a unit vector such that $x = \tilde{x} + \lambda v$,
 and $\pi_{T_{\tilde{x}}\M}\restrict{U_x}$ is the differential of $\pi_{\TanSpace_{\tilde{x}}\M}\restrict{U_x}$.
Then, using \eqref{equation:TDiffPiTYTildeFromMIs1}:
\begin{equation}\label{equation:ExpresionDProjOnMFromTau}
\Differential \left(  \pi_\M\restrict{U_x}\right)  (x)  = \left( \UnitMatrix + \lambda  \mathrm{II}_v(\tilde{x})\right)^{-1}  \pi_{T_{\tilde{x}}\M}\restrict{U_x}.
\end{equation}

%
Using \eqref{equation:TDiffPiTYTildeFromMIs1}, that is the tangent space at $\tilde{x}$ to either $\M_m$ or $\M_{\tilde x}$ coincides, with the same chart,
with the tangent space at   $\tilde{x}$ to $\TanSpace_{\tilde{x}} \M$, we get:

  
\begin{equation}\label{equation:ExpresionDProjOnTmFromM}
 \Differential \left(\pi_{\TanSpace_m\M}\restrict{\M_m} \right)  \left(\tilde{x}\right) 
=   \Differential \left(\pi_{\TanSpace_m\M}\restrict{\TanSpace_{\tilde{x}} \M}   \right)   \left(\tilde{x}\right)   =     \pi_{T_m\M}\restrict{T_{\tilde{x}}\M},
\end{equation}
since $ \pi_{T_m\M}\restrict{T_{\tilde{x}}\M} $ is the derivative of $\pi_{\TanSpace_m\M}\restrict{\TanSpace_{\tilde{x}} \M}$.
Combining \eqref{equation:ExpresionDProjOnMFromTau} and \eqref{equation:ExpresionDProjOnTmFromM}, we obtain that
\eqref{equation:TauOrientedAsMInverse} holds if and only of the
determinant of
\[
\pi_{T_m\M}\restrict{T_{\tilde{x}}\M} \, \left( \UnitMatrix + \lambda  \mathrm{II}_v(\tilde{x})\right)^{-1}  \, \pi_{T_{\tilde{x}}\M}\restrict{U_x}
\]
is positive. This holds since each of the matrices associated to the
three linear maps has a positive determinant: indeed, it holds for the
two projections by the choice of coordinate systems, and for $\left(
\UnitMatrix + \lambda \mathrm{II}_v(\tilde{x})\right)^{-1}$ because,
since $\lambda \leq \rho< \Reach{\M}$ and $\left\| \lambda
\mathrm{II}_v(\tilde{x}) \right\|_2 < 1$, the matrix $\UnitMatrix +
\lambda \mathrm{II}_v(\tilde{x})$ is symmetric and positive definite.

For the reverse implication, we establish the
contraposition. Precisely, letting $\sigma'$ be the simplex sharing
with $\sigma$ the same set of vertices but having the opposite
orientation, we show that if $\sigma'$ has orientation consistent
through projection on $T_m\M$, then $\sigma'$ has its orientation
consistent with \M. But, this statement is a consequence of the direct
proposition that we have just proved, in which we replace $\sigma$ with $\sigma'$.
\end{proof}

\clearpage
\section{Establishing practical conditions}
\label{appendix:practical-conditions}

In this section, we prove formally what is intuitively quite obvious.
Since, in Problem (\ref{problem:reconstruction}), we consider only
cycles $\gamma$ (that is, chains such that $\partial \gamma = 0$), the normalization condition ($\Load {m_0} \M K \gamma
=1$) is, in some sense, ``stable''. Indeed, a small change in $m_0$ or
in the projection direction, may only cross a $(d-1)$-simplex, which,
since $\partial \gamma = 0$ and the small angular change preserves the
orientation, will not change the resulting ``load''.

Therefore, the normalization condition can be replaced by a constraint which does not refer to \M anymore but
refers only to a rough approximation $\Pi$ of \M.

\begin{remark}\label{remark:SameSignInBoundary}
Recall that the  boundary of an oriented simplex $[v_0, \ldots, v_d]$ is defined as:
\[
\partial [v_0, \ldots, v_d] = \sum_{i=0}^d (-1)^i [v_0, \ldots, \hat{v_i},\ldots, v_d]
\]
where $ \hat{v_i}$ means that vertex $v_i$ is omitted,
so that the simplex $\tau = [v_0, \ldots, v_{d-1}]$ has the sign $(-1)^d$ in $\partial [v_0, \ldots, v_d] $.  
It follows that for any   $d$-coface  $\sigma' = \tau \cup \{ v_d'\}$ of $\tau$, $\tau$ appears in $\partial \sigma'$
with the same sign as in $\partial \sigma$ if and only if $\sigma'$ is oriented as $[v_0, \ldots, v_{d-1}, v_{d}']$.
\end{remark}

\Remember{Could the assumption in the lemma below be replaced by : the $d$-simplices of $K$ are $\rho$-small?}

\begin{lemma}\label{lemma:PushForwardOfCycleIsConstant}
Let $K$ be a simplicial complex with vertices in $\Rspace^\Dim$ such
that $| K | \subseteq \M^{\oplus \rho}$ and suppose that the
$d$-simplices of $K$ are non-degenerate and have a diameter upper
bounded by $\rho$
We also assume that for all $d$-simplices $\tau
\in K$ and all points $y\in \Conv{\tau}$, we have
\begin{equation}\label{equation:PracticalNormalizationLemmaConditionAngleSimplexManifold}
 \angle \Aff{\tau}, T_{\pi_\M(y)} < \frac{\pi}{2}.
\end{equation}
Choose an orientation for \M and assume that all $d$-simplices of $K$ inherit this orientation.
Then, for any $d$-cycle $\gamma$ in $K$, the map 
\[
 \sum_{\alpha} \gamma(\alpha) \Indicator{\pi_\M(\Conv{\alpha})} 
 \]
is constant almost everywhere.
\end{lemma}

\begin{proof}
Given a set $\Sigma$ of $d$-simplices in $K$ and a $d$-chain $\gamma$
in $K$, we denote by $\gamma\restrict{\Sigma}$ the restriction of
$\gamma$ to $\Sigma$. In other words, $\gamma\restrict{\Sigma}$ is the
chain that coincides with $\gamma$ on $\Sigma$ and is zero elsewhere.
We define the map $\pi_\M^\sharp \gamma : \M \rightarrow \Rspace$ as
the regularization of $m \mapsto \sum_{\alpha} \gamma(\alpha)
\Indicator{\pi_\M(\Conv{\alpha})}$:
\[
\left( \pi_\M^\sharp \gamma \right) (m) \defunder{=} \sum_{\alpha} \gamma(\alpha) \Indicator{\pi_\M(\Conv{\alpha})} (m). 
\]
The notation $\pi_\M^\sharp \gamma$  is justified by the fact that $\pi_\M^\sharp$ is a linear map from the set of chains in $K$
to the set of piecewise-constant real valued functions on $\M$ modulo equality almost everywhere.
The regularized version of it is:
\[
\left( \overline{\pi_\M^\sharp}  \gamma \right) (m) \defunder{=} \lim_{r \rightarrow 0^+} \frac{ \int_{m' \in  \M  \cap B(m,r)} \left( \pi_\M^\sharp \gamma \right) (m')  \:d\mu_\M    } 
{\int_{m' \in  \M  \cap B(m,r)}   \:d\mu_\M     },
\]
where, as usual, $B(m,r)$ designates the ambient ball with center $m$ and radius $r$, and $\mu_\M$ is the $d$-volume on $\M$.
This regularization will allow us to conclude the proof by exploiting the continuity of this  regularized function.

For a simplex $\sigma \in K$, we denote the set of $d$-simplices in
the star of $\sigma$ in $K$ by $\overline{ \Star{\sigma}{K}}$.  We
denote the $k$-skeleton of $K$ by $\Skel^k(K)$.  We start by proving three
claims:

\medskip\noindent {\bf  \underline{Claim 1:} $\pi_\M( \Skel^{d-1}(K))$ has $\mu_\M$-measure zero and $\M  \setminus \pi_\M( \Skel^{d-2}(K))$ is open and connected.}

 Since $\pi_\M$ is $C^1$, the image $\pi_\M(\sigma)$ of a
 $(d-1)$-simplex $\sigma$ has a zero $d$-Hausdorff measure.
 $\Skel^{d-1}(K)$ is a finite union of such images. $\pi_\M(
 \Skel^{d-2}(K))$, as the image of a compact set by $\pi_\M$, is a
 compact set and its complement is therefore open.  Similarly,
 $\pi_\M( \Skel^{d-2}(K))$ is a finite union of smooth compact
 $(d-2)$-submanifolds of $\M$. It follows that any intersection
 of $\pi_\M( \Skel^{d-2}(K))$ with a smooth curve is non
 generic. Therefore, for any two points $m_1,m_2 \in \M \setminus
 \pi_\M( \Skel^{d-2}(K))$, there must exist a smooth curve in $\M
 \setminus \pi_\M( \Skel^{d-2}(K))$ connecting $m_1$ to $m_2$. This
 proves Claim~1.

\medskip\noindent {\bf  \underline{Claim 2:}  If $\tau \in K$ is a $d$-simplex and $m \in \pi_\M (\tau^\circ)$, then
$\overline{\pi_\M^\sharp}  \gamma\restrict{\overline{ \Star{\tau}{K}}} = \pi_\M^\sharp  \gamma\restrict{\overline{ \Star{\tau}{K}}}$
 is constant in a neighborhood of $m$ in $\M$.}

Indeed, $\overline{ \Star{\tau}{K}} = \{ \tau \}$ and for any $m' \in
\pi_\M (\tau^\circ)$ one has $\overline{\pi_\M^\sharp}
\gamma\restrict{\overline{ \Star{\tau}{K}}} (m') = \gamma(\tau)$.  Due
to Lemma \ref{lemma:ProjectionFromNonOrthogonalFlatIsAChart},
$\pi_\M\restrict{\tau^\circ}$ is an open map and there is a
neighborhood $U_m$ of $m$ in \M such that $U_m \subseteq \pi_\M
(\tau^\circ)$. It follows that for any $m'\in U_m$,
$\overline{\pi_\M^\sharp} \gamma\restrict{\overline{ \Star{\tau}{K}}}
(m') = \gamma(\tau)$. This proves Claim 2.

\medskip

Note that Claim 2 does not use the assumption that $\partial \gamma =
0$.  Since, for any $d$-simplex $\tau$, the complement of
$\pi_\M(\tau)$ in $\M$ is open in $\M$,  a consequence of Claim
2 is that $\overline{\pi_\M^\sharp} \gamma$ and $\pi_\M^\sharp \gamma$
coincide on $\M \setminus \pi_\M( \Skel^{d-1}(K))$.

\medskip\noindent {\bf  \underline{Claim 3:}  If $\sigma \in K$ is a $(d-1)$-simplex and 
if $\partial \gamma = 0$ then for any $m \in \pi_\M (\sigma^\circ)$, 
$\overline{\pi_\M^\sharp}  \gamma\restrict{\overline{ \Star{\sigma}{K}}}$ is constant in a neighborhood of $m$ in $\M$.}

In order to prove the claim, pick one $d$-dimensional coface $\tau$ of $\sigma$ and
let $y \in \sigma^\circ$ be such that $\pi_\M(y)=m$.
For some neighborhood $U_{y}$ of $y$ in $\Aff{\tau}$,
 one can use
$\xi = \psi_\tau \circ \left( \pi_\M\restrict{U_{y}} \right)^{-1}$ as a chart for $\M$ in an open  neighborhood $U_m$ of $m$ in $\M$.
$U_m$ can be chosen small enough to be included in $\pi_\M (\Star{\sigma}{K}^\circ)$ and $U_{\xi(m)} = \xi(U_m)  \subseteq  \Rspace^{d}$ 
is an open neighborhood of $\xi(m)$.

Without loss of generality, assume that  $\sigma= [ v_0, \ldots, v_{d-1}]$ and  that the orientation 
of $\tau$ defined by  $\tau=  [ v_0, \ldots, v_{d-1}, v_{d}]$ is consistent with the given orientation of $\M$.

From the definition of $\xi$ and $\psi_\tau$ in Definitions \ref{definition:ChartOrientation}
 and  \ref{definition:SimplexOrientationTransfert},
one has:
\[
\xi( U_m \cap \pi_\M(\sigma) ) = U_{\xi(m)} \cap \Big( \Rspace^{d-1} \times \{0\} \Big).
\]
For a $d$-simplex $\tau' \in \overline{ \Star{\sigma}{K}}$, 
since the restriction of $\xi \circ \pi_\M$ to $\tau'$ is an homeomorphism, the boundary of $\xi \big( U_m \cap   \pi_\M (\tau') \big)$
in $U_{\xi(m)}$ has to coincide with $U_{\xi(m)} \cap \big( \Rspace^{d-1} \times \{0\} \big)$ and, $\xi \big( U_m \cap   \pi_\M (\tau') \big)$
 therefore equals either $U_{\xi(m)} \cap \big( \Rspace^{d-1} \times\Rspace^- \big)$ or $U_{\xi(m)} \cap \big( \Rspace^{d-1} \times\Rspace^+ \big)$.

Let us denote the vertex in $\tau' \setminus \sigma$  by $v_d'$ and consider the curve  $y'(t) = y+ t (v_d' - v_0)$, with $t \in [0, \lambda]$,  
$\lambda>0$  small enough to have $y' ( [0, \lambda] ) \subseteq   \tau'^\circ$
and $\pi_\M \big( y' ( [0, \lambda] ) \big) \subseteq  U_m$.
If $\tau'$ is oriented as $[v_0, \ldots, v_{d-1}, v_{d}']$ and  if both $\tau$ and $\tau'$ have been given orientations consistent with $\M$, we have
 (Definition \ref{definition:SimplexOrientationTransfert}):
\[
\det \Differential \Big(  \xi \circ \pi_\M \circ \psi_{\tau'}^{-1} \Big) (\psi_{\tau'} (y))  > 0,
\]
which is equivalent to:
\begin{align}\label{equation:ExpressionDeterminantImageBaseByXi}
\det \Bigg( &\Differential \Big( \xi \circ \pi_\M  \Big) (y) \cdot (v_1 - v_0) ,\ldots  \nonumber \\
&\ldots, \Differential \Big( \xi \circ \pi_\M   \Big) (y) \cdot  (v_{d-1} - v_0), 
 \: \Differential \Big( \xi \circ \pi_\M  \Big) (y)  \cdot (v_d' - v_0)   \Bigg)> 0.
\end{align}

But by definition of $\xi$ one has that for $i=1,\ldots, d-1$,  $\Differential \Big( \xi \circ \pi_\M  \Big) (y) \cdot (v_i - v_0) \in \Rspace^d$ equals
$(0,\ldots,0,1,0,\ldots,0)$ where the $1$ appears in position $i$ and thus \eqref{equation:ExpressionDeterminantImageBaseByXi} implies that
the last coordinate of $ \Differential \Big( \xi \circ \pi_\M  \Big) (y)  \cdot (v_d' - v_0)$ is positive, which implies that, for $t$ small enough,
$\xi \circ \pi_\M (y'(t)) \in  \Rspace^{d-1} \times\Rspace^+$ under our assumption for  $\tau'$ to be oriented as $[v_0, \ldots, v_{d-1}, v_{d}']$.
At the same time, we know from Remark \ref{remark:SameSignInBoundary} that  $\sigma$ appears
with the same sign  in $\partial \tau'$ as in $\partial \tau$ if and only if $\tau'$ is oriented as $[v_0, \ldots, v_{d-1}, v_{d}']$.
Since the orientations of $\tau$ and $\tau'$ are consistent with the chosen  orientation of $\M$,
$\xi \big( U_m \cap   \pi_\M (\tau') \big)$
is therefore included   in 
 $H^- = \Rspace^{d-1} \times \Rspace^-$ or  $H^+= \Rspace^{d-1} \times \Rspace^+$ respectively depending on 
whether $\sigma$ appears  negatively or positively  in  $\partial \tau$.

The set $\overline{\Star{\sigma}{K} }$ of $d$-cofaces of $\sigma$ can be decomposed in 
$\overline{\Star{\sigma}{K} }^-$ and $\overline{\Star{\sigma}{K} }^+$ such that:
\[
\partial \tau (\sigma) = -1 \iff \tau \in \overline{\Star{\sigma}{K} }^- \quad \text{and} \quad \partial \tau (\sigma) = 1 \iff \tau \in \overline{\Star{\sigma}{K} }^+.
\]
Then:
\begin{equation}\label{eq:ZeroBoundary1}
\partial \gamma = 0  \Rightarrow \left( \sum_{\tau \in \overline{\Star{\sigma}{K} }^+} \gamma(\tau) \right) 
- \left( \sum_{\tau \in \overline{\Star{\sigma}{K} }^-} \gamma(\tau) \right) = 0.
\end{equation}

If $m^- \in \xi_\sigma^{-1}(H^- \cap \xi_\sigma \circ \pi_\M(U_x) \setminus \Pi)$ then:

\begin{equation}\label{eq:EvalLeft}
\pi_\M^\sharp  \gamma\restrict{\overline{ \Star{\sigma}{K}}}(m^-) 
= \pi_\M^\sharp  \gamma\restrict{\overline{ \Star{\sigma}{K}}^-} (m^-) = \sum_{\tau \in \overline{\Star{\sigma}{K} }^-} \gamma(\tau),
\end{equation}
and if $m^+ \in \xi_\sigma^{-1}(H^+  \cap \xi_\sigma \circ \pi_\M(U_x) \setminus \Pi)$ then:
\begin{equation}\label{eq:EvalRight}
\pi_\M^\sharp  \gamma\restrict{\overline{ \Star{\sigma}{K}}} (m^+) 
= \pi_\M^\sharp  \gamma\restrict{\overline{ \Star{\sigma}{K}}^+} (m^+) = \sum_{\tau \in \overline{\Star{\sigma}{K} }^+} \gamma(\tau).
\end{equation}
Now \eqref{eq:ZeroBoundary1}, \eqref{eq:EvalLeft}, and \eqref{eq:EvalRight} give us:
\[
\pi_\M^\sharp  \gamma\restrict{\overline{ \Star{\sigma}{K}}}(m^+)   - \pi_\M^\sharp  \gamma\restrict{\overline{ \Star{\sigma}{K}}}(m^-)   =  0,
\]
which proves Claim 3.

\medskip

One has, for any $m\in \M$:
\begin{equation}\label{eq:psiAsSUmOverStars}
\pi_\M^\sharp  \gamma(m) = \sum_{m \in \pi_\M (\sigma^\circ)} \pi_\M^\sharp  \gamma\restrict{\overline{ \Star{\sigma}{K}}} (m).
\end{equation}
If $m \in \M \setminus \pi_\M( \Skel^{d-2}(K))$ then the simplices
$\sigma$ such that $m \in \pi_\M (\sigma^\circ)$ that contribute to
the sum \eqref{eq:psiAsSUmOverStars} are of dimension either $d$
either $d-1$.
It follows from Claims 1, 2, and 3, that if $m \in \M \setminus
\pi_\M( \Skel^{d-2}(K))$ then $\overline{\pi_\M^\sharp} \gamma$ is
constant in an open neighborhood of $m$. As a consequence, for $ m_0
\in \M \setminus \pi_\M( \Skel^{d-2}(K))$, the set
\[
\{ m \in \M
\setminus \pi_\M( \Skel^{d-2}(K)) \mid \overline{\pi_\M^\sharp}
\gamma(m) = \overline{\pi_\M^\sharp} \gamma (m_0) \}
\]
is open.  Since
$\M \setminus \pi_\M( \Skel^{d-2}(K))$ is connected by Claim 1, it
follows that $\psi_{\gamma, K }$ is constant on $\M \setminus
\pi_\M(\Skel^{d-2}(K))$. Since $\psi_{\gamma, K }$ and $m \mapsto
\sum_{\alpha \in K} \gamma(\alpha) \Indicator{\pi_\M(\Conv{\alpha})}$
coincide on $\M \setminus \pi_\M( \Skel^{d-1}(K))$, and since, due to
Claim 1, $\pi_\M( \Skel^{d-1}(K))$ has a zero Lebesgue measure, we have
shown that $\sum_{\alpha \in K} \gamma(\alpha)
\Indicator{\pi_\M(\Conv{\alpha})}$ is constant almost everywhere.
\end{proof}


The next lemma is useful to derive a realistic algorithm in Section
\ref{section:RealisticAlgorithm}. Roughly, it says that the
normalization constraint in Problem (\ref{problem:reconstruction}) can
be replaced by a constraint which does not refer to \M anymore but
refers only to a rough approximation $\Pi$ of \M. Such an
approximation can be derived from the mere knowledge of the point set $P$, as explained in
Section \ref{section:RealisticAlgorithm}.  For any $x
\in \Rspace^\Dim$ and any $r \geq 0$, recall that
\[
K[x,r] = \{ \sigma \in K \mid \Conv \sigma \cap B(x,r) \neq \emptyset \}
\]
and note that $K[x,r]$ is not necessarily a simplicial complex.



\begin{lemma}[Practical normalization lemma]\label{lemma:PracticalNormalizationByApproximateTangentSpace}
  Suppose $0 \leq \rho \leq \frac{\reach}{25}$. Let $K$ be a simplicial
  complex with vertices in $\Rspace^\Dim$ such that $| K | \subseteq
  \M^{\oplus \rho}$ and suppose that the $d$-simplices of $K$ are
  non-degenerate and have a diameter upper bounded by $\rho$. Suppose furthermore
   that for all $d$-simplices $\tau \in K$ and all points $y\in
  \Conv{\tau}$, we have
  \begin{equation}\label{equation:PracticalNormalizationLemmaConditionAngleSimplexManifold_2}
    \angle (\Aff{\tau}, \Tangent {\pi_\M(y)} \M) < \frac{\pi}{2}
  \end{equation}
  and that every $d$-simplex in $K$ inherits the orientation of the
  manifold \M.  Consider a $d$-dimensional affine space $\Pi$ passing through a point $x \in \M^{\oplus \rho}$
  such that
  \begin{equation}\label{equation:PracticalNormalizationLemmaConditionAngleSimplexPlane}
    \angle (\Pi, \Tangent {\pi_\M(x)} \M) < \frac{\pi}{8}.
  \end{equation}
  Assume that $\Pi$ is oriented consistently with \M.
  Let $K' = K[x,4\rho]$ and suppose that the following conditions hold:
  \begin{enumerate}[label=\styleitem{(\roman*)}]
  \item \label{eq:GenericWhenTIsZeroNoDMinusOneSimplexProject}
    $\forall \beta \in K'^{[d-1]}$,  $x \notin \pi_\Pi (\Conv{\beta})$;
  \item  \label{eq:AngleUpperBoundKPrimeWithPi0}
    $\forall \alpha \in K'^{[d]}$, $\angle \Aff{\alpha} , T_{\pi_\M(x)}\M < \frac{\pi}{4}$;
  \item \label{equation:PracticalOrientation}
    $\forall \alpha \in K'^{[d]}$, $\alpha$ inherits its orientation from $\Pi$.  
  \end{enumerate}
  Then, for all $d$-cycles $\gamma$ in $K$ and all points $m \in \M \setminus \pi_\M \left( \left| K^{[d-1]} \right| \right)$, we have
  \begin{equation}\label{equation:PracticalNormalizationByApproximateTangentSpace}
    \sum_{\alpha \in K^{[d]}} \gamma(\alpha)
    \Indicator{\pi_\M(\Conv{\alpha})} (m) ~~=~~ \sum_{\alpha \in K'^{[d]}
    } \gamma(\alpha) \Indicator{\pi_\Pi(\Conv{\alpha})} (x).
  \end{equation}
\end{lemma}

\begin{proof}


Let $\Pi_0$ be defined as the unique affine space passing through $x$
and parallel to $T_{\pi_\M(x)}\M$. We distinguish two cases depending
on whether $\Pi=\Pi_0$ or not.

\medskip\noindent {\bf Case $\Pi=\Pi_0$.}
 We first claim that:
\begin{equation}\label{equation:ProjectionOnMCoincidesWithProjectinOnPi0}
\left(  {\pi_\M}_{|_{\M^{\oplus \rho}}} \right)^{-1} \big(\{\pi_\M(x)\}\big) =  \pi_{\Pi_0}^{-1} \big(\{x\}\big) \cap \M^{\oplus \rho} \cap B(x, 5 \rho).
\end{equation}
Indeed, the left hand member is equal to $\NorSpace_{\pi_{\M}(x)}\M \cap B(\pi_{\M}(x), \rho)$,
 where $\NorSpace_{\pi_{\M}(x)}\M = \{\pi_{\M}(x)\} + N_{\pi_{\M}(x)}\M$
 is the $(N-d)$-dimensional affine subspace subspace of $\Rspace^N$ orthogonal to $\M$ at $\pi_{\M}(x)$.  Notice that $\NorSpace_{\pi_{\M}(x)}\M$ 
 contains $x$ and coincides with $\pi_{\Pi_0}^{-1} \big(\{x\}\big)$ --- the affine space through $x$ orthogonal to $\Pi_0$. 
 Since $B(\pi_{\M}(x), \rho) \subseteq  B(x, 5 \rho)$, the left hand member is included in the right hand member.
 To get the reverse inclusion we only need to show that:
 \[
 \pi_{\Pi_0}^{-1} \big(\{x\}\big) \cap \M^{\oplus \rho} \cap B(x, 5 \rho) \: \subseteq  \: B(\pi_{\M}(x), \rho).
\]
 
 For that observe that for  $y \in \pi_{\Pi_0}^{-1} \big(\{x\}\big) \cap \M^{\oplus \rho} \cap B(x, 5 \rho)$ one has  
 $y \in  B(  \pi_{\M}(x) , 6\rho)$ and, since $6 \rho < \reach$ and $y \in \NorSpace_{\pi_{\M}(x)}\M$, 
 one has   $\pi_{\M}(y) =  \pi_{\M}(x)$ and $d(y, \M) = d(y, \pi_{\M}(x))$. Thus, $y \in \M^{\oplus \rho}$ implies $y \in B(\pi_{\M}(x), \rho)$.
 Equality \eqref{equation:ProjectionOnMCoincidesWithProjectinOnPi0} is proved.

Let  $y\in \Conv{\alpha} \subseteq   \M^{\oplus \rho}$ and assume that $x= \pi_{\Pi_0} (y)$.
If $\Conv{\alpha} \subseteq  B(x, 5 \rho)$, then $y \in   \pi_{\Pi_0}^{-1} \big(\{x\}\big) \cap \M^{\oplus \rho} \cap B(x, 5 \rho) $  and
 \eqref{equation:ProjectionOnMCoincidesWithProjectinOnPi0} gives $\pi_\M(y) = \pi_\M(x)$. Then
\[
\Big( \Conv{\alpha}  \subseteq  B(x, 5 \rho)  \quad \textrm{and} \quad x \in \pi_{\Pi_0} (\Conv{\alpha}) \Big) \iff  \pi_{\M}(x) \in \pi_\M (\Conv{\alpha}).
\]
If we assume the diameter of simplices to be upper bounded by $\rho$, 
then $\Conv{\alpha}  \cap B(x, 4 \rho) \neq \emptyset$ implies $\Conv{\alpha}  \subseteq  B(x, 5 \rho)$, and we get:
\[
\Big( \Conv{\alpha}  \cap B(x, 4 \rho) \neq \emptyset \quad \textrm{and} \quad x \in \pi_{\Pi_0} (\Conv{\alpha}) \Big) \iff \pi_{\M}(x) \in \pi_\M (\Conv{\alpha}).
\]
It results that  for any $\alpha \in K$, $ \pi_{\M}(x) \in \pi_\M (\Conv{\alpha})$ implies $\alpha \in K'$, and:
 \begin{equation}\label{equation:EqualityProjectionIndex}
\sum_{\alpha \in K^{[d]}} \gamma(\alpha) \Indicator{\pi_\M(\Conv{\alpha})}  ( \pi_\M(x))
= \sum_{\alpha \in K'^{[d]} } \gamma(\alpha) \Indicator{\pi_{\Pi_0}(\Conv{\alpha})} (x).
\end{equation}

We  assume the following generic condition:
\begin{equation}\label{equation:GenericConditon_1_ForPracticalNormaliation}
\pi_\M(x) \in\M \setminus \pi_\M \left( \left| K^{[d-1]} \right| \right).
\end{equation}
The condition is generic because, if it does not hold, a sufficiently small $C^2$-perturbation $\M'$ of $\M$ 
would satisfy it and still meet all the conditions of the lemma.
Assuming the generic condition to hold, one can see using
\eqref{equation:EqualityProjectionIndex} that, in the particular case
where $\Pi = \Pi_0$, the lemma is just another formulation of Lemma
\ref{lemma:PushForwardOfCycleIsConstant}.

\medskip \noindent {\bf Case $\Pi \neq \Pi_0$.}  From now on, we
assume that $\Pi \neq \Pi_0$.  Consider a differentiable path
$\Gamma : [0,1] \rightarrow \left( \Rspace^N \right)^{d+1}$, with
$\Gamma(t) \defunder{=} (x, \mathcal{B}( t))$ in the space of
orthonormal $d$-dimensional frames in $\Rspace^N$. Precisely, $\Gamma(t) = (x,
\mathcal{B}( t))$ is the orthonormal frame of a $d$-dimensional affine
space $\Pi(t)$ containing $x$, with $\Pi(0) = \Pi_0$ and $\Pi(1)=\Pi$,
and $\mathcal{B}(t)$ is an orthonormal basis of the vector space
associated to $\Pi(t)$. Lemma \ref{lemma:RotationBetweenVectorSubspaces} allows us to give an explicit formulation for $\mathcal{B}$.
Since $\Pi \neq \Pi_0$,  $\theta = \angle (\Pi_0, \Pi)$ satisfies $0 < \theta <  \frac{\pi}{8}$, and both affine spaces contain $x$.
Consider two $d$-dimensional vector spaces $V$ and $V'$ such that $\Pi_0 = x + V$ and $\Pi = x + V'$, respectively.
Applying Lemma \ref{lemma:RotationBetweenVectorSubspaces} to $V$ and $V'$ and borrowing its notation, we define the parametrized family
of orthonormal bases  $\mathcal{B}(t)$ for $t\in [0,1]$ as:
\begin{align}
\mathcal{B}(t) &\defunder{=}  \left(\mathcal{B}_1(t)  , \ldots, \, \mathcal{B}_d(t) \right) \nonumber \\
&= \left(u_1(t) , \ldots, \, u_{d'}(t), w_1,\ldots, \,w_{d-d'} \right), \nonumber \\
& \textrm{where} \quad u_k(t) = \cos (t \theta_k) v_k + \sin (t \theta_k) v_k'.  \label{equation:DefinitionBasisFamilyB}
\end{align}
For any $t\in [0,1]$, $\Gamma(t) = \left( x, \mathcal{B}(t) \right)$ is an orthonormal frame of $\Pi(t)$.
We want to follow the evolution of the function $\phi:[0,1] \rightarrow \Rspace$ defined as:
\[
\phi (t) \defunder{=} \sum_{\alpha \in K'^{[d]} } \gamma(\alpha) \Indicator{\pi_{\Pi(t)} (\Conv{\alpha})} (x),
\]
or its regularization $\hat{\phi}$, defined similarly as in the proof of Lemma \ref{lemma:PushForwardOfCycleIsConstant}:
\[
\hat{\phi} (t) \defunder{=}  \lim_{h\rightarrow 0} \frac{1}{\min(1,t+h) - \max(0, t - h) } \int_{\max(0, t - h)}^{\min(1,t+h} \phi(s) ds.
\]
The proof then consists of showing that $\hat{\phi} (t)$ remains
constant along the path $t \mapsto \Pi(t)$ and thus, since
Equation
  \ref{equation:PracticalNormalizationByApproximateTangentSpace} is
  satisfied for $\Pi_0 =\Pi(0)$ by
\eqref{equation:EqualityProjectionIndex}, it will extend to $\Pi=
\Pi(1)$. 

The  family of bases $\mathcal{B}(t) = (u_1(t) , \ldots, \, u_{d'}(t), w_1,\ldots, \,w_{d-d'} )$, parametrized by $t\in[0,1]$,
induces a smooth map $\psi : [0,1] \times \Rspace^N \rightarrow \Rspace^d$ where,
 for $y \in \Rspace^N$, the components of $\psi(t, y) \in \Rspace^d$ are the coordinates of $\pi_{\Pi(t)}(y)-x$ in the basis $\mathcal{B}(t)$:
 \[
\psi (t, y)_k \defunder{=} \langle \mathcal{B}_k(t) , \, \pi_{\Pi(t)}(y)-x \rangle =  \langle \mathcal{B}_k(t) , \, y-x \rangle.
\]
With this definition one has:
\begin{equation}\label{eq:ZeroOfPsiIffProjIsX}
\psi (t, y) = 0 \: \iff \: \pi_{\Pi(t)}(y) = x.
\end{equation}

If $\sigma \in K'^{[d-2]}$, the set $\psi \left( [0,1], \Conv{\sigma}\right)$ is included in a compact smooth $(d-1)$-manifold with boundary
in $\Rspace^d$. Since it corresponds to the finite union of complements of sets of codimension $1$,
the condition:
\begin{equation}\label{equation:GenericCondition_2_ForPracticalNormalisation}
0 \notin  \psi  \left([0,1], \: \bigcup_{\sigma \in K'^{[d-2]}} \Conv{\sigma}\right)
\end{equation}
is generic. We  now make the assumption that this generic condition holds, since, if it does not, it
can be satisfied after an arbitrarily small perturbation of $t \mapsto \Gamma(t)$.
One can easily check using \eqref{equation:DefinitionBasisFamilyB} that for any $t\in [0,1]$, $\angle \Pi(t), \Pi_0 \leq \angle \Pi, \Pi_0$,
and even if a small  perturbation is required for ensuring the generic condition,
we can assume that, from \eqref{equation:PracticalNormalizationLemmaConditionAngleSimplexPlane}:
\begin{equation}\label{eq:IUpperBoundAngleOnPiOfT}
\forall t\in [0,1], \angle \Pi(t), T_{\pi_\M(x)}\M = \angle \Pi(t), \Pi_0 < \frac{\pi}{8}.
\end{equation}
We will need the following claim:
\begin{equation}\label{eq:InclusionProjectionOnXInKAreNearXPiOfX}
\forall t \in [0,1], \quad    \left( \pi_{\Pi(t)} \right)^{-1} (x)  \cap   \M^{\oplus \rho}  \cap B(\pi_\M(x) , 6 \rho)  \: \subseteq  \:  B(\pi_\M(x) , 3 \rho).
\end{equation}
Indeed, consider $y \in  \left( \pi_{\Pi(t)} \right)^{-1} (x)  \cap   \M^{\oplus \rho}  \cap B(\pi_\M(x) , 6 \rho)$.
Then $\pi_\M(y) \in B(\pi_\M(x) , 7 \rho)$, and using Lemma \ref{lemma:distanceToTangent}, one has:
\[
\pi_\M(y)  \in \left(T_{\pi_\M(x)}\M \right)^{\oplus \frac{(7 \rho)^2}{2 \reach} } \subseteq   \left(T_{\pi_\M(x)}\M \right)^{\oplus \rho},
\]
which, since $d(y, \pi_\M(y)) < \rho$, implies:
\begin{equation}\label{eq:InclusionYInOffsetTangentSpace}
y \in \left(T_{\pi_\M(x)}\M \right)^{\oplus 2 \rho}.
\end{equation}
Denote the orthogonal projection onto $\NorSpace_{\pi_\M(x)}\M$
by $\pi_{\NorSpace_{\pi_\M(x)}\M}$. Then, since $x
\in \NorSpace_{\pi_\M(x)}\M \cap \left(T_{\pi_\M(x)}\M \right)^{\oplus
  \rho}$, one has with \eqref{eq:InclusionYInOffsetTangentSpace}:
\[
\left\|  \pi_{\NorSpace_{\pi_\M(x)}\M} (y) - \pi_{\NorSpace_{\pi_\M(x)}\M}(x) \right\| \leq 2 \rho + \rho =  3 \rho.
\]
From \eqref{eq:IUpperBoundAngleOnPiOfT}, one has that $\angle (y-x),  \NorSpace_{\pi_\M(x)}\M < \frac{\pi}{8}$, which gives:
\[
d( y, \NorSpace_{\pi_\M(x)}\M) \leq 3 \rho \tan \frac{\pi}{8} < \frac{3}{2} \rho.
\]
Applying \eqref{eq:InclusionYInOffsetTangentSpace} again, we obtain
\[
d(y, \pi_\M(x)) \leq \sqrt{ d(y, T_{\pi_\M(x)}\M  )^2 +  d( y, \NorSpace_{\pi_\M(x)}\M)^2} < \sqrt{4 + \frac{9}{4} } \rho  < 3 \rho,
\]
which proves Equation \eqref{eq:InclusionProjectionOnXInKAreNearXPiOfX}.
Since $B(x, 5 \rho) \subseteq  B(\pi_\M(x) , 6 \rho)$ and  $B(\pi_\M(x) , 3 \rho) \subseteq  B(x, 4 \rho)$,
Equation \eqref{eq:InclusionProjectionOnXInKAreNearXPiOfX} gives us:
\begin{equation}\label{eq:InclusionProjectionOnXInKAreNearX}
\forall t \in [0,1], \:    \left( \pi_{\Pi(t)} \right)^{-1} (x)  \cap   \M^{\oplus \rho}  \cap B(x , 5 \rho)  \: \subseteq  \: B(x, 4 \rho).
\end{equation}
Since $K'$ is not a simplical complex, it is convenient to introduce
the smallest simplicial complex $\overline{K'}$ containing $K'$, in
other words the set of all faces of simplices in $K'$.  In
particular, $\overline{K'}$ contains all $(d-1)$-faces of the
$d$-simplices in $K'$.  One has $\left|\overline{K'}\right| \subseteq
B(x , 5 \rho)$ and, using \eqref{eq:ZeroOfPsiIffProjIsX},
\eqref{eq:InclusionProjectionOnXInKAreNearX} yields:
\begin{equation}\label{eq:InOverLineKPrimeAndProjectOnXThenInKPrime}
 \sigma \in \overline{K'}^{[d-1]}, \quad \textrm{and}\quad   0 \in \psi([0,1],  \Conv{\sigma})  \: \Rightarrow \: \sigma \in K'^{[d-1]}.
\end{equation}
As shown below, the changes in $t \mapsto \hat{\phi} (t)$ may happen only when $\psi(t, \Conv{\sigma}) =0$,
where $\sigma$ is a $(d-1)$-face of some $d$-simplex $\alpha \in K'^{[d]}$, in other words $\sigma \in \overline{K'}^{[d-1]}$.
We need  \eqref{eq:InOverLineKPrimeAndProjectOnXThenInKPrime} to ensure that every such $\sigma$
in fact belongs to $K'^{[d-1]}$.

Let $\sigma \in K'^{[d-1]}$ be such that $0 \in \psi([0,1],
\Conv{\sigma})$, and denote the set of $d$-cofaces
of $\sigma$ in $K'$ (respectively in $K$) by  $\overline{\Star{\sigma}{K'} }$
(respectively $\overline{\Star{\sigma}{K} }$).  Note that, if $\tau$ is a
$d$-coface of $\sigma$ in $K$, and since $\Conv{\sigma} \subseteq
\Conv{\tau}$, one has the implication $\Conv{\sigma} \cap B(x, 4 \rho) \neq \emptyset
\Rightarrow \Conv{\tau} \cap B(x, 4 \rho) \neq \emptyset$. Thus,
$\tau \in K'$ and we have:
\begin{equation}\label{eq:SameStarInKAndKPrime}
\sigma \in K'^{[d-1]}   \implies \overline{\Star{\sigma}{K'} } = \overline{\Star{\sigma}{K} }.
\end{equation}

\begin{remark}\label{remark:CycleConditionCanBeUsedOnKPrime}
  Thanks to \eqref{eq:SameStarInKAndKPrime},   the cycle condition $(\partial \gamma) (\sigma) = 0$ is inherited on each $(d-1)$-simplex $\sigma \in K'^{[d-1]}$
  by the restriction of $\gamma$ to $K'$. This is not true for the  $(d-1)$-simplices of $\overline{K'}^{[d-1]} \setminus K'^{[d-1]}$.
  Thanks to \eqref{eq:InOverLineKPrimeAndProjectOnXThenInKPrime} we only have to consider how  $\phi(t)$ evolves as $t$ continuously increases from $0$ to $1$ on
  $(d-1)$-simplices in $K'^{[d-1]}$ and benefit from the cycle condition.
\end{remark}

From Condition \ref{eq:AngleUpperBoundKPrimeWithPi0} and
\eqref{eq:IUpperBoundAngleOnPiOfT} we get that:
\begin{equation}\label{eq:AngleUppeBoundDSimpexWIthPi0}
\forall t \in [0,1], \tau \in K'^{[d]} \implies  \angle \Aff(\tau), \Pi(t) < \frac{3\pi}{8}.
\end{equation}
As a consequence of Lemma \ref{lemma:RotationBetweenVectorSubspaces}, the restriction of $\pi_{\Pi(t)}$
to $\Aff{\sigma}$ is an affine bijection and in particular an homeomorphism. It sends the boundary of a $d$-simplex $\Conv{\alpha}$
onto the boundary of the image of $\Conv{\alpha}$: $\pi_{\Pi(t)}(\partial \Conv{\alpha})= \partial \pi_{\Pi(t)}(\Conv{\alpha})$.
 
The fact that $\psi$ is uniformly continuous (in fact $C^\infty$ with a compact domain) also in $t$, means that in particular:
\[
\forall \epsilon>0, \exists \eta >0, \forall y \in \Aff{\sigma}, \forall  t, t^\star \in [0,1],
\quad \quad \left| t^\star - t \right| < \eta
\implies
\left\| \psi(t^\star, y) - \psi(t, y) \right\| < \epsilon.
\]
If $0$ is not on the boundary of $ \psi\left( t, \: \Conv{\alpha} \right)$  there is a $\epsilon>0$
such that either $B(0,\epsilon) \subseteq   \psi  \left( t, \: \Conv{\alpha} \right)^\circ$ or
 $B(0,\epsilon) \cap  \psi \left( t, \: \Conv{\alpha} \right) = \emptyset$. It follows that:
 \begin{align} 
 0 \in  \psi \left( t, \Conv{\alpha} \right)^\circ \quad &\implies \quad \exists \eta>0, 
    \: 0 \in \bigcap_{t^\star \in [t-\eta, t+ \eta]} \psi  \left( t^\star,\: \Conv{\alpha} \right), \label{equation:CaseZeroInsideImageOfAlpha}\\
 0 \notin  \psi \left( t, \Conv{\alpha} \right) \quad &\implies \quad   \exists \eta>0,
   \:  0 \notin  \psi  \left( [t-\eta, t+ \eta] ,\: \Conv{\alpha} \right). \label{equation:CaseZeroOutsideImageOfAlpha}
 \end{align} 
We are now ready to track the evolution of $t\mapsto \hat{\phi} (t)$. For
that we consider two cases. First we consider the values of $t$ such that $x \notin \pi_{\Pi(t)} \left(\overline{K'}^{[d-1]} \right)$, or in other words 
$0 \notin \psi \left(t, \: \left| \overline{K'}^{[d-1]}  \right| \right)$. 
Let $\alpha \in K'^{[d]}$. Since $0$ is not on the boundary of $ \psi\left( t, \: \Conv{\alpha} \right)$,
 one of the two cases \eqref{equation:CaseZeroInsideImageOfAlpha} or \eqref{equation:CaseZeroOutsideImageOfAlpha}
 must occur. This implies that, for some $\eta>0$,  $t \mapsto \phi(t)$ 
  is constant on $ [t-\eta, t+ \eta]$, which in turn implies that in this case $\phi$  and $\hat{\phi}$ coincide:
\[
0 \notin \psi \left(t, \: \left| \overline{K'}^{[d-1]}  \right|  \right) \Rightarrow  \phi(t) = \hat{\phi}(t).
\]
We  now consider the second case, namely when $t$ is such that $x \in
\pi_{\Pi(t)} \left( \left| \overline{K'}^{[d-1]} \right| \right)$.
According to \eqref{eq:InOverLineKPrimeAndProjectOnXThenInKPrime}, if
$x \in \pi_{\Pi(t)} \left( \Conv{\sigma} \right)$ for some $\sigma \in
\overline{K'}^{[d-1]}$, then $\sigma \in K'^{[d-1]}$.  Therefore
\eqref{eq:SameStarInKAndKPrime} and Remark
\ref{remark:CycleConditionCanBeUsedOnKPrime} applies.
 
We are interested in the possible change of value of $\phi(t^\star)$
when $t^\star$ belongs to a neighborhood of $t$.  Generically, if
$x \in \pi_{\Pi(t)} \left( \left| K'^{[d-1]} \right| \right)$, there
is a unique $\sigma \in K'^{[d-1]}$ such that $x \in \pi_{\Pi(t)}
\left( \Conv{\sigma} \right)$. However, we do not need this generic
condition, since $\phi(t^\star)$ can be expressed as the following sum
for any $t^\star \in [0,1]$.  We consider the $d$-simplices $\alpha$
separately, depending on whether the inverse image of $0$ by $\psi(t,.)$
is (1) interior: $0 \in \psi(t, \Conv{\alpha} )^\circ$ (the first sum
below), (2) on the boundary: $0 \in \partial (\psi(t, \Conv{\alpha}
))$ (the second sum below), or (3) the complement: $0 \in \psi(t,
\Conv{\alpha} )^c$ (the third sum below):
\begin{align*}
\forall t^\star \in [0,1], \: \phi(t^\star) &= \sum_{\substack{\alpha \in K'^{[d]} \\ 0 \in \psi(t^\star, \Conv{\alpha} )}} \gamma(\alpha) \nonumber   \\
&= \sum_{\substack{\alpha \in K'^{[d]} \\ 0 \in \psi(t, \Conv{\alpha})^\circ  \\ 0 \in \psi(t^\star, \Conv{\alpha} )}} \gamma(\alpha)  
\quad + \sum_{\substack{\alpha \in K'^{[d]} \\ 0 \in \partial\left(\psi(t, \Conv{\alpha}) \right) \\ 0 \in \psi(t^\star, \Conv{\alpha} )}} \gamma(\alpha)  
\quad + \sum_{\substack{\alpha \in K'^{[d]} \\ 0 \in \psi(t, \Conv{\alpha})^c  \\ 0 \in \psi(t^\star, \Conv{\alpha} )}} \gamma(\alpha).
\end{align*}
Thanks to  \eqref{equation:CaseZeroInsideImageOfAlpha} and \eqref{equation:CaseZeroOutsideImageOfAlpha} there is a $\eta>0$ such that,
for any $t^\star \in [\max(0, t - \eta) , \min(1, t+ \eta) ]$,
$0 \in \psi(t^\star, \Conv{\alpha} )$ always holds in the first sum and never occurs in the third one. Then:
\begin{align}
\forall t^\star &\in  [\max(0, t - \eta) , \min(1, t+ \eta) ],\nonumber \\
\phi(t^\star)&= \sum_{\substack{\alpha \in K'^{[d]} \\ 0 \in \psi(t, \Conv{\alpha})^\circ }} \gamma(\alpha)  
\quad + \sum_{\substack{\alpha \in K'^{[d]} \\ 0 \in \partial\left(\psi(t, \Conv{\alpha}) \right) \\ 0 \in \psi(t^\star, \Conv{\alpha} )}} \gamma(\alpha) \nonumber \\
&= \sum_{\substack{\alpha \in K'^{[d]} \\ 0 \in \psi(t, \Conv{\alpha} )^\circ}} \gamma(\alpha) 
\quad + \sum_{\substack{\sigma \in K'^{[d-1]} \\ 0 \in \psi(t, \Conv{\sigma} )}} 
\sum_{\substack{\alpha \in  \overline{\Star{\sigma}{K'} }\\ 0 \in \psi(t^\star, \Conv{\alpha} ) }} \gamma(\alpha). \label{equation:phiAsSumOverStars}
\end{align}
In \eqref{equation:phiAsSumOverStars}, the $d$-simplices in the second sum have been regrouped by the stars of $(d-1)$-simplices.
Notice that no $d$-simplex can be counted twice as, 
under the generic condition  \eqref{equation:GenericCondition_2_ForPracticalNormalisation} 
and with the angle bound \eqref{eq:AngleUppeBoundDSimpexWIthPi0}, one has:
\begin{equation}\label{eq:StarsOfEdgesAboveXAreDisjoint}
\sigma_1 \neq \sigma_2 \quad \text{and} \quad 0 \in \psi(t, \Conv{\sigma_1}),  \quad \text{and} \quad 0 \in \psi(t, \Conv{\sigma_2} ) \Rightarrow 
 \overline{\Star{\sigma_1}{K'} } \cap  \overline{\Star{\sigma_2}{K'} } = \emptyset.
\end{equation}
The first sum does not depend on $t^\star \in [\max(0, t - \eta) , \min(1, t+ \eta) ]$ and therefore remains constant in this interval, as in the first case.
Thanks to \eqref{eq:StarsOfEdgesAboveXAreDisjoint}, there are several $(d-1)$-simplices $\sigma \in K'^{[d-1]}$ such that
$0 \in \psi(t, \Conv{\sigma} )$ in the second sum,  their stars are disjoint.
It is then enough to study the variation of:
\[
t^\star \mapsto \phi_{\sigma} (t^\star ) \defunder{=} \sum_{\substack{\alpha \in  \overline{\Star{\sigma}{K'} } \\ 0 \in \psi(t^\star, \Conv{\alpha} )  }} \gamma(\alpha) 
\]
for a single $\sigma \in K'^{[d-1]}$  such that $0 \in \psi(t, \Conv{\sigma})$, when $t^\star$ belongs to a neighborhood of $t \in [0,1]$.

According to Lemma \ref{lemma:TangentSpaceProjectionIsAMap}, the
projection $\pi_{\Pi(0)}$ is a chart of $\M$ with a domain $U_0 = \M
\cap B\left(\pi_\M(x), \sin (\pi/4) \Reach{\M}\right)$.  Since $\M$ is
orientable, we can orient $\M$ consistently with a given orientation
of $\pi_{\Pi(0)}$. The simplices in $\left| K'^{[d]} \right|\subseteq
\B(x, 5 \rho) $ are sent into $\M \cap \B(x, 6 \rho) \subseteq \M \cap
\B(\pi_\M(x), 7 \rho)$ by $\pi_\M$.  Since, as assumed in the lemma, one has $\rho < \frac{\reach}{25}$, we get
$\pi_\M \left( \left| K'^{[d]} \right| \right) \subseteq U_0$.  In
other words, $\psi(0,.)$ restricted to $U_0$ is a chart of $\M$
consistent with the orientation.

Let $\tau$ be a $d$-simplex in $K'$ oriented consistently with the orientation of $\M$, $y \in \tau$, and $U_y$ a neighborhood of $y$ in  $\Aff(\tau)$.
According to Lemma \ref{lemma:TangentSpaceProjectionIPreservesOrientation}, the projection $\psi(0, .) (\tau)$ onto $\Pi(0)$
is positively oriented with respect to the orientation of $\Pi(0)$.
Since, for $y \in \tau$, the map $t \mapsto \det D \left( \psi(t, .)\restrict{\Conv{\tau}} \right) (y)$
is continuous and does not vanish by \eqref{eq:AngleUppeBoundDSimpexWIthPi0}, one has 
\[
 \det D \left( \psi(0, .)\restrict{\Conv{\tau}} \right) (y) > 0 \Rightarrow  \forall t\in [0,1], \det D \left( \psi(t, .)\restrict{\Conv{\tau}} \right) (y) > 0.
 \]
We have proven that $ \psi(t, .)\restrict{\Conv{\tau}} $ preserves the
orientation of any $d$-simplex in $K'$, which means that the
projection $\psi(t, \tau) = [\psi(t, v_0), \ldots, \psi(t, v_d)]$ of
any $d$-simplex $\tau= [v_0,\ldots, v_d]$ oriented consistently with
$\M$ is positively oriented with respect to the orientation of
$\Pi(t)$.

Therefore, we can apply the same argument as for Claim 3 of the proof of Lemma \ref{lemma:PushForwardOfCycleIsConstant} to the $d$-simplices in the
star of $\sigma$: if $x \in \pi_{\Pi(t)} \left( \Conv{\sigma} \right)$, in other words if for some $y\in \Conv{\sigma}^\circ$ one has $\psi(t,y) = 0$,
and if $H^-$ and $H^+$ are the two half spaces in $\Rspace^d$ bounded by the hyperplane spanned by $\psi(t, \sigma)$, 
then  $\alpha \in \overline{\Star{\sigma}{K'} }$ appears positively
(resp. negatively) in  the boundary of a $d$-coface $\alpha \in  \overline{\Star{\sigma}{K'}}$ if $\psi(t, \alpha) \subseteq  H^+$ (resp. $\psi(t, \alpha) \subseteq  H^-$).

It follows that, for a point $z$ in a neighborhood of $V_0$ of $0$, 
\[
z \mapsto \sum_{\substack{\alpha \in  \overline{\Star{\sigma}{K'} } \\ z \in \psi(t^\star, \Conv{\alpha} )  }} \gamma(\alpha) 
\]
has the same value in $V_0 \cap (H^-)^\circ$ and in  $V_0 \cap (H^+)^\circ$.
We consider the $C^1$ map $F : [0,1] \times \Conv{\sigma}^\circ \rightarrow [0,1] \times \Rspace^d$ defined by:
\[
F(t^\star, y^\star) \defunder{=} \left( t^\star, \psi(t^\star, y^\star) \right).
\]
We note that, in particular, $F(t,y)=(t,0)$.

By the Thom Transversality Theorem \cite[Chapter 2]{guillemin2010differential}
 the map $F$ is generically transversal to the manifold  $[0,1] \times \{0\}$, 
which implies that  $F$ is regular at $(t,y)$, i.e., its derivative at $(t,y)$ has rank $d$ and 
the vector $(1, 0) \in  \Rspace \times \Rspace^d$ does not belong to the image of the derivative of $F$ when $F(t^\star, y^\star) \in (0,1) \times \{0\}$.
Since the transversality property on $F$ is generic, if it does not hold, it will after an arbitrary small perturbation
 of $F$.


It follows that the image of $F$ in a neighborhood of $F(t,y)$  is a smooth hyper-surface whose tangent space does not contain the vector
$(1, 0) \in  \Rspace \times \Rspace^d$ and therefore separates the two vectors $(t-\eta, 0)$ and $(t+\eta, 0)$ of $\Rspace \times \Rspace^d$ for $\eta>0$  small enough.
In particular, (since $K'^{[d-1]}$ is finite) the set of all $t'$ such that $\psi \left( t',  \left| K'^{[d-1]}  \right| \right)$ is made of isolated values.

Therefore, for some $\epsilon>0$ and $\alpha>0$, the complement of $B((t,0), \epsilon)^\circ \cap F([t-\eta, t+\eta], \Conv{\sigma}^\circ )$ in $B((t,0), \epsilon)^\circ$ 
has exactly two connected components, which are open and contain  
$( \{t\} \times (H^-)^\circ ) \cap B((t,0), \epsilon)$ and $( \{t\} \times (H^+)^\circ ) \cap B((t,0), \epsilon)$ respectively.

We know by \eqref{equation:CaseZeroInsideImageOfAlpha} and 
\eqref{equation:CaseZeroOutsideImageOfAlpha}  that the sum
\[
\sum_{\substack{\alpha \in  \overline{\Star{\sigma}{K'} } \\ 0 \in \psi(t^\star, \Conv{\alpha} )  }} \gamma(\alpha)
\]
is locally constant in each connected component.
Then, since it has same value in  $( \{t\} \times (H^-)^\circ ) \cap B((t,0), \epsilon)$ and $( \{t\} \times (H^+)^\circ ) \cap B((t,0), \epsilon)$,
it has same value in both connected components. We have proven that $t^\star \mapsto \phi_{\sigma} (t^\star )$
has the same value for $t^\star \in (t - \eta, t)$ and  $t^\star \in (t , t+ \eta)$, and therefore its regularization $\hat{\phi_{\sigma}}$
is locally constant. This ends the proof of the lemma.

\end{proof}

\clearpage
\section{Approximate tangent space computed by PCA}
\label{appendix:subsection:PCA}

\begin{lemma}\label{lemma:PerturbationAndLLL_pseudoTangent}
Let $0 < 16 \varepsilon \leq \rho < \frac{\reach}{10}$, and suppose that $P \subseteq \Offset \M \delta$ 
for some $0 \leq \delta < \frac{\rho^2}{4 \reach}$, 
$\M \subseteq \Offset P \varepsilon$ and $\Sep{P} > \eta \varepsilon$ for some $\eta >0$.
If, for any point $p\in P$, $c_p$ is the center of mass of $P \cap B(p, \rho)$ and $V_p$ the linear space spanned by the $n$ eigenvectors corresponding 
to the $n$ largest eigenvalues of the inertia tensor of $\big( P \cap B(p, \rho) \big) - c_p$, 
then one has:
\[
\angle V_p , T_{\pi_\M(p)}\M \,< \,\Xi _0( \eta, d) \frac{\rho}{\reach},
\]
where the function $\Xi_0$  is polynomial in $\eta$ and exponential in $d$.
\end{lemma}
\begin{proof} 
Consider a frame centered at $c_p$ with an orthonormal basis of $\Rspace^N$,
whose $d$ first vectors  $e_1,\ldots,e_d$ belong to $T_{\pi_{\M}(p)}\M$ and the $N-d$ last vectors  $e_{d+1},\ldots,e_N$ to the normal fiber $N_{\pi_{\M}(p)}\M$.
Consider the symmetric $N\times N$ normalized inertia tensor $\mathbf{A}$ of $ P \cap B(p, \rho) $ in this frame:
\[
\mathbf{A}_{ij} \defunder{=} \frac{1}{\sharp P \cap B(p, \rho)} \sum_{p \in P \cap B(p, \rho)} \langle v_i, p - c_p\rangle  \langle v_j, p - c_p\rangle.
\]
The symmetric matrix $A$ decomposes into $4$ blocs:
\begin{equation}
\mathbf{A}=
\begin{pmatrix}
\mathbf{A}_{TT} & \mathbf{A}_{TN} \\
\mathbf{A}_{TN}^t & \mathbf{A}_{NN}
\end{pmatrix},
\end{equation}
where  the tangental inertia $\mathbf{A}_{TT}$  is a $d\times d$ symmetric define positive matrix. Because of the sampling conditions, 
we claim\footnote{This claim has to be  detailed if one wants to provide an explicit expression of the quantity $\Xi _0( \eta, d)$.} 
that there is a constant $C_{TT}>0$ depending only on $\eta$ and $d$ such that the smallest eigenvalue of 
$\mathbf{A}_{TT}$ is at least $C_{TT} \rho^2$:
\begin{equation}\label{equation:LowerBoundOnATT}
\forall u \in \Rspace^d, \|u\| = 1 \Rightarrow \: u^t \, \mathbf{A}_{TT} \, u \geq C_{TT} \rho^2.
\end{equation}

Observe that, by Lemma \ref{lemma:distanceToTangent}, the points in $P
\cap B(p, \rho)$ are at a distance less than
$\frac{(\rho+\delta)^2}{2\reach}$ from the space $\pi_{\M}(p)+
T_{\pi_{\M}(p)}\M$, and therefore so is $c_p$. Thus, the points in $P
\cap B(p, \rho)$ are at a distance less than $2
\frac{(\rho+\delta)^2}{2\reach} \leq 2 \frac{\rho^2}{\reach} $
(assuming $(\rho+ \delta)^2 \leq 2 \rho^2$) from the space $c_p+
T_{\pi_{\M}(p)}\M$.  Then, there are constants $C_{TN}$ and $C_{NN}$
such that the operator norms of
$\mathbf{A}_{TN}$ and $\mathbf{A}_{NN}$  induced by the Euclidean vector norm are upper bounded by:
\begin{equation}\label{equation:LowerBoundOnATN}
\forall u,v \in \Rspace^d,  \|u\| =  \|v\| = 1  \Rightarrow  \: v^t \, \mathbf{A}_{TN} \, u \leq C_{TN}  \frac{ \rho^2}{\reach} \rho,
\end{equation}
and:
\begin{equation}\label{equation:LowerBoundOnANN}
\forall u \in \Rspace^d,  \|u\| = 1 \Rightarrow  \:  u^t \,\mathbf{A}_{NN} \, u  \leq C_{NN}   \frac{ \rho^2}{\reach}  \frac{ \rho^2}{\reach}.
\end{equation}
Let $v\in \Rspace^N$ be a unit eigenvector of $\mathbf{A}$ with an eigenvalue $\lambda$:
\begin{equation}\label{equation:EigneVectorOfA}
\mathbf{A}\, v = \lambda \, v.
\end{equation}
Define $T \defunder{=} \Rspace^d \times \{0\}^{N-d} \subseteq  \Rspace^N$ and $N \defunder{=}  \{0\}^d \times  \Rspace^{N-d} \subseteq  \Rspace^N$,
corresponding, in the space of coordinates, to $T_{\pi_{\M}(p)}\M$ and $N_{\pi_{\M}(p)}\M$ respectively.

Let $\theta$ be the angle between $v$ and $T$. There are unit vectors $v_T \in \Rspace^d$ and $v_N \in  \Rspace^{N-d}$
such that:
\[
v = ((\cos \theta) v_T, (\sin \theta) v_N)^t,
\] 
where for a matrix $u$, $u^t$ denotes the transpose of $u$.
 \eqref{equation:EigneVectorOfA} can be rewritten as:
 \begin{equation}\label{equation:EigneVectorOfA_2X2}
\begin{pmatrix}
\mathbf{A}_{TT} & \mathbf{A}_{TN} \\
\mathbf{A}_{TN}^t & \mathbf{A}_{NN}
\end{pmatrix}\,  
\begin{pmatrix}
(\cos \theta) v_T \\
 (\sin \theta) v_N
\end{pmatrix}  
= \lambda \,
\begin{pmatrix}
(\cos \theta) v_T \\
 (\sin \theta) v_N
\end{pmatrix}.
\end{equation}
Equivalently:
\begin{align}
 (\cos \theta) \mathbf{A}_{TT}  v_T &+  (\sin \theta) \mathbf{A}_{TN}   v_N = \lambda (\cos \theta) v_T, \label{equation:LambdaCosThetaV_T} \\
 (\cos \theta)\mathbf{A}_{TN}^t  v_T &+  (\sin \theta)  \mathbf{A}_{NN}  v_N = \lambda (\sin \theta) v_N.  \label{equation:LambdaSinThetaV_N} 
\end{align}
Multiplying the two equations on the left hand side by $(\sin
\theta)v_T^t$ and $ (\cos \theta) v_N^t$ respectively, we get:
 \begin{align*}
(\sin \theta) (\cos \theta)  v_T^t \mathbf{A}_{TT}  v_T +  (\sin \theta)^2  v_T^t  \mathbf{A}_{TN}   v_N &= \lambda  (\sin \theta) (\cos \theta),   \\
 (\cos \theta)^2 v_N^t \mathbf{A}_{TN}^t  v_T +   (\cos \theta) (\sin \theta) v_N^t \mathbf{A}_{NN}  v_N &= \lambda (\cos \theta)  (\sin \theta).
\end{align*}
Thus,
\[
(\sin \theta) (\cos \theta)  v_T^t \mathbf{A}_{TT}  v_T +  (\sin \theta)^2  v_T^t  \mathbf{A}_{TN}   v_N  
=  (\cos \theta)^2 v_N^t \mathbf{A}_{TN}^t  v_T +   (\cos \theta) (\sin \theta) v_N^t \mathbf{A}_{NN}  v_N.
\]
Using  \eqref{equation:LowerBoundOnATN} and \eqref{equation:LowerBoundOnANN}, 
we get:
\[
(\sin \theta) (\cos \theta)  v_T^t \mathbf{A}_{TT}  v_T  \leq 2  C_{TN}  \frac{ \rho^2}{\reach} \rho  +  C_{NN}   \frac{ \rho^2}{\reach}  \frac{ \rho^2}{\reach}.
\]
 Using \eqref{equation:LowerBoundOnATT}, we get:
 \[
(\sin \theta) (\cos \theta)   \leq 2    \frac{ C_{TN}}{C_{TT}}  \frac{ \rho}{\reach}   +  C_{NN}   \frac{ C_{NN}}{C_{TT}}   \frac{ \rho^2}{\reach^2}.
\]
Using $\sin 2\theta = 2 \sin \theta \cos \theta$, we get:
\begin{equation}\label{equation:UpperBoundOnSinTwoTheta}
\frac{1}{2} \sin 2 \theta \leq 2    \frac{ C_{TN}}{C_{TT}}  \frac{ \rho}{\reach}   +   C_{NN}   \frac{ C_{NN}}{C_{TT}}   \frac{ \rho^2}{\reach^2} = \mathcal{O}\left(  \frac{ \rho}{\reach} \right).
\end{equation}
Thus,
\[
\theta \in \left[0, t\right] \cup \left[\frac{\pi}{2} - t, \frac{\pi}{2}\right],
\]
with
\[
t = \frac{1}{2}  \arcsin 2 \left(2    \frac{ C_{TN}}{C_{TT}}  \frac{ \rho}{\reach}   
+   C_{NN}   \frac{ C_{NN}}{C_{TT}}   \frac{ \rho^2}{\reach^2} \right) = \mathcal{O}\left(  \frac{ \rho}{\reach} \right).
\]
This means that the eigenvectors of $\mathbf{A}$ form an angle less
than $\mathcal{O}\left( \frac{ \rho}{\reach} \right)$ with either
$T$ or $N$.  For the non-generic situation of a multiple
eigenvalue, one chooses arbitrarily the vectors of an orthogonal basis
of the corresponding eigenspace. Since no more than $d$ pairwise
orthogonal vectors can make a small angle with the $d$-dimensional
space $T$, and the same holds for the $(N-d)$-dimensional space $T$, we know
that $d$ eigenvectors form an angle $\mathcal{O}\left( \frac{
  \rho}{\reach} \right)$ with $T$ and the $(N-d)$ others, an angle $\mathcal{O} \left( \frac{\rho}{\reach} \right)$ with $N$.  Multiplying 
the left hand side of \eqref{equation:LambdaCosThetaV_T} by $v_T^t$, and
the left hand side of \eqref{equation:LambdaSinThetaV_N} by $v_N^t$, we get:
\begin{align*}
 (\cos \theta)  v_T^t \mathbf{A}_{TT}  v_T &+  (\sin \theta) v_T^t \mathbf{A}_{TN}   v_N = \lambda (\cos \theta),\\
 (\cos \theta)  v_N^t \mathbf{A}_{TN}^t  v_T &+  (\sin \theta) v_N^t \mathbf{A}_{NN}  v_N = \lambda (\sin \theta).
\end{align*}
When the angle between the eigenvector $v$ and the space $T$ is in $\mathcal{O}\left(  \frac{ \rho}{\reach} \right)$, 
then $|1 - \cos \theta| = \mathcal{O}\left( \left( \frac{ \rho}{\reach} \right)^2\right)$ 
and $|\sin \theta | = \mathcal{O}\left(  \frac{ \rho}{\reach} \right)$.
The first equation implies that $\lambda$ approximately equals
$v_T^t \mathbf{A}_{TT}  v_T  \geq  C_{TT} \rho^2$.
When the angle between the eigenvector $v$ and the space $N$ is in $\mathcal{O}\left(  \frac{ \rho}{\reach} \right)$, 
the second equation implies that $\lambda =\mathcal{O}\left( \left( \frac{ \rho}{\reach} \right)^2\right)$,
which is smaller than $C_{TT} \rho^2$ for $\frac{ \rho}{\reach}$ small enough. 

So far we have proven that the $d$ orthonormal eigenvectors $v_1,\ldots, v_d$ corresponding to the $d$ largest eigenvalues of $\mathbf{A}$
form an angle with $T_{\pi_{\M}(p)}\M$ that is upper bounded by $C \left(\frac{ \rho}{\reach} \right)$, for some constant $C$ that depends only upon $d$ and $\eta$.
For any unit vector  $u$, its angle with $T_{\pi_{\M}(p)}\M$ satisfies:
\[
\sin \angle u, T_{\pi_{\M}(p)}\M = \| u - \pi_{T_{\pi_{\M}(p)}\M} (u)\|.
\]
Any  unit vector  $u= \sum_{i=1,d} a_i v_i$  in the $d$-space spanned by $v_1,\ldots, v_d$ satisfies:
\begin{align*}
\sin \angle u, T_{\pi_{\M}(p)}\M &= \left\|  \sum_{i=1,d} a_i v_i - \pi_{T_{\pi_{\M}(p)}\M} \left( \sum_{i=1,d} a_i v_i\right) \right\| \\
& =    \left\|  \sum_{i=1,d} a_i  \left(v_i - \pi_{T_{\pi_{\M}(p)}\M} (  v_i)  \right)\right\| \\
& \leq \sum_{i=1,d} | a_i |  \left\| v_i - \pi_{T_{\pi_{\M}(p)}\M} (  v_i) \right\|  \\
& \leq  \sum_{i=1,d} | a_i | C \left(\frac{ \rho}{\reach} \right)   \leq \sqrt{d} \,C \left(\frac{ \rho}{\reach} \right),
\end{align*}
since $\sum_{i=1,d} a_i^2 = 1$, we conclude that $\sum_{i=1,d} |a_i | \leq \sqrt{d}$.
\end{proof}

\end{document}